%% file: cofinite.tex
\documentclass[a4paper,USenglish,cleveref,autoref,numberwithinsect]{lipics-v2021}
\input{head.tex}

\author{D\'aniel Marx}
{CISPA Helmholtz Center for Information Security, Saarbrücken, Germany}
{marx@cispa.de}
{https://orcid.org/0000-0002-5686-8314}
{}
\author{Govind S. Sankar}
{Duke University, Durham, NC, USA}
{govind.subash.sankar@duke.edu}
{https://orcid.org/0000-0002-7443-9599}
{}
\author{Philipp Schepper}
{CISPA Helmholtz Center for Information Security, Saarbrücken, Germany}
{philipp.schepper@cispa.de}
{https://orcid.org/0000-0002-5810-7949}
{Part of Saarbrücken Graduate School of Computer Science, Germany.}
\Copyright{D\'aniel Marx, Govind S.~Sankar, and Philipp Schepper}
\authorrunning{D.~Marx, G.\,S.~Sankar, P.~Schepper}


\funding{Research supported by the European Research Council (ERC) consolidator grant No.~725978 SYSTEMATICGRAPH.}

\relatedversiondetails[linktext={doi:10.4230/LIPIcs.IPEC.2022.22},cite={MarxSS22}]{Conference Version}{https://doi.org/10.4230/LIPIcs.IPEC.2022.22}

\hideLIPIcs
\nolinenumbers 

\newcommand{\AntiFactor}[1]{\ensuremath{#1}-\textsc{AntiFactor}\xspace}
\newcommand{\AntiFactorR}[1]{\ensuremath{#1}-\textsc{AntiFactor}$^{\mathcal R}$\xspace}
\newcommand{\MaxAntiFactor}[1]{\textsc{Max}-\AntiFactor{#1}}
\newcommand{\MaxAntiFactorR}[1]{\textsc{Max}-\AntiFactorR{#1}}
\newcommand{\CountAntiFactor}[1]{\#\AntiFactor{#1}}
\newcommand{\CountAntiFactorR}[1]{\#\AntiFactorR{#1}}

\newcommand{\CountMinAntiFactor}[1]{\#\textsc{Min-}\AntiFactor{#1}}

\newcommand{\AntiFactorSize}[1]{\textsc{AntiFactor}$_{#1}$\xspace}
\newcommand{\MaxAntiFactorSize}[1]{\textsc{Max}-\AntiFactorSize{#1}}
\newcommand{\CountAntiFactorSize}[1]{\#\AntiFactorSize{#1}}
\newcommand{\AntiFactorSizeR}[1]{\textsc{AntiFactor}$_{#1}^{\mathcal R}$\xspace}
\newcommand{\CountAntiFactorSizeR}[1]{\#\AntiFactorSizeR{#1}}

\newcommand{\ColHS}{\ensuremath{k}-ColHS\xspace}

\renewcommand{\from}{\colon}

\newextmathcommand{\Mat}{\mathcal M}
\newextmathcommand{\calI}{\mathcal I}
\newextmathcommand{\MatU}{\mathcal U}
\newextmathcommand{\RepSet}{\mathcal S}
\newextmathcommand{\Alg}{\mathfrak A}
\newextmathcommand{\Fam}{\mathcal F}
\newextmathcommand{\Ex}{X}
\newextmathcommand{\Exbar}{\overline{\Ex}}
\newextmathcommand{\ex}{x}
\newextmathcommand{\Y}{Y}
\newcommand{\bag}[1]{B_{#1}} 

\newcommand{\represents}[2][]{\mathrel{\subseteq_{#2\text{-rep}}^{#1}}}
\newcommand{\comp}[2][]{\mathrel{\sim_{#2}^{#1}}}
\newcommand{\Comp}{\mathcal{C}}

\DeclareMathOperator{\ind}{\textsf{ind}}
\DeclareMathOperator{\ParSol}{\textsf{ParSol}}
\newcommand{\wtnode}[1]{\ensuremath{w{[#1]}}}

\DeclareMathOperator{\Time}{\textsf{Time}}
\DeclareMathOperator{\Size}{\textsf{Size}}

\DeclareMathOperator{\LEFT}{\textsf{left}}
\DeclareMathOperator{\RIGHT}{\textsf{right}}

\newcommand{\totalDeg}{\Delta^{\!*}}

\newcommand{\Slns}{\operatorname{\#Sol}}

\title{
Anti-Factor Is FPT Parameterized by Treewidth and List Size
(But Counting Is Hard)
}
\titlerunning{Anti-Factor Is FPT Parameterized by Treewidth and List Size}

\EventEditors{Holger Dell and Jesper Nederlof}
\EventNoEds{2}
\EventLongTitle{17th International Symposium on Parameterized and Exact Computation (IPEC 2022)}
\EventShortTitle{IPEC 2022}
\EventAcronym{IPEC}
\EventYear{2022}
\EventDate{September 7--9, 2022}
\EventLocation{Potsdam, Germany}
\EventLogo{}
\SeriesVolume{249}
\ArticleNo{22}

\begin{document}
\maketitle

\begin{abstract}
  In the general \textsc{AntiFactor} problem,
	a graph $G$ and, for every vertex $v$ of $G$, a set
	$\Ex_v\subseteq \SetN$ of forbidden degrees is given.
	The task is to find a set $S$ of edges
	such that the degree of $v$ in $S$ is \emph{not} in the set $\Ex_v$.
	Standard techniques (dynamic programming plus fast convolution)
	can be used to show that if $M$ is the largest forbidden degree,
	then the problem can be solved in time $(M+2)^{\tw}\cdot n^{\O(1)}$
	if a tree decomposition of width $\tw$ is given.
	However, significantly faster algorithms are possible
	if the sets $\Ex_v$ are sparse:
	our main algorithmic result shows that
	if every vertex has at most $\ex$ forbidden degrees
	(we call this special case \AntiFactorSize{\ex}),
	then the problem can be solved in time
	$(\ex+1)^{\O(\tw)}\cdot n^{\O(1)}$.
	That is, \AntiFactorSize{\ex} is fixed-parameter tractable
	parameterized by treewidth $\tw$
	and the maximum number $\ex$ of excluded degrees.

  Our algorithm uses the technique of representative sets,
	which can be generalized to the optimization version,
	but (as expected) not to the counting version of the problem.
	In fact, we show that \CountAntiFactorSize{1} is already \sharpW{1}-hard
	parameterized by the width of the given decomposition.
	Moreover, we show that, unlike for the decision version,
	the standard dynamic programming algorithm
	is essentially optimal for the counting version.
	Formally, for a fixed nonempty set $\Ex$,
	we denote by \AntiFactor{\Ex} the special case
	where every vertex $v$ has the same set $\Ex_v=\Ex$ of forbidden degrees.
	We show the following lower bound for every fixed set $\Ex$:
	if there is an $\epsilon>0$ such that
	\CountAntiFactor{\Ex} can be solved in time
	$(\max \Ex+2-\epsilon)^{\tw}\cdot n^{\O(1)}$
	given a tree decomposition of width $\tw$,
	then the Counting Strong Exponential-Time Hypothesis (\#SETH) fails.
\end{abstract}

\newpage

\input{cofinite/intro}
\input{cofinite/algo}
\input{cofinite/repset}
\input{cofinite/half-induced}
\input{cofinite/dec}

\input{cofinite/opt}

\input{cofinite/count}

\input{cofinite/ecover}

\bibliography{bib}

\end{document}

%% file: head.tex
\usepackage{xifthen}

\usepackage{mathtools}
\usepackage{commath}
\usepackage{xspace}
\usepackage{bm}

\usepackage{xcolor}

\usepackage{tablefootnote}
\usepackage{comment}
\usepackage{environ}

\usepackage{tikz}
\usetikzlibrary{calc,positioning}

\newcommand{\executeiffilenewer}[3]{%
\ifnum\pdfstrcmp{\pdffilemoddate{#1}}%
{\pdffilemoddate{#2}}>0%
{\immediate\write18{#3}}\fi%
}

\graphicspath{{./img/svg/}}

\theoremstyle{plain}

\let\abs\relax
\DeclarePairedDelimiter\abs{\lvert}{\rvert}
\let\originalleft\left
\let\originalright\right
\def\left#1{\mathopen{}\originalleft#1}
\def\right#1{\originalright#1\mathclose{}}

\newcommand{\floor}[1]{\ensuremath{\lfloor #1 \rfloor}}

\newcommand{\ceil}[1]{\ensuremath{\lceil #1 \rceil}}

\newcommand{\GenFac}{\textsc{GenFac}\xspace}

\newcommand{\Factor}[1]{\ensuremath{#1}-\textsc{Factor}\xspace}
\newcommand{\BFactor}{\Factor{B}}
\newcommand{\MinBFactor}{\textsc{Min}-\BFactor}
\newcommand{\MaxBFactor}{\textsc{Max}-\BFactor}
\newcommand{\CountFactor}[1]{\#\Factor{#1}}
\newcommand{\CountBFactor}{\CountFactor{B}}

\newcommand{\BFactorRelation}{$B$-\textsc{Factor with Relations}\xspace}

\newcommand{\BFR}{$B$-\textsc{Factor}$^{\mathcal R}$\xspace}

\newcommand{\CountBFR}{\#\BFR}

\newcommand{\countMIS}{\#\textsc{MaxIndSet}\xspace}
\newcommand{\countECover}{\#\textsc{EdgeCover}\xspace}

\newcommand{\newextmathcommand}[2]{%
  \newcommand{#1}{{\xspace\ensuremath{#2}\xspace}}
}

\newcommand{\deff}{\coloneqq}

\newextmathcommand{\from}{\leftarrow}

\newcommand{\EQ}[1]{\ensuremath{\mathtt{EQ}_{#1}}}
\newcommand{\HWin}[2][]{\ensuremath{\mathtt{HW}_{\in #2}%
\ifthenelse{\equal{#1}{}}{}{^{(#1)}}}}
\newcommand{\HWeq}[2][]{\ensuremath{\mathtt{HW}_{= #2}%
\ifthenelse{\equal{#1}{}}{}{^{(#1)}}}}

\newcommand{\Hol}[1]{\ensuremath{\operatorname{Holant}(#1)}}

\newcommand{\IN}{\texttt{in}}
\newcommand{\OUT}{\texttt{out}}

\newcommand{\true}{\textsf{true}\xspace}
\newcommand{\false}{\textsf{false}\xspace}

\DeclareMathOperator{\poly}{poly}
\DeclareMathOperator{\hw}{hw}
\DeclareMathOperator{\maxgap}{max-gap}
\DeclareMathOperator{\cEC}{\#EC}

\newcommand{\dotcup}{\mathrel{\dot{\cup}}}

\newextmathcommand{\SetB}{\{0,1\}}
\newextmathcommand{\SetF}{\mathbb F}
\newextmathcommand{\SetN}{\mathbb N}
\newextmathcommand{\SetR}{\mathbb R}
\newextmathcommand{\SetZ}{\mathbb Z}

\newcommand{\sharpP}{\textup{\textsf{\#P}}\xspace}
\newcommand{\NP}{\textup{\textsf{NP}}\xspace}

\newcommand{\sharpW}[1]{\textup{\#\textsf{W}$[#1]$}\xspace}

\renewcommand{\O}{\mathcal O}

\newcommand{\Ostar}[2][n]{\ensuremath{#2 #1^{\mathcal{O}(1)} }}
\newcommand{\tw}{{\operatorname{tw}}}
\newcommand{\pw}{{\operatorname{pw}}}

\let\eps\epsilon

%% file: cofinite/intro.tex
\section{Introduction}

Matching problems and their generalizations form a well studied class of problems in combinatorial optimization and computer science \cite{MR88b:90087}.
A \emph{perfect matching} is a set $S$ of edges such that every vertex has degree exactly 1 in $S$;
finding a perfect matching is known to be polynomial-time solvable   \cite{Edmonds65,HopcroftK73,MicaliV80}.
In the \Factor{f} problem,
an integer $f(v)$ is given for each vertex $v$
and the task is to find a set of edges
where every vertex $v$ has degree exactly $f(v)$.
A simple transformation reduces \Factor{f} to finding a perfect matching.
Conversely, in \AntiFactor{f}
the task is to find a set $S$ of edges
where the degree of $v$ is \emph{not} $f(v)$ \cite{DBLP:journals/jct/Sebo93a}.

The problems above can be unified under the \textsc{General Factor} (\GenFac) problem
\cite{Cornuejols88,Lovasz72, MarxSS21},
where one is given a graph $G$ and an associated set of integers $B_v$ for
every vertex $v$ of $G$. The objective is to find a subgraph such that
every vertex $v$ has its degree in $B_v$.
Cornu{\'e}jols \cite{Cornuejols88} showed that the complexity of \GenFac
depends on the maximum gap of the sets $B_v$.
The maximum gap of a set $B$ (denoted by $\maxgap(B)$) is defined as the largest
contiguous sequence of integers not in $B$ but whose boundaries are
in $B$.
Cornu{\'e}jols \cite{Cornuejols88} showed that if $\maxgap(B_v)\leq1$,
then \GenFac is polynomial-time solvable.
In a sense, we can say that this case is the only one that is polynomial-time solvable.
Formally, for a fixed, finite set $B$ of integers,
\BFactor is the special case of \GenFac
where every vertex has the same set $B_v=B$ of allowed degrees.
It follows from a result of Dalmau and Ford \cite{DalmauF03}
that if $B$ is a fixed finite set such that $\maxgap(B)>1$,
then \BFactor is \NP-hard.

Given the hardness of \BFactor in general,
Marx et al.~\cite{MarxSS21} studied the complexity of the problem on bounded treewidth graphs.
Recall the long history of study on treewidth,
which is a measure for how ``tree-like'' a graph is,
\cite{Bodlaender88,Bodlaender97,BodlaenderK08}.
For a wide range of hard problems,
algorithms with running time of the form $f(k)\cdot n^{\O(1)}$ exist
if the input graph comes with a tree decomposition of width $k$.
In many cases even the best possible form of $f(k)$ in the running time is known
(under suitable complexity assumptions, such as
the Strong Exponential Time Hypothesis (SETH)
\cite{ImpagliazzoP01}).
Marx et al.~\cite{MarxSS21} use a combination of standard dynamic programming techniques
with fast subset convolution (cf.~\cite{Rooij20})
to give optimal (under SETH) $\Ostar{(\max B+1)^\tw}$ time algorithms
for the decision, optimization, and counting versions.
\begin{theorem}[Theorems~1.3--1.6 in \cite{MarxSS21}]
  \label{thm:factor-tw-results}
  Fix a finite, non-empty set $B\subseteq \SetN$.
  \begin{itemize}
  \item
  We can count in time $\Ostar{(\max B+1)^\tw}$
  the solutions of a certain size
  for a \BFactor instance
  if we are given a tree decomposition of width $\tw$.

  \item
  For any $\eps>0$,
  there is no $\Ostar{(\max B+1-\eps)^{\pw}}$ algorithm
  for the following problems,
  even if we are given a path decomposition of width $\pw$,
  unless SETH (resp.\ \#SETH) fails:
  \begin{itemize}
    \item
    \BFactor and \MinBFactor
    if $0 \notin B$ and $\maxgap(B) > 1$,
    \item
    \MaxBFactor
    if $\maxgap(B) > 1$,
    \item
    \CountBFactor
    if $B \neq \{ 0\}$.
  \end{itemize}
  \end{itemize}
\end{theorem}
We study the complementary problem of
\AntiFactor{\Ex} for finite sets $\Ex$ of excluded degrees.
\begin{definition}[\AntiFactor{\Ex}]
  \label{def:antifactor}
  Let $\ex\in\SetN$ be fixed.
  \AntiFactorSize{\ex} is
  the decision problem of finding for an undirected graph $G$
  where all vertices $v$ are assigned a finite set $\Ex_v\subseteq \SetN$
  with $\abs{\Ex_v} \le \ex$,
  a set $S\subseteq E(G)$ such that for all $v \in V$ we have $\deg_S(v) \notin \Ex_v$.

  For a fixed $\Ex \subseteq \SetN$ with $\abs{X}=x$,
  we define \AntiFactor{\Ex} as the restriction of \AntiFactorSize{x}
  to those graphs where all vertices are labeled with the same set $\Ex$.
\end{definition}

\begin{note*}
  \Factor{\Exbar}, the special case of
  \GenFac where every vertex has set $\Exbar$,
  precisely corresponds to \AntiFactor{\Ex}
  where we set $\Exbar \deff \SetN \setminus \Ex$.
\end{note*}
The decision and minimization versions are trivially solvable if $0\notin \Ex$
as the empty set is a valid solution.
Further, if $\Ex$ does not contain two consecutive numbers,
then $\overline \Ex$ has no gap of size at least two.
In this case, by results from Cornu{\'e}jols \cite{Cornuejols88}
and Dudycz and Paluch \cite{DudyczP18},
the decision, maximization and minimization
version of \Factor{\overline \Ex} are poly-time solvable.

\subparagraph{Our Results.}
One could expect that similar results can be obtained for \AntiFactor{\Ex} as for \Factor{B},
but this is very far from the truth and the exact complexity of \AntiFactor{\Ex} is much less clear.
In the \Factor{B} problem, a partial solution
(a set of edges that we intend to further extend to a solution)
can have degree at most $\max B$ at each vertex,
which is the main reason one needs $\Ostar{(\max B+1)^\tw}$ running time.
For \AntiFactor{\Ex}, a vertex can also have degree
larger than $\max \Ex$ in a (partial) solution,
but all degrees larger than $\max \Ex$ are
equivalent in some sense.
Therefore, the natural running time we expect is $\Ostar{(\max \Ex+2)^\tw}$.
We show that this running time can be achieved,
but requires some modification of the convolution
to handle the state ``degree more than $\max \Ex$.''
\begin{theorem}
  \label{thm:algo:paraByList}
  Let $\Ex \subseteq \SetN$ be finite and fixed.
  Given an \AntiFactor{\Ex} instance and its tree decomposition of width $\tw$.
  Then we can
  count the number of solutions of size exactly $s$ in time $\Ostar{(\max{\Ex}+2)^{\tw}}$
  for all $s$ simultaneously.
\end{theorem}
However, there are many cases where algorithms significantly faster than $\Ostar{(\max \Ex+2)^\tw}$ are possible.
At first, this may seem unlikely:
at each node of the tree decomposition,
the partial solutions can have up to $(\max \Ex+2)^{\tw+1}$ different equivalence classes%
\footnote{
Recall that in a graph with treewidth $\tw$, the largest bag has size $\tw+1$.
}
and it may seem necessary to find a partial solution for each of these classes.
Nevertheless, we show that the technique of \emph{representative sets}
can be used to achieve a running time
lower than the number of potential equivalence classes.
Representative sets were defined by Monien~\cite{Monien85}
for use in an FPT algorithm for \textsc{$k$-Path},
and subsequently found use in many different contexts,
including faster dynamic programming algorithms on tree decompositions
\cite{DBLP:journals/algorithmica/AgrawalJKS20,DBLP:journals/iandc/BodlaenderCKN15,DBLP:journals/algorithmica/BonnetBKM19,DBLP:journals/talg/EibenK20,FominLPS16,FominLPS17,KratschW20,DBLP:conf/soda/MarxW15,ShachnaiZ16}.
The main idea is that we do not need to find a partial solution for each equivalence class,
but it is sufficient to find a representative set of partial solutions
such that if there is a partial solution that is compatible with some extension,
then there is a partial solution in our set
that is also compatible with this extension.
Our main algorithmic result shows that if $\Ex$ is sparse,
then this representative set can be much smaller than $(\max \Ex+2)^{\tw+1}$,
yielding improved algorithms.
In particular, \AntiFactorSize{\ex} is FPT
parameterized by $\tw$ and $\ex$.

\begin{theorem}\label{thm:algo:main}
  One can decide in time $\Ostar{(\ex+1)^{\O(\tw)}}$
  whether there is a solution of a certain size for \AntiFactorSize{\ex}
  assuming a tree decomposition of width $\tw$ is given.
\end{theorem}
We note that \cref{thm:algo:main} clearly distinguishes \AntiFactor{\Ex}
from \BFactor.
By the known lower bounds from Marx et al.~\cite{MarxSS21}
(cf.~\cref{thm:factor-tw-results}),
a similar result for \BFactor is not possible.
In light of \cref{thm:algo:main},
it is also far from obvious to determine
the exact complexity of \AntiFactor{\Ex} for a fixed set $\Ex$.
The combinatorial properties of the set $\Ex$
influence the complexity of the problem in a subtle way
and new algorithmic techniques seem to be needed to fully exploit this.
Currently, we do not have a tight bound similar to \cref{thm:factor-tw-results} for every fixed $\Ex$.
Instead we propose a candidate for the combinatorial property that influences the complexity:
We define a bipartite compatibility graph for every set $\Ex$
and conjecture that the maximum size of a so-called \emph{half-induced matching} is the key property
to obtain a faster algorithm via representative sets.
  See \cref{conj:himImpliesUpperBoundOnRepSet} for a formal statement.

We use such half-induced matchings of large size
to show a lower bound for \AntiFactorSize{\ex} that, assuming SETH,
complements the algorithm in \cref{thm:algo:main} up to constant factors
in the exponent (see \cref{thm:dec:lbAntiFactorBySize}).
Moreover, if there is a half-induced matching of size $h$,
then, assuming SETH,
we show that
there is no $\Ostar{(h-\epsilon)^{\tw}}$ algorithm for \AntiFactor{\Ex}
for any $\epsilon>0$ (\cref{thm:dec:lbAntiFactor}).
Although, in this case
the representative set cannot be smaller than $(h-\epsilon)^{\tw+1}$
for any $\epsilon>0$
 (\cref{lem:half-induced-implies-repset})
we do not have matching upper bounds at this point.
There are two main reasons why it is difficult to obtain tight upper bounds:
\begin{itemize}
  \item \textbf{Representative set bounds.}
  In \cref{thm:algo:main}, the upper bound on the size of representative sets
  are based on earlier algebraic techniques
  \cite{FominLPS16,FominLPS17,KratschW20,ShachnaiZ16}.
  It is not clear how they can be extended to the combinatorial notion of half-induced matchings.

  \item \textbf{Join nodes.}
  Even if we have tight bounds on the size of representative sets
  there is an additional issue that can increase the running time.
  At join nodes of the tree decomposition,
  we need to compute from two representative sets a third one.
  Doing this operation in a naive way results in a running time
  that is at least the \emph{square} of the bound on the size of the representative set.
  If we want to have a running time that matches the size of the representative set,
  we need a more clever way of handling join nodes.
\end{itemize}
Representative sets of the form we study here could be relevant for other problems and tight bounds for such representative sets could be of fundamental importance. In particular, the notion of half-induced matchings could be a key property in other contexts as well.

\subparagraph*{Counting Problems.}
We also investigate the \CountAntiFactorSize{} problem,
where we need to count the total number of solutions satisfying the degree constraints.
The idea of representative sets is fundamentally incompatible with exact counting:
if we need to count every solution,
then we cannot ignore certain partial solutions
even if they can be always replaced by others.
Therefore, the algorithm of \cref{thm:algo:main} cannot be extended to the counting version.%
\footnote{
Counting the solutions approximately is a problem of independent interest.
}
In fact, we show that already \CountAntiFactorSize{1}
is unlikely to be FPT
by showing the following stronger statement for path decompositions.
\begin{theorem}\label{thm:count:lower-intro}
  There is a fixed constant $c$ such that \CountAntiFactorSize{1}
  cannot be solved in time
  $\O(n^{\pw-c})$ on graphs with $n$ vertices given a path decomposition
  of width $\pw$, unless \#SETH is false.
  Furthermore, \CountAntiFactorSize{1} is \sharpW{1}-hard
  parameterized by pathwidth.
\end{theorem}
Recall that \#SETH (cf.~\cite{CurticapeanM16,DellHMTW14})
is actually a weaker assumption than SETH.
Hence, the first result is stronger than a version based on SETH.
Moreover, the algorithm from \cref{thm:algo:paraByList}
is essentially optimal for \CountAntiFactor{\Ex}.
\begin{theorem}
	Let $\Ex\subseteq \SetN$ be a non-empty, finite and fixed set.
	For any constant $\epsilon>0$, there
	is no algorithm that can solve \CountAntiFactor{\Ex} in time
	$\Ostar{(\max \Ex+2-\epsilon)^{\pw}}$ given a graph along
	with a path decomposition of width $\pw$, unless \#SETH fails.
\end{theorem}

\subparagraph*{Organization.}
The paper is organized as follows.
\cref{sec:algorithms} presents
the algorithms of \cref{thm:algo:paraByList,thm:algo:main}.
\cref{sec:computing-rep-sets} introduces representative sets
and proves the results we need in our algorithms.
\cref{sec:half-induced-matchings} introduces half-induced matchings
and discusses some combinatorial properties related to representative sets.
	\cref{sec:lower:dec,sec:lower:opt,sec:lower:count} present
the lower bounds for the decision, optimization,
and counting versions, respectively.
One special case, the counting version of
\textsc{Edge Cover} is treated separately in \cref{sec:lower:edge-cover}.

%% file: cofinite/algo.tex
\section{Algorithms}\label{sec:algorithms}
In this section, we use without loss of generality
\emph{``nice''} tree decompositions that
have \emph{introduce edge nodes}
(see, e.g., \cite{ParameterizedAlgos} for formal definitions).
When given a node $t$ of a tree decomposition,
we denote by $\bag{t}$ the bag of $t$,
by $V_t$ the vertices introduced at the subtree rooted at $t$,
and by $E_t$ the edges introduced in the subtree rooted at $t$.

\subsection{Parameterizing by the Maximum Excluded Degree}
The proof of \cref{thm:algo:paraByList} follows mostly the ideas of the algorithm for \GenFac from Theorem~1.3 in \cite{MarxSS21}.
The main difference is that
$\Ex$ is finite and thus $\overline\Ex$ is cofinite.
Therefore,
every degree $d \ge \max \Ex+1$ is valid for a solution
and we
identify all such states $d$ with the state $\max\Ex+1$,
denoted by $\top$ in the following.
This modification can be handled quite easily for the leaf, introduce vertex, introduce edge and forget nodes.
As the convolution technique for the join nodes does not directly transfer,
we use another result from \cite{Rooij20}
which additionally involves zeta and M{\"o}bius transforms
to obtain the improved running time for the cofinite case.

To simplify notation, we set $U = [0,\max \Ex] \cup \{\top\}$ in the following.

\subparagraph*{Algorithm.}
The dynamic program fills a table $c$
such that, for all nodes $t$ of the tree decomposition,
all functions $f\from \bag{t} \to U$,
and all possible sizes $s \in [0,m]$,
it holds that $c[t, f, s]=a$
if and only if
there are $a$ partial solutions $S \subseteq E_t$ with $\abs{S}=s$
such that, for all $v \in V_t\setminus \bag{t}$, $\deg_S(v) \notin \Ex$
and, for all $v \in \bag{t}$, we have $\deg_S(v) = f(v)$ if $f(v) \neq \top$
and $\deg_S(v) > \max \Ex$ otherwise.

If the node $t$ is a leaf node, introduce vertex node, introduce edge node, or forget node,
then the values $c[t, f, s]$ can be easily computed from the values for $c[t', \cdot, \cdot]$
where $t'$ is the unique child of $t$ in the tree decomposition.
For each node $t$,
this computation takes time $\Ostar{(\max \Ex +2)^\tw}$ as $m \in \O(n^2)$.

It remains to compute the table entries for the join nodes $t$.
Unless mentioned otherwise,
$k$ denotes the size of the bag we consider.
\begin{lemma}
  \label{thm:algo:paraByList:joinNode}
  For a given join node $t$
  let $t_1$ and $t_2$ be the children.
  Assume we are given all values of $c[t_1,\cdot,\cdot]$ and $c[t_2,\cdot,\cdot]$.
  Then, for all (valid) $f$ and $s$, we can compute the value of $c[t,f,s]$
  in time $\Ostar{(\max\Ex+2)^\tw}$.
\end{lemma}
We define $\oplus$ coordinatewise by extending the standard addition
as follows:
\[
  \forall u,v \in U:
  u \oplus v \deff \begin{cases}
    \top & \text{if } u=\top \lor v=\top \lor u+v > \max \Ex \\
    u+v  & \text{else}
  \end{cases}
\]
Recall that, for the join node $t$ and given $f$ and $s$,
we want to compute the following:
\[
  c[t,f,s] \deff \sum_{\substack{f_1,f_2\from\bag{t} \to U \\ \text{s.t.~}f_1\oplus f_2=f} }
  \sum_{s_1+s_2=s} c[t_1,f_1,s_1]\cdot c[t_2,f_2,s_2]
	.
\]
\begin{remark*}
  The naive computation of $c[t,\cdot,\cdot]$ takes time
  \(
    \abs{U}^k \cdot m \cdot \abs{U}^{2k} \cdot m \cdot n^{\O(1)}
    = \abs{U}^{3k} \cdot n^{\O(1)}
  \).
\end{remark*}

  For the ease of notation,
  we refer to the values of the table $c[t_1, \cdot, \cdot]$ by a new table $c_1$
  and likewise for $c_2$ and $t_2$.
  We further omit $t$ from the index of $c$ such that this table only contains the new values we want to compute.

We follow the ideas of the proof of Theorem~2 in \cite{Rooij20}.
We impose a partial ordering $\preceq$ on $U$
where $u \preceq \top$, for all $u \in U$, and undefined otherwise.
We extend this to functions (and also vectors)
such that for $f_1,f_2\from \bag{t} \to U$ we have $f_1 \preceq f_2$
if and only if
$f_1(x) \preceq f_2(x)$, for all $x \in \bag{t}$.
Observe that
this partial ordering does \emph{not} extend the standard ordering $\le$.

\begin{definition}[Variant of Definition~6 in \cite{Rooij20}]
  We define the zeta transform of each table $t \in \{c,c_1,c_2\}$ as follows:
  \[
  \zeta(t)(f, s) \deff \sum_{g \preceq f} t[g, s]
	.
  \]
\end{definition}
For the computation of $c[f, s]$, we make use of the following lemma.
\begin{lemma}[Proposition~7 in \cite{Rooij20}]
  \label{lem:algo:paraByList:zetaIsEnough}
  Given the memoisation table $A(f,x)$ indexed by state colorings $f \in U^k$ over the label set $U$
  and some additional indices $x$ with domain $I$,
  \footnote{$I$ might have more than one dimension.}
  the zeta transform $\zeta(A)$ of $A$ based on the partial order $\preceq$
  can be computed in $\O(\abs{U}^k k\abs{I})$ arithmetic operations.
  Also, given $\zeta(A)$,
  $A$ can be reconstructed in $\O(\abs{U}^k k\abs{I})$ arithmetic operations.
\end{lemma}
Using this lemma, it is clear
that it suffices to compute only the zeta transform of the table $c$,
as the original values can be recovered afterwards.
This transformation introduces an additional overhead of $\Ostar{(\max \Ex+2)^k}$ in the running time
as we choose $I$ such that $\abs{I} \in \poly(n)$.

\begin{lemma}
  \label{lem:algo:paraByList:computingZeta}
  For a fixed set $S \subseteq \bag{t}$,
  we can compute $\zeta(c)(f,s)$, for all $s \in [0,m]$
	and $f\from\bag{t} \to U$ with $f^{-1}(\top) = S$,
  in time $\O((\max \Ex+1)^{k-\abs{S}} \cdot m \cdot \max \Ex \cdot k^2 \cdot \log(m \cdot \max \Ex))$.
\end{lemma}
\begin{proof}
	We let $\overline S = \bag{t}\setminus S$ be the complement of $S$
	with respect to $\bag{t}$.
	We decompose the function $f$
  as $f=\langle f^\top, f^0 \rangle$
  where $f^\top\from S \to \{\top\}\subseteq U$
  and $f^0\from \overline S \to [0,\max \Ex]\subseteq U$.

  Using the definition of the partial ordering $\preceq$,
  the zeta-transform,
  and the fact that $f^{-1}(\top)=S$,
  we get:
  \begin{align*}
    \zeta(c)(f, s)
   &= \sum_{g \preceq f}
      c[g, s]
    = \sum_{g^\top \preceq f^\top}
      c[\langle g^\top, f^0 \rangle, s] \\
   &= \sum_{g^\top \preceq f^\top}
      \sum_{\substack{f_1,f_2\from\bag{t} \to U \\ \text{s.t.~}f_1\oplus f_2=\langle g^\top, f^0 \rangle} }
      \sum_{s_1+s_2=s}
      c_1[f_1,s_1]\cdot c_2[f_2,s_2] \\
   &= \sum_{g^\top \preceq f^\top}
      \sum_{\substack{g_1,g_2\from S \to U \\ \text{s.t.~}g_1\oplus g_2=g^\top} }
      \sum_{\substack{f_1,f_2\from\overline S \to U \\ \text{s.t.~}f_1 + f_2=f^0} }
      \sum_{s_1+s_2=s}
      c_1[\langle g_1,f_1 \rangle,s_1]\cdot c_2[\langle g_2,f_2\rangle,s_2] \\
  \intertext{Now, observe that, for each pair $g_1,g_2$,
	there is a unique $g^\top\preceq f^\top$ such that $g_1 \oplus g_2 = g^\top$.
  Thus, we sum over each such pair exactly once.
	}
   &= \sum_{g_1,g_2 \preceq f^\top}
      \sum_{\substack{f_1,f_2\from\overline S \to U \\ \text{s.t.~}f_1 + f_2=f^0} }
      \sum_{s_1+s_2=s}
      c_1[\langle g_1,f_1 \rangle,s_1]\cdot c_2[\langle g_2,f_2 \rangle,s_2] \\
  \intertext{
	We reorder the terms
	using that the $g_i$'s are chosen independently from the $f_i$'s.
	}
  &= \sum_{\substack{f_1,f_2\from\overline S \to U \\ \text{s.t.~}f_1 + f_2=f^0} }
     \sum_{s_1+s_2=s}
     \left(
     \sum_{g_1 \preceq f^\top} c_1[\langle g_1,f_1 \rangle,s_1]
     \right)
     \cdot
     \left(
     \sum_{g_2 \preceq f^\top} c_2[\langle g_2,f_2 \rangle,s_2]
     \right) \\
  \intertext{
  We exploit that $f_1^{-1}(\top)=f_2^{-1}(\top)=\emptyset$
  and apply the definition of the $\zeta$-transform.
  }
  &= \sum_{\substack{f_1,f_2\from\overline S \to U \\ \text{s.t.~}f_1 + f_2=f^0} }
     \sum_{s_1+s_2=s}
     \zeta(c_1)(\langle f^\top, f_1 \rangle, s_1)
     \cdot
     \zeta(c_2)(\langle f^\top, f_2 \rangle, s_2)
  \end{align*}
  As the vertices $v$ with $f(v) = \top$ are fixed by assumption,
  the codomain of the functions $f_1$ and $f_2$ can be restricted to $[0,\max \Ex]$
  as the value $\top$ cannot occur anymore.
  Hence, the sum almost corresponds to a standard convolution.

  We apply the same techniques as the one for the \GenFac-algorithm
  presented in \cite{MarxSS21} or originally in \cite{Rooij20}.
  To simplify notation, we set $k^* \deff k - \abs{S}$ in the following.
  We define two functions
  $a_1,a_2\from [0,\max\Ex]^{k^*} \times [0,k^* \max\Ex] \times [0,m] \to \SetF_p$,
  for some prime $p > 2^{\abs{E}}$.
  For $i=1,2$,
	for all $F \in [0,k^* \max \Ex]$,
	and for all $f': \overline S \to [0,\max\Ex]$,
	we set:
  \[
    a_i(f, F, s) \deff \begin{cases}
      \zeta(c_i)(\langle f^\top, f'\rangle, s) & \text{if } \norm{f'} = F \\
      0 & \text{otherwise}
    \end{cases}
    \qquad
    \text{where }
    \norm{f'} \deff \sum_{v \in \bag{t}} f'(v).
  \]
  It suffices to compute the following cyclic
	(i.e., addition is done modulo $\max\Ex+1$)
	and non-cyclic convolution
  for all $f'\in [0,\max\Ex]^{k^*}$,
  $F\in [0,k^*\max \Ex]$,
  and $s\in [0,m]$:
  \[
    a(f',F,s) \deff \sum_{f_1+f_2 \equiv f'} \sum_{\substack{F_1+F_2=F \\ s_1+s_2=s}}
    a_1(f_1,F_1,s_1) \cdot a_2(f_2,F_2,s_2)
  \]
  Then, for all $f'\from \bar S \to [0,\max\Ex]$ and $s\in[0,m]$,
	we set $\zeta(c)(\langle f^\top, f'\rangle, s) = a(f', \norm{f'}, s)$.
  Though the sum ranges over $f_1+f_2\equiv f'$
  where the addition is computed modulo $\max\Ex+1$ for each component,
  we effectively just sum values if $f_1+f_2=f$.
  This is due to the fact that we otherwise have
  that $\norm{f_1}+\norm{f_2} > \norm{f'}$.

  By Lemma~3 in \cite{Rooij20},
  we can compute all table entries $\zeta(c)(f,s)$ in time
  \[
		\O\left(
			(\max \Ex+1)^{k^*} \cdot m \cdot \max \Ex \cdot k^* \cdot
			\left(
				k^* \log(\max \Ex+1) + \log (k^* m \cdot \max\Ex)
			\right)
		\right)
		.\qedhere
	\]
\end{proof}
Now, we can prove how to compute the table entries for the join nodes.
\begin{proof}[Proof of \cref{thm:algo:paraByList:joinNode}]
  First, we compute the zeta transform of $c_1$ and $c_2$
	by invoking \cref{lem:algo:paraByList:zetaIsEnough}.
  This takes time $\O((\max \Ex + 2)^k k m)$.

  \cref{lem:algo:paraByList:computingZeta} shows how to handle the case
  when the set of vertices mapped to $\top$ is fixed.
  Actually we do not have this assumption
	as we consider all possible functions $f$.
  Instead, we iterate over all subsets $S \subseteq \bag{t}$
	and guess by this the vertices $v$ with $f(v) = \top$.
  Then, for each such $S$, we compute $\zeta(c)(f, s)$
	where $f^{-1}(\top) = S$.
  We use the binomial theorem to bound the running time for this computation:
  \begin{align*}
    &\sum_{S \subseteq \bag{t}}
      \O\left(
				(\max \Ex+1)^{k-\abs{S}} \cdot m \cdot \max \Ex
				\cdot k^2 \cdot \log(m \cdot \max \Ex)
			\right) \\
    =&
    \O\left(m \cdot \max \Ex \cdot k^2 \cdot \log(m \cdot \max \Ex)\right)
    \cdot
    \sum_{s \in [0,k]}
      \binom{k}{s}
      \cdot
      (\max \Ex+1)^{k-s}
      \\
    =&
      \O\left(m \cdot \max \Ex \cdot k^2 \cdot \log(m \cdot \max \Ex)\right)
      \cdot
      (\max \Ex+2)^{k}
  \end{align*}
  Using \cref{lem:algo:paraByList:zetaIsEnough},
	we recover the values of $c$ from $\zeta(c)$ in time $\O((\max \Ex + 2)^k k m)$.
  By this the final running time follows.
\end{proof}

\subsection{Parameterizing by the Number of Excluded Degrees}
In this section we prove \cref{thm:algo:main}
which shows that \AntiFactorSize{\ex} is FPT
parameterized by treewidth and the \emph{size} $\ex$ of the set.
We first show a naive algorithm,
i.e., the standard dynamic programming approach,
solving the problem.
In a second step, we improve this algorithm
by using representative sets.
That is, we do not store all solutions
but only so much information such that we can correctly solve the decision and optimization version.

\subsubsection{Naive Algorithm}
Let $\Ex_v \subseteq \SetN$ be the set assigned to vertex $v$
with $\abs{\Ex_v} \le \ex$.
Let $n$ be the number of vertices of $G$
and $m$ the number of edges of $G$.
Let $U = [0,n]$ be the universe of the values in the following.

The idea is to fill a table $\ParSol[\cdot,\cdot]$ with partial solutions.
That is,
for all nodes $t$ of the tree decomposition with bag $\bag{t}$ of size $k$
and all $s\in[0,m]$,
we have $\ParSol[t, s] \subseteq U^{\bag{t}}$ and
$a \in \ParSol[t,s]$
if and only if
there is a set $S \subseteq E_t$ with $\abs{S}=s$
such that $\deg_S(v) \notin \Ex_v$, for all $v \in V_t\setminus \bag{t}$,
and $\deg_S(v) = a[v]$, for all $v \in \bag{t}$.

\subparagraph*{Dynamic Program.}
Initialize the table $\ParSol$ with $\emptyset$ for every entry.
We fill the table iteratively, for all nodes $t$ of the tree decomposition
and all $s\in[0,m]$, in the following way,
depending on $s$ and the type of $t$.
\begin{description}
  \item [Leaf Node.]
  As $\bag{t} = \emptyset$, we set
  $\ParSol[t,0] \deff \{ \emptyset \}$.

  \item [Introduce Vertex Node.]
  Assume $v$ is introduced at $t$,
  i.e., $\bag{t} = \bag{t'} \cup \{ v\}$.
  We define
  \[
    \ParSol[t,s] \deff \{ a_{v \mapsto 0} \mid a \in \ParSol[t',s] \}.
  \]

  \item[Introduce Edge Node.]
  Assume the edge $e=uv$ is introduced at the node $t$.
	We combine the cases where $e$ is not selected for the solution
	and where $e$ is selected.
	Thus, we define:
  \[
    \ParSol[t,s] \deff \ParSol[t',s] \cup
    \{ a_{u \mapsto a(u)+1, v \mapsto a(v)+1} \mid a \in \ParSol[t',s-1] \}.
  \]

  \item[Forget Node.]
  Assume vertex $v$ is forgotten at $t$,
  i.e., $\bag{t} = \bag{t'} \setminus \{v\}$.
  We define
  \[
    \ParSol[t,s] \deff \{ a|_{\bag{t}} \mid a \in \ParSol[t',s]: a[v] \notin \Ex_v \}.
  \]

  \item[Join Node.]
  Assume $t_1$ and $t_2$ are the two children of $t$ with
  $\bag{t} = \bag{t_1} = \bag{t_2}$.
  Then we define
  \[
    \ParSol[t,s] \deff \{ a_1 + a_2 \mid a_1 \in \ParSol[t_1,s_1], a_2 \in \ParSol[t_2,s_2], s_1+s_2=s \}.
  \]
\end{description}
Let $r$ be the root of the tree decomposition with $\bag{r} = \emptyset$.
For a given $s\in[0,m]$,
the algorithm finally checks if
$\ParSol[r,s]\neq \emptyset$,
i.e., $\ParSol[r,s]$ contains the empty vector.
Otherwise no solution exists.
The correctness of this algorithm follows directly from its definition.
Note that the computation might take time $\Omega(n^{\tw+1})$
since the largest bag has size $\tw+1$.

\subsubsection{Improving the Naive Algorithm}
The final algorithm is based on the naive algorithm
but
makes use of so-called representative sets
to keep the size of the set stored for each node of the tree decomposition small.

We first define the notion of representative set
to state the final algorithm.
In \cref{sec:computing-rep-sets}
we show how to actually compute the representative sets.

\begin{definition}[$H$-Compatibility]
  \label{def:Hcompatibility}
	Let $H=(U \dotcup V, E)$
	be an undirected (potentially infinite) bipartite graph.
	We say that \emph{$a \in U$ is $H$-compatible with $b \in V$},
  denoted by $a \comp{H} b$,
  if $(a,b) \in E$.
  \footnote{
  Though the graph is undirected,
  we use tuples to denote the edges.
  By this the first value denotes the vertex from $U$
  and the second value the vertex from $V$.}
\end{definition}
Based on this compatibility notation, we define the $H$-representation of a set.
\begin{definition}[$H$-Representation]
  \label{def:Hrepresentation}
	Let $H=(U \dotcup V, E)$
	be an undirected (potentially infinite) bipartite graph.
	For any $\RepSet \subseteq U$, we say that
	$\RepSet' \subseteq \RepSet$ \emph{$H$-represents} $\RepSet$,
  denoted by $\RepSet'\represents{H} \RepSet$
	if for every $b\in V$:
  \(
    \exists a\in \RepSet:
    a \comp{H} b
    \iff
    \exists a'\in \RepSet':
    a' \comp{H} b.
  \)
\end{definition}
For the algorithm we make use of this $H$-compatibility and $H$-representation
where we use the following graphs.
\begin{definition}[Compatibility Graph]
  \label{def:compatibilityGraph}
  For a set $B = \{ v_1,\dots,v_k \}$ of $k$ vertices
  with sets $\Ex_1,\dots,\Ex_k$ of excluded degrees,
	we define the compatibility graph $\Comp_{B}$ as follows:
	\begin{itemize}
		\item
    $V(\Comp_{B})= U^k \dotcup V^k$
    where the elements in $U,V$ are copies of numbers,
		i.e., $U,V = \SetN$.
		\item
    $E(\Comp_{B})
    =\{((i_1,\dots,i_k),(j_1,\dots,j_k)) \mid \forall~\ell\in[k]: i_\ell+j_\ell \not\in \Ex_\ell\}$.
	\end{itemize}
For a node $t$ with bag $\bag{t}$ of the tree decomposition,
we denote by $\Comp_t$
the graph $\Comp_{\bag{t}}$.
\end{definition}
The intuition is that the vertices in $U^k$
represent the degrees of the constructed partial solution.
The vertices in $V^k$ correspond to the degrees
of some (disjoint) partial solution one might see in the future.
The edges then ``check'' whether both solutions can be combined,
i.e., the degree of each vertex is valid with respect to the union of the solutions.

\subparagraph*{Final Algorithm.}
The improved algorithm applies the same operations
as the naive algorithm to fill a table $c$.
Then,
the algorithm computes a $\Comp_t$-representative set for the table entries
and just stores these values in $c$.
Only these values are used in the next steps
to compute the other table entries.

We show that this preserves the correctness of the algorithm.
\begin{claim}
  \label{lem:algo:paraBySize:correctness}
  For all $t,s$:
  $c[t,s] \represents{\Comp_t} \ParSol[t,s]$.
\end{claim}
\begin{claimproof}
	\newcommand{\ctwo}{\widehat c}
  For ease of notation, we write $a \comp{t} b$ in the following
	if we mean $a \comp{\Comp_t} b$.
  When given an $a \in \ParSol[t,s]$,
  we denote an arbitrary partial solution that agrees with $a$ by $S(a)$.

	Let $\ctwo[t,s]$ be the table entry
	\emph{before} the algorithm computes the representative set.
	By the transitivity of $H$-representation,
  it suffices to show
	that $\ctwo[t,s] \represents{\Comp_t} \ParSol[t,s]$
	to prove the claim.
  As $\ctwo[t,s] \subseteq \ParSol[t,s]$,
  it suffices to show the forward direction
  from the definition of $\Comp_t$-representation.
	The proof is a structural induction on the tree decomposition.
  \begin{description}
    \item[Leaf Node.]
    Obviously true as $\ParSol[t,s]$ only contains the empty vector.

    \item[Introduce Vertex Node.]
    Let $v$ be the vertex introduced at $t$
    and let $t'$ be the unique child of $t$.
    Given some $a \in \ParSol[t,s]$ and some $b$ such that $a \comp{t} b$.

    As $v$ is not incident to any edges yet,
    $a[v]=0$ which implies $b[v]\notin \Ex_v$.
    Thus, $a\vert_{\bag{t'}} \comp{t'} b\vert_{\bag{t'}}$
    and $a\vert_{\bag{t'}} \in \ParSol[t',s]$.
    The induction hypothesis gives us some $a' \in c[t',s]$
    such that $a' \comp{t'} b\vert_{\bag{t'}}$.
    Thus,
		$a'_{v \mapsto 0} \comp{t} b$ and further $a'_{v \mapsto 0} \in \ctwo[t,s]$.

    \item[Introduce Edges Node.]
    Let $uv$ be the edge introduced at $t$
    and let $t'$ be the unique child of $t$.
    Given some $a \in \ParSol[t,s]$
		and some $b$ such that $a \comp{t} b$.
    \begin{itemize}
      \item
      If $uv \notin S(a)$,
      then $a \in \ParSol[t',s]$ and the induction hypothesis provides some
      $a' \in c[t',s]$ with $a' \comp{t'} b$.
      We further get $a' \in \ctwo[t,s]$.
      \item
      If $uv \in S(a)$,
      let $\overline a = a_{u \mapsto a(u)-1, v \mapsto a(v)-1}$.
      Observe that $\overline a \in \ParSol[t',s-1]$
			and $\overline a \comp{t'} \overline b$
      where $\overline b = b_{u \mapsto b(u)+1, v \mapsto b(v)+1}$.
      Again, the induction hypothesis provides some $\overline a' \in c[t',s-1]$
      such that $\overline a' \comp{t'} \overline b$.
      With $a' = \overline a'_{u \mapsto \overline a'(u)+1,v\mapsto \overline a'(v)+1}$,
      we then directly get $a' \in \ctwo[t,s]$
      and further $a' \comp{t} b$.
    \end{itemize}

    \item[Forget Node.]
    Let $v$ be the vertex forgotten at $t$
    and let $t'$ be the unique child of $t$.
    Given some $a \in \ParSol[t,s]$ and some $b$ such that $a \comp{t} b$.

    There must be some $c \notin \Ex_v$ such that $a_{v \mapsto c} \in \ParSol[t',s]$.
    Further $a_{v \mapsto c} \comp{t'} b_{v \mapsto 0}$.
    The IH provides some $a' \in c[t',s]$
    such that $a' \comp{t'} b_{v \mapsto 0}$.
    Hence, $a'[v] \notin \Ex_v$
		and thus, $a'\vert_{\bag{t}} \in \ctwo[t,s]$
		and $a'\vert_{\bag{t}} \comp{t} b$.

    \item[Join Node.]
    Let $t_1$ and $t_2$ be the children of $t$.
    Given some $a\in \ParSol[t,s]$ and some $b$ such that $a \comp{t} b$.

    By the definition of the tree decomposition,
    $S(a)$ can be partitioned according to the decomposition.
    Thus,
		there are $a_1\in \ParSol[t_1, s_1]$
    and $a_2 \in \ParSol[t_2,s_2]$
    such that $a_1+a_2=a$ and $s_1+s_2=s$.
    From the definition of $\Comp_t$-compatibility, we directly get
    $a_1 \comp{t} b+a_2$ and $a_2 \comp{t} b+a_1$.

    The induction hypothesis provides some $a_1' \in c[t_1,s_1]$ such that $a_1' \comp{t} b+a_2$.
    As this is equivalent to $a_2 \comp{t} b+a_1'$,
    we can apply the induction hypothesis once more to get some $a_2' \in c[t_2,s_2]$
    such that $a_2' \comp{t} b+a_1'$.
    As this is equivalent to $a_1'+a_2' \comp{t} b$,
    the claim follows
    since $a_1'+a_2'\in \ctwo[t,s]$.
    \claimqedhere
  \end{description}
\end{claimproof}
\begin{lemma}\label{lem:algo:idea}
	Assume there is an algorithm that can,
	for given $\bag{}=\{v_1,\dots,v_k\}$
	with $\abs{\Ex_{v}} \le \ex$ for all $v\in\bag{}$,
	compute, for a set $\RepSet \subseteq [0,n]^k$,
  a new set $\RepSet' \represents{\Comp_{\bag{}}} \RepSet$
	of size $\Size(k)$
  in time $\Time(k, \abs{\RepSet})$,
  where $\Time$ and $\Size$ are allowed to depend on
  $\Comp_{\bag{}}$ and $\ex$.

  Then, we can decide, for a given \AntiFactorSize{\ex} instance,
  whether there is a solution of size exactly $s$
  in time
  \(
    \Ostar{\Time(\tw+1, (m+1) \cdot \Size(\tw+1)^2)}
		\) and \(
    \Ostar{\Time(\pw+1, 2\Size(\pw+1))}
  \)
  given a tree and a path decomposition of width $\tw$ and $\pw$, respectively.
\end{lemma}
\begin{proof}
  We can assume that $\Time$ and $\Size$ are non-decreasing functions
	and inductively
  that the size of the given table entries is bounded by $\Size(\tw+1)$.
  The running time follows immediately by bounding the size of $c[t,s]$
  and then computing its representative set.
  The correctness follows directly from \cref{lem:algo:paraBySize:correctness}.
\end{proof}

%% file: cofinite/repset.tex
\section{Computing Representative Sets} \label{sec:computing-rep-sets}
As mentioned in the previous section,
one can think of $\Comp_t$-compatibility
as checking whether the given partial solution of degree $a$
fits together with some partial solution of degree $c$ arriving in the future.
This is done via the bipartition of the compatibility graph
and the (non-)existence of the edges,
i.e., checking if $a+c$ is not in $\Ex$.
To compute the representative set
we avoid this two step procedure
by defining the more standard $k$-$q$-compatibility.

\begin{definition}[$k$-$q$-Compatibility]
  \label{def:qCompatibility}
	Let $k$ and $q$ be positive integers.
	For an $a \in \SetN^k$
	and a $b \in \binom{\SetN}{q}^k$,
  we say $a$ \emph{is $k$-$q$-compatible with} $b$,
	denoted by $a \comp[k]{q} b$,
	if and only if,
	for all $i \in [k]$,
  it holds that
  $a[i] \notin b[i]$.
\end{definition}
For our purposes we can relate the two compatibility definitions as follows:
In $\Comp_t$-compatibility one computes $a+c$
and checks if $a+c \notin \Ex$.
Instead $k$-$q$-compatibility checks if $a \notin \Ex-c$.
While both checks are equivalent at this point,
the new compatibility version considers \emph{all} possible sets
of size at most $q=\abs{\Ex}$ and not just $\Ex-c$ for all $c$.
Hence, $k$-$q$-compatibility is independent
from the sets $\Ex_v$ which are assigned to the vertices $v$ of the graph.

\begin{remark*}
  One can also define $k$-$q$-compatibility
  using $H$-compatibility for a bipartite graph $H$ from \cref{def:Hrepresentation}:
  One side of the nodes represents $\SetN^k$
  and the other side $\binom{\SetN}{q}^k$.
  Then, we have an edge between $a \in \SetN^k$
  and $b \in \binom{\SetN}{q}^k$
  if and only if $a[i] \notin b[i]$ for all $i\in[k]$.
\end{remark*}
We extend the notion of compatibility in the standard way
to $k$-$q$-representation.
\begin{definition}[$k$-$q$-Representation]
  \label{def:representativeSet}
	Let $k$ and $q$ be positive integers.
	Given a set $\RepSet \subseteq \SetN^k$,
	and a set $\RepSet' \subseteq \RepSet$.
	We say $\RepSet'$ \emph{$k$-$q$-represents} $\RepSet$,
	denoted by $\RepSet' \represents[k]{q} \RepSet$,
	if and only if
	for all $b \in \binom{\SetN}{q}^k$:
  $\exists a \in \RepSet: a \comp[k]{q} b
  \iff
	\exists a' \in \RepSet': a' \comp[k]{q} b$.
\end{definition}
For both notations, we omit the value $k$ from the notation if $k=1$.
It remains to check that $k$-$q$-compatibility generalizes
$\Comp_t$-compatibility.
\begin{lemma}
  \label{lem:q-RepImpliesH-Rep}
  Let $\bag{}$ be a set of $k$ vertices
  where each $v\in \bag{}$ is assigned a set $\Ex_v$
  such that $\abs{\Ex_v} \le \ex$.
  Then, the following holds for all $\RepSet,\RepSet' \subseteq \SetN^k$:
  If $\RepSet' \represents[k]{\ex} \RepSet$,
  then $\RepSet' \represents{\Comp_{\bag{}}} \RepSet$.
\end{lemma}
\begin{proof}
  To simplify notation we set $\Comp = \Comp_{\bag{}}$ in the following.

  Let $a\in \RepSet$ be such that there is some $b \in \SetN^k$
  with $a \comp{\Comp} b$.
  This implies that
  $a[v]+b[v] \notin \Ex_v$ for all $v \in \bag{}$.
  Rearranging terms yields
  $a[v] \notin \Ex_v - b[v] \deff \{z-b[v] \mid z \in \Ex_v\}$
  from which we get that
  $a \comp[k]{\ex} (\Ex_1-b[v_1], \dots, \Ex_k-b[v_k])$.
  As $\RepSet' \represents[k]{\ex} \RepSet$,
  there is some $a' \in \RepSet'$ such that $a' \comp[k]{\ex} (\Ex_1-b[v_1], \dots, \Ex_k-b[v_k])$.
  We directly get $a'[v] \notin \Ex_v-b[v]$
  or equivalently $a'[v] +b[v] \notin \Ex_v$,
  for all $v \in \bag{}$.
  By the definition of $\Comp$-representation,
  we get $a' \comp{\Comp} b$.
\end{proof}

\subparagraph*{Matroids.}
For the computation of the representative sets we make use of
matroids.
They allow us to formally state the operations we are using.
\begin{definition}[Matroid]
  A matroid is a pair $\Mat = (E,\calI)$,
  where we refer to $I \subseteq 2^E$
  as the \emph{independent sets} of $\Mat$,
  satisfying the following axioms:
  \begin{enumerate}
    \item
    $\emptyset \in \calI$
    \item
    if $I_1 \subseteq I_2$ and $I_2 \in \calI$,
    then $I_1 \in \calI$
    \item
    if $I_1, I_2 \in \calI$ and $\abs{I_1} < \abs{I_2}$,
    then there is some $x \in I_2 \setminus I_1$ such that
    $I_1 \cup \{ x \} \in \calI$
  \end{enumerate}
\end{definition}
In this paper we consider only uniform matroids
as they are sufficient for our purpose.
\begin{definition}[Uniform Matroid]
  Let $U$ be some universe with $n$ elements
  and $r\in \SetN$.
  Then, $\MatU_{r,n} = (U, \binom{U}{\le r})$ is the uniform matroid of rank $r$,
  that is, the matroid over the ground set $U$
  where the independent sets are all subsets of $U$ of size at most $r$.
\end{definition}
Later the rank of these uniform matroids corresponds
to the number of excluded degrees (plus one).
Since the matroid contains all subset of size at most the rank,
we automatically consider all possibilities for upcoming solutions.

There are results proving the existence of small representative sets for matroids \cite{FominLPS16,FominLPS17,KratschW20}.
Since these results are usually for general matroids,
they also apply to uniform matroids which we use here.
However,
as we are not considering a single matroid but the product of several matroids,
the previous results can only be applied partially to our setting.
Moreover, one can suspect that these results can be improved
by exploiting properties of the uniform matroids.
In the following we show two different approaches
to compute the representative sets.
Surprisingly both are incomparable to each other:
The first method gives the faster algorithm when parameterizing by treewidth
while the second method gives the faster algorithm when parameterizing by pathwidth.

\subsection{First Method}
Our first algorithm is based on a previous result for computing representative sets.
Despite the fact that \cref{lem:repset:first} is a special case of Lemma~3.4 in \cite{KratschW20},
our proof uses a completely different technique
as we exploit that the given matroids are uniform.

Let $\omega$ be the matrix multiplication coefficient in the following,
i.e., $\omega < 2.37286$ \cite{AlmanW21}.

\begin{lemma}\label{lem:repset:first}
  Let $\Mat_1, \dots, \Mat_k$ be $k$ uniform matroids,
  each of rank $r$,
  with integer universes $U_1,\dots,U_k$.
  Given a set $\RepSet \subseteq U_1 \times \dots \times U_k$,
  we can find a set $\RepSet' \represents[k]{r-1} \RepSet$ of size $r^k$
  in time $\O(\abs{\RepSet} \cdot r^{k (\omega-1)} k)$.
\end{lemma}
We essentially follow the proof of Theorem~12.15 in \cite{ParameterizedAlgos}
and modify it at those places where we can get better results.
Before we start with the proof,
we first introduce some notation and results related to matroids.

It is known that every uniform matroids $\MatU=(U,\calI)$ of rank $r$ can be represented
by a $r \times \abs{U}$ Vandermonde matrix $M$
(where the first row consists only of $1$s)
over the finite field with $p$ elements,
where $p$ must be larger than $\abs{U}$.
Each column of $M$ then corresponds to one element in $U$.
If a subset of these columns is independent,
then the corresponding elements form an independent set in $\calI$.
For all $A \subseteq U$,
we define $M^{(A)}$ as the submatrix of $M$
consisting only of those columns that correspond to elements from $A$.

For a matrix $M$ with $r$ rows,
let $M[I]$ denote the submatrix of $M$
containing only the rows indexed by the elements of $I \subseteq [r]$.

We use the following observation in the proof of \cref{lem:repset:first}.
\begin{observation}[Observation~12.17 in \cite{ParameterizedAlgos}]
  \label{obs:repset:relationDeterminantAndCompatibility}
  Let $A \in U$ and $B \subseteq U$ with $\abs{B} = r-1$.
  Then, $A \comp{r-1} B$
  if and only if
  $\det([M^{(A)} \mid M^{(B)}]) \neq 0$.
\end{observation}

\begin{proof}[Proof of \cref{lem:repset:first}]
  Assume that $M_1,\dots,M_k$ are the matrix representations
  of the matroids $\Mat_1,\dots,\Mat_k$.
  Enumerate all elements $I$ of $[r]^k$
  in an arbitrary order $I_1,\dots,I_{r^k}$.

  Then, for all $A \in \RepSet$, we compute the vector $v_A$,
  where, for all $j \in [r^k]$,
  we set
  \[
    v_A[j] \deff \prod_{i=1}^k M_i[I_j[i], A[i]].
  \]
  Construct a $r^k \times \abs{\RepSet}$ matrix $Q$,
  where the vectors $v_A$ are the columns of $Q$.
  Then, find a column basis $B_Q$ of $Q$
  and output the set $\RepSet' = \{ A \mid v_A \in B_Q \}$ as solution.
  Obviously $B_Q$ contains at most $r^k$ elements.

  Computing the vectors $v_A$
  takes time $\O(\abs{\RepSet} \cdot r^k \cdot k)$ in total.
  As the computation of the basis takes time
  $\O(\abs{\RepSet}\cdot r^{k (\omega-1)})$,
  the complete procedure requires time
  $\O(\abs{\RepSet} \cdot r^{k (\omega-1)} \cdot k)$.

  \proofsubparagraph{Correctness.}
  It remains to show that the set $\RepSet'$
  indeed $k$-$q$-represents $\RepSet$.
  For this we start with some observations about the vectors we just computed.
  Let $A \in U_1 \times \dots \times U_k$ and
  $B \in \binom{U_1}{r-1} \times \dots \times \binom{U_k}{r-1}$
  in the following.

  By \cref{obs:repset:relationDeterminantAndCompatibility}
  and the coordinatewise definition of $k$-$q$-compatibility, we get
  \begin{equation}
    \label{eq:repset:relationCompAndDet1}
    A \comp[k]{r-1} B
    \iff
    \prod_{i=1}^k \det\left(\left[M_i^{(A[i])} \mid M_i^{(B[i])}\right]\right)
    \neq 0.
  \end{equation}
  Since $\abs{A[i]}=1$,
  it holds that, for all $j \in [r^k]$,
  \[
    v_A[j]
    = \prod_{i=1}^k M_i[I_j[i], A[i]]
    = \prod_{i=1}^k\det\left( M_i^{(A[i])}[I_j[i]] \right).
  \]
  Now, for all $I_j$,
  we define $\overline I_j \in \binom{r}{r-1}^k$
  such that $\overline I_j[\ell] \deff [r] \setminus \{ I_j[\ell] \}$,
  for all $\ell\in[k]$.
  We set:
  \[
    u_B[j] = \prod_{i=1}^k \det\left( M_i^{(B[i])}[\overline I_j[i]] \right)
    .
  \]
  To simplify notation, we set $\sigma_j = (-1)^{\sum_{i=1}^{r^k} I_j[i]}$.

  Using the Laplacian expansion to rewrite the computation of the determinant
  in Equivalence~\eqref{eq:repset:relationCompAndDet1},
  we get:
  \begin{equation}
    \label{eq:repset:relationCompAndDet2}
    A \comp[k]{r-1} B
    \iff
    \sum_{j=1}^{r^k} \sigma_j v_A[j] u_B[j] \neq 0
    .
  \end{equation}

  Let $A \in \RepSet$ such that $A \comp[k]{r-1} B$ for some suitable $B$.
  We know that $\sum_{j=1}^{r^k} \sigma_j v_A[j] u_B[j] \neq 0$.
  As we are given a column basis $B_Q$ for the matrix $Q$,
  we can write $v_A$ such that
  \[
    v_A = \sum_{A' \in \RepSet'} \lambda_{A'} v_{A'}
    .
  \]
  We can substitute this in the sum of Equivalence~\eqref{eq:repset:relationCompAndDet2}
  to get
  \[
    0
    \neq \sum_{j=1}^{r^k} \sigma_j v_A[j] u_B[j]
    = \sum_{j=1}^{r^k}  \sum_{A' \in \RepSet'} \lambda_{A'} \sigma_j v_{A'}[j] u_B[j]
    = \sum_{A' \in \RepSet'} \lambda_{A'} \sum_{j=1}^{r^k} \sigma_j v_{A'}[j] u_B[j]
    .
  \]
  Hence, there must be at least one $A' \in \RepSet'$ such that
  $\sum_{j=1}^{r^k} \sigma_j v_{A'}[j] u_B[j] \neq 0$.
  By Equivalence~\eqref{eq:repset:relationCompAndDet2},
  this implies that $A' \comp[k]{r-1} B$
  which completes the proof.
\end{proof}
To finish the algorithm for \AntiFactorSize{\ex}
from \cref{thm:algo:main}, it remains,
by \cref{lem:q-RepImpliesH-Rep}, to compute a $k$-$\ex$-representative set
as $\abs{\Ex_v} \le \ex$.
To achieve this we define $k$ uniform matroids
with universe $\{0,\dots,n\}$ and rank $\ex+1$.
Then, plugging in the values from \cref{lem:repset:first} into \cref{lem:algo:idea},
directly gives the following result.
Note that we can assume $\ex \le n$.
\begin{corollary}
  Given a tree and a path decomposition,
  \AntiFactorSize{\ex} can be solved in time
  $\Ostar{(\ex+1)^{(\omega+1)\cdot\tw}}$
  and $\Ostar{(\ex+1)^{\omega\cdot\pw}}$,
  respectively.
\end{corollary}

\subsection{Second Method}
Kratsch and Wahlström showed \cite[Lemma~3.4]{KratschW20}
a method to compute representative set of products of matroids
when the underlying matroid is general.
It is known
that computing representative sets for a single uniform matroid can be done faster
(\cite[Section~4]{FominLPS16} and \cite[Section~3]{ShachnaiZ16}).
As the result by Kratsch and Wahlström heavily depends on the algorithm for general matroids
(i.e., they do not use it as a black box),
we cannot combine it with the improved algorithms for uniform matroids.
In the previous section we modified an algorithm to exploit the properties of uniform matroids.
In this section, we directly use the faster algorithm
for a single uniform matroid
to give algorithms
not depending on matrix multiplication.
\begin{lemma}[Modified version of Theorem~4.15 in \cite{FominLPS16}]
  \label{lem:repset:helperSecond}
  There is an algorithm that,
  for a universe $U$ of size $n$,
  given a set $\RepSet \subseteq \binom{U}{p}$
  and an integer $q$,
  computes in time
  $\O(\abs{\RepSet} \cdot (1-t)^{-q} \cdot 2^{o(q+p)} \cdot \log n)$
  a subfamily $\RepSet' \represents{q} \RepSet$
  such that $\abs{\RepSet'} \le t^{-p} (1-t)^{-q} \cdot 2^{o(q+p)}$.
  Where $0<t<1$ can be chosen arbitrarily.
\end{lemma}
We combine \cref{lem:repset:helperSecond} with the result
from \cref{lem:algo:idea} to prove \cref{lem:repset:second}.
\begin{lemma}\label{lem:repset:second}
  Assuming a tree and path decomposition is given,
  \AntiFactorSize{\ex} can be solved in time
  $\Ostar{(\ex+1)^{4.1\cdot\tw}2^{o(\ex\cdot\tw)}}$
  and $\Ostar{(\ex+1)^{2.28\cdot\pw}2^{o(\ex\cdot\pw)}}$,
  respectively.
\end{lemma}
\begin{proof}
  Let $U=[0,n]$.
  By \cref{thm:algo:main},
  it remains to compute the representative set of a given set $\RepSet$.
  Create $k$ copies $U^{(1)},\dots,U^{(k)}$ of $U$,
  that is $U^{(i)} = \{ j^{(i)} \mid j \in U \}$,
  and define $\widehat U = U^{(1)}\cup\dots\cup U^{(k)}$.
  The copies of the universe allow us to distinguish between the values for the different dimensions.
  For all $A \in \RepSet$,
  we let $\widehat A \deff \{ u_j^{(i)} \mid A[i]=u_j \}$
  and define $\widehat \RepSet \deff \{ \widehat A \mid A \in \RepSet \}$.

  Finally, we define $\Mat$ to be the uniform matroid of rank $(\ex+1)k$
  with universe $\widehat U$ of $k(n+1)$ elements.
  Clearly all sets in $\widehat \RepSet$ are independent sets.
  Then, we apply \cref{lem:repset:helperSecond} with
  $p = k$ and
  $q= \ex k$
  where we choose the parameter $t$ later
  to get a set $\widehat \RepSet' \represents{\ex k}{} \widehat \RepSet$
  of size
  \(
    \abs{\widehat \RepSet'}
    \le t^{-k} (1-t)^{-\ex k} \cdot 2^{o(\ex k+k)}
  \)
  in time
  \(
    \abs{\widehat \RepSet} \cdot (1-t)^{-\ex k} \cdot 2^{o(\ex k+k)} \cdot \log kn
  \).

  It remains to set $t$ such that the running time of the final algorithm is minimized.
  For the parameterization by treewidth this means minimizing
  \[
    {\left(t^{-k} (1-t)^{-\ex k} \cdot 2^{o(\ex k+k)}\right)}^2
    \cdot (1-t)^{-\ex k} \cdot 2^{o(\ex k+k)} \cdot \log kn
    \le
    t^{-2k} (1-t)^{-3\ex k} \cdot 2^{3 \cdot o(\ex k)} \cdot \log n
    .
  \]
  It suffices to minimize $t^{-2} (1-t)^{-3\ex}$
  which is achieved by setting $t=2 / (3\ex+2)$.
  When replacing this in the above runtime
  the claim follows by simple computations.
  The coefficient follows from considering the case for $\ex=2$,
  which is maximizing the runtime.
  \footnote{The case where $\ex \le 1$ is polynomial time solvable and thus not relevant.}

  For the parameterization by pathwidth we have to choose a different parameter
  as we can exploit that there are no join nodes.
  We minimize
  \[
    2 \cdot t^{-k} (1-t)^{-\ex k} \cdot 2^{o(\ex k+k)}
    \cdot (1-t)^{-\ex k} \cdot 2^{o(\ex k+k)} \cdot \log kn
    \le
    t^{-k} (1-t)^{-2\ex k} \cdot 2^{2 \cdot o(\ex k)} \cdot \log n
  \]
  by setting $t=1/(2\ex +1)$
  which leads to the claimed running time.
\end{proof}

%% file: cofinite/half-induced.tex
\section{Half-Induced Matchings}\label{sec:half-induced-matchings}
In this section, we introduce \emph{half-induced matchings}
and show relations to compatibility graphs and representative sets.
We use these properties later
to prove the lower bounds for the decision and optimization version
of \AntiFactor{\Ex} and \AntiFactorSize{\ex}.

\begin{definition}[Half-induced Matching]
	Let $G=(U\dotcup V, E)$ be a bipartite graph.
	$G$ has a \emph{half-induced matching} of size $\ell$
	if there are
	pairwise different $a_1,\dots,a_\ell \in U$
	and pairwise different $b_1,\dots,b_\ell \in V$
	such that
	(1) $(a_i,b_i)\in E$, for all $i$,
	but
	(2) $(a_i,b_j)\not\in E$, for all $j>i$.
\end{definition}
By an abuse of notation,
$\Comp_\Ex$ denotes the compatibility graph for a vertex
with set $\Ex$ of forbidden degrees.
We show that arithmetic progressions in the set of excluded degrees
are sufficient to obtain large half-induced matchings
in the corresponding compatibility graph.
\begin{lemma}
	\label{claim:lower:progression-to-HIM}
	If $\Ex$ contains an arithmetic progression of length $\ell$,
	but not one of length $\ell+1$,
	then $\Comp_\Ex$
	has a half-induced matching of size $\ell+1$.
\end{lemma}
\begin{proof}
	Let $a,a+d,a+2d,\ldots,a+(\ell-1)d\in \Ex$ be an
	arithmetic progression with $d\geq 1$
	such that $a+\ell d\not\in \Ex$.
	We construct the following half-induced matching in $\Comp_\Ex$
	where,
	for all $i\in[\ell+1]$, we set
	\(
		a_i \deff d(i-1)
		\) and \(
		b_i \deff a+(\ell+1-i)d
	\).

	Then, for all $i\in[\ell+1]$,
	we have
	\(
		a_i + b_i =
		d(i-1) + a+(\ell+1-i)d = a+\ell d \not\in \Ex
	\)
	and hence, $(a_i,b_i)\in E(\Comp_\Ex)$.
	Similarly, for all $i\in[\ell]$ and all $i<j\in [\ell+1]$,
	we have $(a_i,b_j)\not\in E(\Comp_\Ex)$ because
	\(
		a_i + b_j =
	 	d(i-1) + a+(\ell+1-j)d = a+(\ell+i-j)d \in \Ex
		.
	\)
\end{proof}
Conversely to the previous lemma,
we also prove that arithmetic progressions are necessary
to obtain large half-induced matchings.
\begin{lemma}
	\label{lem:lower:HIM-to-progression}
	Let $\Ex \subseteq \SetN$ with $\abs{\Ex}=\ell\ge 2$.
	Suppose $\Comp_\Ex$ contains a half-induced matching of size $\ell+1$.
	Then, $\Ex$ is an arithmetic progression.
\end{lemma}
\begin{proof}
	Let $a_1,\dots,a_{\ell+1}$ and $b_1,\dots,b_{\ell+1}$ be the vertices
	of the half-induced matching of size $\ell+1$ in $\Comp_\Ex$.
	Then, we have the following constraints:
	\begin{align*}
		a_1+b_2 \in \Ex,&& a_1+b_3\in \Ex,&& \dots,&& a_1+b_{\ell+1}\in \Ex \\
                   && a_2+b_3\in \Ex,&& \dots,&& a_2+b_{\ell+1}\in \Ex
	\end{align*}
	Let $\Ex=\{x_1,x_2,\ldots,x_\ell\}$ where $x_i\leq x_{i+1}$.
	From $a_1+b_j\neq a_1+b_{j'}$ for any $j\neq j'$, we get
	\[
		\{a_1+b_2,a_1+b_3,\ldots,a_1+b_{\ell+1}\}=\Ex.
	\]
	Now, consider the second set of constraints.
	As the $b_j$ are pairwise different, there is some $i$ such that
	\begin{align*}
		\{a_2+b_3,a_2+b_4,\ldots,a_2+b_{\ell+1}\}&=\Ex\setminus\{x_i\}
		.
		\intertext{
		We set \(d \deff a_2-a_1 \).
		Assuming \( d > 0 \), we get \( x_\ell + d \notin \Ex \) and thus
		}
		\{x_1+d,x_2+d,x_3+d,\ldots, x_{\ell-1}+d\}&=\Ex\setminus\{x_i\}
		.
	\end{align*}
	Now, observe that $x_1$ cannot belong to the left-hand side because
	$d>0$ and $x_1=\min(\Ex)$.
	Thus, we have that $i=1$.
	Similarly, for $a_2-a_1=d<0$ we can argue that $i=\ell$.
	Without loss of generality consider the former case.
	Then, we have that
	\( x_i + d = x_{i+1} \), for all $i\in[\ell]$.
	Hence, $\Ex$ is an arithmetic progression of length $\ell$.
\end{proof}
For a graph $\Comp$ and an integer $k > 1$,
we extend $\comp{\Comp}$ and $\represents{\Comp}$ to $k$ dimensions,
denoted by $\comp[k]{\Comp}$ and $\represents[k]{\Comp}$,
such that the $\Comp$-compatibility must hold for each dimension.
\begin{lemma}\label{lem:half-induced-implies-repset}
	Let $\Ex\subseteq \SetN$, and
	$\epsilon>0$ and $\ell\geq 2$ be constants.
	Then, there exists a constant $k$ depending only on $\epsilon$ and $\ell$
	such that the following holds.
	Suppose the compatibility graph $\Comp_\Ex$
	contains a half-induced matching of size $\ell$.
	Then, there is a set $\RepSet\subseteq \SetN^k$ such that every
	representative set $\RepSet'\represents[k]{\Comp_\Ex}\RepSet$ has size
	$\abs{\RepSet'}\geq (\ell-\epsilon)^k$.
\end{lemma}
Before we prove the lemma, we briefly discuss its implications.
The running time of the algorithm for \AntiFactor{\Ex} from \cref{thm:algo:main}
depends on the size of the representative sets computed.
\Cref{lem:half-induced-implies-repset} implies that any such algorithm using
representative sets in a similar way takes time at least
$(\ell-\epsilon)^{\pw}$.
This can be seen as an \emph{unconditional} version of the lower bounds
for the decision and optimization version
shown in \cref{thm:dec:lbAntiFactor,thm:opt:lbAntiFactor}.

\begin{proof}
	We set the value of $k$ later.
	Let the half-induced matching be between $A, B \subseteq \SetN$
	with $A=\{a_1,\dots,a_\ell\}$ and $B=\{b_1,\dots,b_\ell\}$.
	We define indexing functions $\ind_A$ and $\ind_B$
	such that $\ind_A(a_i)=\ind_B(b_i)=i$.
	For $s\in A^k$, we define $\ind_A(s)=\sum_{i \in [k]}\ind_A(s[i])$.
	We partition $A^k$ into sets $\RepSet_q$
	with $q \in [\ell \cdot k]$ such that
	\[
		\RepSet_q = \{s\in A^k \mid \ind_A(s) = q \}.
	\]
	Hence,
	there exists some $q' \in [\ell \cdot k]$ such that
	\begin{align*}
		\abs{\RepSet_{q'}} \geq \frac{\ell^k}{\ell \cdot k}.
	\end{align*}
	Let $\RepSet = \RepSet_{q'}$ be the set
	for which we want a lower bound on the size of its representative sets.
	To simplify notation, we let $q=q'$.
	Now, consider some $s\in \RepSet$.
	We claim that $s$ is the unique compatible element for
	$t \in B^k$, where
	\[
		t[i] = b_{\ind_A(s[i])} .
	\]
	It is clear that $s\comp[k]{\Comp_\Ex} t$.
	Suppose there is some other $s'\in \RepSet$ such that
	$s'\comp[k]{\Comp_\Ex} t$.
	Then, since $\ind_A(s)=\ind_A(s')=q$ and since $s\neq s'$,
	there is some index $j$ such that $\ind_A(s[j]) > \ind_A(s'[j])$.
	The $j$th index of $s'+t$ is
	\(
		a_{\ind_A(s'[j])} + b_{\ind_A(s[j])}
	\)
	because $s'[j]=a_{\ind_A(s'[j])}$.
	However, observe that this sum must be in $\Ex$ from the fact that
	there is a half-induced matching between $A$ and $B$ in $\Comp_\Ex$
	and the fact that $\ind_A(s[j]) > \ind_A(s'[j])$.
	This is a contradiction,
	implying that $s$ is the only compatible partner of $t$.
	Thus, $s$ is forced to belong to any
	representative set $\RepSet'\represents[k]{\Comp_\Ex}\RepSet$.

	Since the above argument holds for all $s\in \RepSet$,
	we conclude that the only representative set for $\RepSet$ is itself.
	Now, we set $k$ to be large enough such that
	$k\log (\ell- \epsilon) \leq k\log(\ell)-\log(\ell \cdot k)$.
	Then, we have
	\[
		\abs{\RepSet} \geq \frac{\ell^k}{\ell \cdot k} \geq (\ell-\epsilon)^k.
		\qedhere
	\]
\end{proof}
We conjecture that the converse of \cref{lem:half-induced-implies-repset} is also true.
For example, for $\Ex=\{10,100,1000,\ldots\}$
the largest half-induced matching in $\Comp_{\Ex}$ is of size three, a constant,
(even though $\Ex$ itself is infinite).
Intuitively, the
size of the representative set itself must be small
because knowing any
\emph{two} forbidden degrees of a vertex in
the future solution is enough for us to deduce
the degree of the vertex in the partial solution.
\begin{conjecture}
	\label{conj:himImpliesUpperBoundOnRepSet}
Let $\Ex\subseteq \SetN$ and $\ell\geq 2$ be a constant.
Then, there exists a constant $k$
depending only on $\ell$
such that the following holds.
Suppose the largest half-induced matching in $\Comp_\Ex$
has size $\ell$.
Then, every $\RepSet\subseteq \SetN^k$ has a
representative set  $\RepSet'\represents[k]{\Comp_\Ex}\RepSet$ with
$\abs{\RepSet'}\leq \ell^{k+o(k)}$.

\end{conjecture}
Recall,
the runtime of the algorithm in \cref{thm:algo:paraByList} depends on $\max\Ex$
but the lower bound in \cref{thm:dec:lbAntiFactor}
on the size $\ell$ of the half-induced matching.
With \cref{conj:himImpliesUpperBoundOnRepSet} it seems reasonable to get algorithms
for the decision and optimization version based on representative sets
with a running time depending on $\ell$.
This would complement the lower bound.
Note, by \cref{thm:count:lbAntiFactor},
the $\Ostar{(\max\Ex+2)^{\pw}}$ algorithm for the counting version
is essentially optimal.

%% file: cofinite/dec.tex
\section{Lower Bounds for the Decision Version}\label{sec:lower:dec}
In this section we prove the lower bounds for the decision version of \AntiFactor{\Ex} and \AntiFactorSize{\ex}.
Instead of showing the lower bound directly,
we first define the following intermediate problem
and show the hardness of this problem.

\begin{definition}[\AntiFactorR{\Ex}]
	\label{def:dec:antiFactorR}
	Let $\Ex\subseteq \SetN$ be fixed and finite.
	Let
	$G=(V_S\dotcup V_C, E)$ be a vertex labeled graph
	such that
	\begin{itemize}
		\item
		all vertices in $V_S$,
		called \emph{simple} vertices,
		are labeled with set $\Ex$,

		\item
		all vertices $v \in V_C$,
		called \emph{complex} vertices,
		are labeled with a relation $R_v$
		that is given as a truth table
		such that $R_v \subseteq 2^{I(v)}$
		where $I(v)$ is the set of edges incident to $v$ in $G$.
	\end{itemize}
	A set $\widehat E \subseteq E$ is a \emph{solution} for $G$ if
	(1) for $v\in V_S$: $\deg_{\widehat E}(v) \notin \Ex$
	and
	(2) for $v \in V_C$: $I(v)\cap \widehat E \in R_v$.

	\emph{\AntiFactor{\Ex} \textsc{with Relations} (\AntiFactorR{\Ex})}
	is the problem of deciding
	if such an instance $G$ has a solution.
\end{definition}
We show our lower bounds
based on this problem definition.

\begin{lemma}[Lower Bound for \AntiFactorR{\Ex}]
	\label{lem:dec:lbAntiFactorR}
	Let $\Ex \subseteq \SetN$ be a fixed set
	which contains a half-induced matching of size $h\ge 2$.

	Let $f_\Ex\from \SetN\to\SetR^+$ be an arbitrary function
	that may depend on the set $\Ex$.

	For every constant $\epsilon>0$,
	there is no algorithm that can solve \AntiFactorR{\Ex}
	in time $\Ostar{(h-\epsilon)^{\pw+f_{\Ex}(\totalDeg)}}$,
	where $\totalDeg = \max_{\text{bag } \bag{}} \sum_{v \in \bag{}\cap V_C} \deg(v)$,
	even if we are given a path decomposition of width $\pw$,
	unless SETH fails.
\end{lemma}
In a second step we
remove the relations and replace them by appropriate gadgets.
To be able to reuse the reduction later
we introduce a slightly more general version of the problem.
For two finite sets $\Ex,\Y\subseteq\SetN$,
we define \AntiFactor{(\Ex, \Y)}
as the generalization of \AntiFactor{\Ex}
where we allow the sets $\Ex$ and $\Y$ to be assigned to the vertices.
We show hardness when $0\in \Ex$
and when $\maxgap(\Exbar)>1$. The former is to
ensure that there are no trivial solutions and the
latter ensures that the problem is not polynomial-time solvable \cite{Cornuejols88}.
Recall that $\maxgap(\Exbar)$ is the size of the
largest contiguous sequence of integers not in $\Exbar$ but whose boundaries are
in $\Exbar$.

\begin{lemma}
	\label{lem:dec:AntiFactorRToAntiFactor}
	Fix a finite set $\Ex \subseteq \SetN$ such that
	$0\in \Ex$
	and $\maxgap(\Exbar) > 1$.
	Let $\Y \subseteq\SetN$ be arbitrary.
	There is a many-one reduction from \AntiFactorR{\Y} to \AntiFactor{(\Ex,\Y)}
	such that pathwidth increases by at most $f(\totalDeg)$ and
	size by a factor of $f(\totalDeg)$,
	where $\totalDeg = \max_{\text{bag } \bag{}} \sum_{v \in \bag{}\cap V_C} \deg(v)$.
\end{lemma}
By combining the lower bound for the intermediate problem
from \cref{lem:dec:lbAntiFactorR}
with the reduction from \cref{lem:dec:AntiFactorRToAntiFactor},
we can show the lower bounds
for \AntiFactor{\Ex} and \AntiFactorSize{\ex}.
\begin{theorem}[Lower Bound for Decision Version I]
	\label{thm:dec:lbAntiFactor}
	Fix a finite set $\Ex \subseteq \SetN$ such that
	\begin{itemize}
		\item
		$0\in \Ex$
		and
		$\maxgap(\Exbar)>1$,
		\item
		and $\Ex$ contains a half-induced matching of size $h$.
	\end{itemize}
	For every constant $\epsilon>0$,
	there is no algorithm that can solve \AntiFactor{\Ex} in time $\Ostar{(h-\epsilon)^\pw}$
	even if we are given a path decomposition of width $\pw$,
	unless SETH fails.
\end{theorem}
\begin{proof}
	Let $H$ be a given \AntiFactorR{\Ex} instance
	for which we apply \cref{lem:dec:AntiFactorRToAntiFactor}
	with $\Ex=\Y$
	to obtain the \AntiFactor{\Ex} instance $G$.

	We know $n_G \le n_H \cdot f(\totalDeg_H)$
	and $\pw_G \le \pw_H + f(\totalDeg_H)$.
	Assume we are given a faster algorithm for \AntiFactor{\Ex}
	and run this algorithm on the instance $G$ in the following time:
	\begin{align*}
		(h-\epsilon)^{\pw_G} \cdot n_G^{\O(1)}
		&\le (h-\epsilon)^{\pw_H + f(\totalDeg_H)} \cdot (n_H \cdot f(\totalDeg_H))^{\O(1)} \\
		&\le (h-\epsilon)^{\pw_H + f(\totalDeg_H) + f''(\totalDeg_H)} \cdot n_H^{\O(1)} \\
		&\le (h-\epsilon)^{\pw_H + f''(\totalDeg_H)} \cdot n_H^{\O(1)}
	\end{align*}
	Thus, the running time directly contradicts SETH by \cref{lem:dec:lbAntiFactorR}.
\end{proof}
The following \cref{thm:dec:lbAntiFactorBySize}
extends \cref{thm:dec:lbAntiFactor}
to the more general \AntiFactorSize{\ex} problem
and shows a more informative lower bound.
\begin{theorem}[Lower Bound for Decision Version II]
	\label{thm:dec:lbAntiFactorBySize}
	For all $\ex\ge 3$ and $\epsilon>0$,
	\AntiFactorSize{\ex} cannot be solved in time $\Ostar{(\ex+1-\epsilon)^{\pw}}$
	on graphs given with a path decomposition of width $\pw$,
	unless SETH fails.
\end{theorem}
\begin{proof}
	We set $\Y \deff \{2,4,\dots,2\ex\}$
	and $\Ex \deff \{0,2,3\}$.
	By \cref{claim:lower:progression-to-HIM},
	$\Y$ contains a half-induced matching of size $\ex+1$.
	Moreover, $\Exbar$ contains a gap of size two.

	We use \cref{lem:dec:AntiFactorRToAntiFactor}
	to transform a \AntiFactorR{\Y} instance
	into an \AntiFactor{(\Ex,\Y)} instance.
	From $\abs{\Ex},\abs{\Y}\le\ex$
	and by the properties of $\Ex$ and $\Y$,
	the claim follows directly.
\end{proof}

\input{cofinite/dec-lb}
\input{cofinite/dec-rel}

%% file: cofinite/dec-lb.tex
\subsection{Replacing Finite Sets by Cofinite Sets}
In this section, we prove \cref{lem:dec:lbAntiFactorR},
i.e., the lower bound for \AntiFactorR{\Ex},
based on a lower bound from \cite{MarxSS21}
for the intermediate problem \BFR.
\begin{definition}[\BFR (Definition~4.1 in \cite{MarxSS21})]\label{def:bfactor-relation}
	Let $B\subseteq \SetN$ be fixed of finite size.
	$G=(V_S\dotcup V_C, E)$ is an instance of \emph{\BFactorRelation (\BFR)}
	if all nodes in $V_S$ are labeled with set $B$
	and all nodes $v \in V_C$ are labeled with a relation $R_v$
	that is given as a truth table
	such that the following holds:
	\begin{enumerate}
		\item Let $I(v)$ be the set of edges incident to $v$ in $G$.
		Then $R_v \subseteq 2^{I(v)}$.
		\item There is an even $c_v>0$
		such that for all $x \in R_v$ we have $\hw(x)=c_v$.
	\end{enumerate}
	A set $\widehat E \subseteq E$ is a \emph{solution} for $G$ if
	(1) for $v\in V_S$: $\deg_{\widehat E}(v) \in B$
	and
	(2) for $v \in V_C$: $I(v)\cap \widehat E \in R_v$.
	\BFR is the problem of deciding if such an instance has a solution.
	\\
	We call $V_S$ the set of \emph{simple} nodes and $V_C$ the set of \emph{complex} nodes.
\end{definition}
We use the corresponding lower bound for \BFR and the restrictions to the graph
as a starting point for our construction.
\begin{lemma}[Corollaries~4.7 and 4.8 in the full version of \cite{MarxSS21}]
	\label{lem:lbBFactorR}
	Let $B \subseteq \SetN$ be a fixed and finite set.
	Given a \BFR instance
	\begin{itemize}
		\item
		and its path decomposition of width $\pw$
		with $\totalDeg = \max_{\text{bag } \bag{}} \sum_{v \in \bag{}\cap V_C} \deg(v)$,
		\item
		moreover the simple vertices form an independent set
		and each simple vertex is only connected to $2$ complex nodes
		by exactly $\max B$ (parallel) edges each,
		\item
		and we are given the promise that with respect to any solution
		the degree of the simple vertices is exactly $\max B$.
	\end{itemize}
	Assume \BFR can be solved in such a case in $\Ostar{(\max B+1-\epsilon)^{\pw+f_B(\totalDeg)}}$ time
	for some $\epsilon > 0$
	and some function $f_B\from\SetN\to\SetR^+$ that may depend on the set $B$.
	Then SETH fails.
	Moreover the result also holds for \CountBFR and \#SETH.
\end{lemma}
To show a lower bound for \AntiFactorR{\Ex},
it suffices to replace the simple vertices with set $B$
by an appropriate gadgets consisting of simple vertices with set $\Ex$
and complex vertices.

\begin{figure}
	\centering
	\begin{subfigure}[b]{0.5\textwidth}
		\centering
		\includegraphics[page=1]{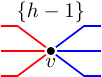}
		\caption{The simple vertex $v$ before the modifications.}
		\label{fig:dec:BFactorRToAntiFactorR:original}
	\end{subfigure}\hfill
	\begin{subfigure}[b]{0.5\textwidth}
		\centering
		\includegraphics[page=2]{img/dec/dec-finite-to-cofinite}
		\caption{The gadget replacing vertex $v$.}
		\label{fig:dec:BFactorRToAntiFactorR:modified}
	\end{subfigure}\hfill

	\caption{The transformation
	in the proof of \cref{lem:dec:lbAntiFactorR}.
	The red, orange, green, and blue edges represent the
	left-external, left-internal, right-internal, and right-external edges, respectively.}
	\label{fig:dec:BFactorRToAntiFactorR}
\end{figure}

\subparagraph*{Modification of the Graph.}
	Let $a_0,\dots,a_{h-1}$ and $b_0,\dots,b_{h-1}$
	be the labels of the half-induced matching of size $h$ of $\Ex$
	and let $U$ be the maximum over these labels.
	Let $H$ be a \BFR instance as stated in \cref{lem:lbBFactorR}
	with $\max B = h-1$.%
	\footnote{%
	It actually suffices to set $B = \{h-1\}$.
	}
	We replace each simple vertex by the following gadget
	and keep the other vertices unchanged
	(see \cref{fig:dec:BFactorRToAntiFactorR}).

	By assumption, each simple vertex $v$ is incident to $2(h-1)$ edges
	which we can partition into two sets of size $h-1$
	depending on their endpoints.
	We call these groups of edges the left-external
	and right-external edges.
	We remove $v$
	and connect the left-external edges
	to a new complex vertex $v_{\LEFT}$ with relation $R_{\LEFT}$.
	The right-external edges are connected similarly
	to another new complex vertex $v_{\RIGHT}$ with relation $R_{\RIGHT}$.
	As a last step, we create a new simple vertex $v'$ with set $\Ex$.
	We connect $v'$ by $U$ (parallel) edges to $v_{\LEFT}$ and
	call these edges the left-internal edges.
	Additionally, we connect $v'$
	by $U$ parallel edges to $v_{\RIGHT}$ and
	call these edges the right-internal edges.

	The relation $R_{\LEFT}$ accepts
	if and only if, for some $i\in[0,h-1]$,
	exactly $i$ left-external
	and exactly $a_{h-1-i}$ left-internal edges are selected.
	Similarly, $R_{\RIGHT}$
	accepts if and only if, for some $j\in[0,h-1]$,
	exactly $b_j$ right-internal
	and exactly $j$ right-external edges are selected.

	We claim that the above replacement
	does not change the existence of solutions.
	For this we show
	that the number of selected left-external edges
	plus the number of selected right-external edges is at most $h-1$
	for each such modification.
	Then, by the properties in \cref{lem:lbBFactorR},
	they sum to exactly $h-1$ selected edges.

	If $i$ left-external edges are selected,
	then $v'$ is incident to $a_{h-1-i}$ selected left-internal edges,
	by definition of $R_{\LEFT}$.
	As $R_{\RIGHT}$ rejects
	when $v_{\RIGHT}$ is incident to exactly $k$ selected right-internal edges
	where $k\neq b_j$ for all $j$,
	vertex $v'$ must be incident
	to $b_j$ right-internal edges for some $j$.
	By the definition of the half-induced matching,
	we get $a_{h-1-i}+b_j \in \Ex$ if $j>h-1-i$.
	Thus, some $b_{h-1-i-i'}$ with $h-1-i \ge i' \ge 0$ must be chosen.
	The relation $R_{\RIGHT}$ maps
	the $b_{h-1-i-i'}$ selected right-internal edges
	to $h-1-i-i'$ selected right-external edges.
	Thus, the gadget is incident to $i+(h-1-i-i') = h-1-i' \le h-1$ edges in total.

	As $v$ was only adjacent to complex vertices,
	we can merge the complex vertices $v_{\LEFT}$ and $v_{\RIGHT}$
	with the existing complex vertices
	and thus, also the corresponding relations.

	We analyze how the size and the pathwidth change.
	Replacing the simple vertices by the gadget
	does not change the pathwidth of the graph
	but only increases the degree of the complex vertices (due to the merging of the relations).
	Hence, $\totalDeg$ increases to at most $\totalDeg \cdot U$.
	As $U$ only depends on the set $\Ex$,
	it can be bounded by $\hat f(\max\Ex)$ for some function $\hat f$.

\begin{proof}[Proof of \cref{lem:dec:lbAntiFactorR}]
	Let $H$ be a given \BFR instance
	where $B = \{h-1\}$
	for which we apply the above construction
	to obtain the \AntiFactorR{\Ex} instance $G$.

	Since the size and the pathwidth of $G$ and $H$ are the same,
	we denote them by $n$ and $\pw$ in the following.
	From the construction we get
	$\totalDeg_G \le \totalDeg_H \cdot \hat f(\max \Ex)$.
	Assume we are given a faster algorithm for \AntiFactorR{\Ex}
	for some function $f_\Ex$.
	We run this algorithm on the instance $G$ in the following time:
	\begin{align*}
		(h-\epsilon)^{\pw + f_\Ex(\totalDeg_G)} \cdot n^{\O(1)}
		&\le (h-\epsilon)^{\pw + f_\Ex(\totalDeg_H \cdot \hat f(\max\Ex))} \cdot n^{\O(1)} \\
		&\le ((h-1)+1-\epsilon)^{\pw + f_\Ex'(\totalDeg_H)} \cdot n^{\O(1)}
	\end{align*}
	As $\Ex$ is fixed,
	$\hat f(\max \Ex)$ can be seen as part of the function $f_\Ex'$
	which is allowed to depend on $B$ and thus also on $\Ex$.
	By \cref{lem:lbBFactorR}, the running time directly contradicts SETH.
\end{proof}

%% file: cofinite/dec-rel.tex
\subsection{Replacing the Relations}
It remains to prove \cref{lem:dec:AntiFactorRToAntiFactor},
where we replace the relations by appropriate graphs.
We use the same definition of realization as in \cite{MarxSS21}.
\begin{definition}[Realization (Definition~5.1 in \cite{MarxSS21})]
  Let $R \subseteq \SetB^k$ be a relation.
  Let $G$ be a node-labeled graph with dangling edges $D=\{d_1, \dots, d_k\} \subseteq E(G)$.
  We say that graph $G$ \emph{realizes} $R$ if for all $D' \subseteq D$:
  $D' \in R$
  if and only if
  there is a solution $S \subseteq E(G)$ with $S \cap D = D'$.
\end{definition}
By assumption we have $\maxgap(\Exbar)>1$
and $0\in \Ex$.
Hence,
we can use some of the constructions from \cite{MarxSS21}
although a careful analysis is necessary
since we are now using the cofinite set $\Exbar$
and the original constructions are for finite sets.
Nevertheless, this makes the construction easier,
as we know that $\Exbar$ always contains an even and an odd number,
e.g., $\max\Ex+1$ and $\max\Ex+2$.

\begin{definition}
  We denote by \HWin[d]{S} the symmetric $d$-ary relation that only accepts
  if the Hamming weight of the input (i.e., the number of ones)
  is contained in $S$.

  To simplify notation we set $\HWeq[d]{k}=\HWin[d]{\{k\}}$
  and $\EQ{d}=\HWin[d]{\{0,d\}}$.
  We call the latter relations also equality relations.

  We write $\HWeq{k}$ for the set of all relations $\HWeq[d]{k}$
  with $d\ge 1$.
\end{definition}
\begin{lemma}\label{lem:dec:realization-hw11-eqk}
	We can realize the relations $\HWeq[1]{1}$ and $\EQ{k}$, for all $k\geq 1$,
	by simple graphs of at most $k\cdot(\max \Ex)^{\O(1)}$ vertices
  of degree at most $\max \Ex+3$.
\end{lemma}
\begin{proof}
  We use ideas from Lemma~5.4 and Item~1 of Lemma~5.8 in \cite{MarxSS21}.
  \begin{itemize}
  \item
  We first construct a \HWeq[2]{2} gadget.
  For this we take a clique with $\min \Exbar+1>1$ vertices
  and split one edge into two dangling edges.
  As we know that any degree $d \le \min\Exbar$ is forbidden,
  both dangling edges must always be selected.

  \item
  Next, we realize \HWin[1]{\{0,1\}} which is equivalent to \EQ{1}.
  For this we take a new vertex $v$
  and force $2\ceil{ (\max\Ex+1) / 2 }$ edges to it
  by using $\ceil{ (\max\Ex+1) / 2 }$ copies of a \HWeq[2]{2} gadget.
  Adding one dangling edge to $v$ gives the stated gadget,
  as the vertex $v$ already has $2\ceil{ (\max\Ex+1) / 2 } \ge \max\Ex+1$ neighbors
  and adding more neighbors does not make the solution invalid.

  \item
  For a \HWeq[1]{1} gadget we connect a \HWeq[2]{2} gadget and a \EQ{1} gadget.
  The unused edge of the \HWeq[2]{2} gadget acts as the dangling edge.

  \item
  It suffices to realize \EQ{d+1}
  where $d = \maxgap(\Exbar)>1$.
  Then, we can realize \EQ{k} for arbitrary $k$
  by connecting $\ell=\ceil{ (k-2) / (d-1) }$ copies of a \EQ{d} gadget in a path-like manner.
  To $k-\ell(d-1)-2$ dangling edges we add one \EQ{1} gadget each
  such that the \EQ{d+1} gadget has the correct arity.

  From $\maxgap(\Exbar)=d$,
  we know that there is some $a\ge 1$
  such that $[a,a+d+1] \cap \Ex = [a+1,a+d]$.
  We start with a new vertex $v$ and force $a$ edges to $v$
  using the \HWeq[1]{1} relations.
  Then, it suffices to add $d+1$ dangling edges to $v$ to realize \EQ{d+1}.
  \qedhere
  \end{itemize}
\end{proof}
We additionally need \HWeq[k]{1} gadgets which we realize next.
\begin{lemma}\label{lem:relation:hw1-decision}
  For all $k \geq 1$,
	we can realize the relations $\HWeq[k]{1}$
	by simple graphs of at most $k\cdot(\max \Ex)^{\O(1)}$ vertices
  of degree at most $\max \Ex+3$.
\end{lemma}
\begin{proof}
	It suffices to construct $\HWeq[2]{1}$ and $\HWeq[3]{1}$
	as the case $k=1$ transfers directly from the previous lemma.
	Then, we realize the other relations inductively:
	Assume we already realized \HWeq[1]{\ell} for some $\ell\ge 3$.
	We connect a \HWeq[1]{\ell} and a \HWeq[1]{3} node
	to a common \HWeq[1]{2} node.
	Replacing the relations by their realization,
	completes the construction.
	See Lemma~5.5 in \cite{MarxSS21} for more details.

  The construction of $\HWeq[2]{1}$ and $\HWeq[3]{1}$
  follows the construction from Lemma~5.8 in \cite{MarxSS21}.
	Since $\Exbar$ has a gap of size at least 2, there is some $d\geq 2$ and some $a\geq 1$ such that
	$[a,a+d+1]\cap \Exbar=\{a,a+d+1\}$.
	To construct $\HWeq[k]{1}$, for $k=2,3$, create two vertices $u$ and $v$.
	Create $k$ copies of an $\EQ{3}$ gadgets
  and for each copy connect one dangling edge to $u$, one dangling edge to $v$, and
  use the last dangling edge as a dangling edge for the \HWeq[k]{1} gadget.
	Attach $a+d$ pendant nodes with relation $\HWeq[1]{1}$ to $u$
	and attach $a-1$ pendant nodes with relation $\HWeq[1]{1}$ to $v$.
	This realizes $\HWeq[k]{1}$ because at most one of the $k$ edges
	connecting $u$ to the $\EQ{3}$ gadgets can be selected since
	$a-1+k\in {\Ex}$ and
	at least 1 of the $k$ edges connecting $v$ to the $\EQ{3}$
	gadgets has to be selected since $a+d \in {\Ex}$.
\end{proof}
Based on the results in \cite{CurticapeanM16} and their extension to \BFactor
in Theorem~5.2 of \cite{MarxSS21}, one can easily show \cref{lem:dec:realization}.
These constructions use only $\EQ{k}$ and $\HWeq[k]{1}$ gadgets
which are available to us via \cref{lem:dec:realization-hw11-eqk,lem:relation:hw1-decision}.
Moreover,
as we can realize \EQ{k} and \HWeq[k]{1}, for all $k$,
we do not need to impose any restrictions on the relations we want to realize.
\begin{lemma}\label{lem:dec:realization}
	Let $R\subseteq\{0,1\}^e$ be any relation.
  Then, we can realize $R$ by a simple graph of at most
	$f(e)\cdot(\max \Ex)^{\O(1)}$ vertices
  with degree at most $\max \Ex +3$.
\end{lemma}
With \cref{lem:dec:realization} we can prove the correctness of the second step of our reduction,
which removes the relations from the graph.
\begin{proof}[Proof of \cref{lem:dec:AntiFactorRToAntiFactor}]
  For a given \AntiFactorR{\Y} instance $H$,
  replace all complex vertices and their relations
  by the realizations guaranteed by \cref{lem:dec:realization}
  to get an \AntiFactor{(\Ex,\Y)} instance $G$.
  This step uses only vertices with set $\Ex$
  while the vertices with set $\Y$ remain unchanged.

	Observe that the size of the realizations of a relation with degree $d$
	is bounded by $f(d)$ for some function $f$ which might depend on $\Ex$.
	By this we can bound the size of $G$ by $n_G \le n_H \cdot f(\totalDeg_H)$.
	We additionally modify the path decomposition.
	We replace each complex vertex in a bag $\bag{}$
	by the vertices of the realization of the corresponding relation.
	Thus, the size of the bags increases at most by $\totalDeg_H\cdot f(\totalDeg_H)$
  which can be bounded by $f'(\totalDeg_H)$ for some appropriate $f'$.
\end{proof}

%% file: cofinite/opt.tex
\section{Lower Bounds for the Optimization Version}\label{sec:lower:opt}
For the optimization version we need to consider the minimization and the maximization version.
Observe that when assuming $0 \in \Ex$
(i.e., the empty set is \emph{not} a valid solution),
the hardness and lower bounds for these two problems follow from the decision version.
Further, the minimization version is trivial if $0 \notin \Ex$,
as the empty set is the unique smallest solution.
Thus, it suffices to focus on the maximization version where $0 \notin \Ex$,
i.e., the empty set is always a valid solution.
As the optimization version is polynomial time solvable
if all gaps are of size at most one \cite{DudyczP18},
we can assume that at least one gap has size at least two.
Surprisingly it can be shown that
if there is only one gap of size larger than two
and this gap is from $1$ to $k$ for some $k\ge 1$,
then the maximization version can be solved in polynomial time.
\begin{theorem}
	Let $\Ex =[1,k]$ for some $k \ge 1$.
	Then, \MaxAntiFactor{\Ex} can be solved in polynomial time.
\end{theorem}
\begin{proof}
	We want to find a subset of edges such that every vertex has either 0 or at least $k+1$ selected incident edges.

	The algorithm is as follows,
	where $V$ are the vertices of the given graph $G$ and $G[W]$ returns the induced graph on the vertices of $W \subseteq V$:
	\begin{enumerate}
		\item
		Let $W = V$.
		\item
		Repeatedly remove all vertices $v \in W$ with $\deg_{G[W]}(v) \le k$ from $W$.
		\item
		Return $E(G[W])$ as the solution.
	\end{enumerate}
	Clearly the algorithm runs in polynomial time,
	as each vertex can be removed at most once
	and the check of the degree can be done in time linear in the size of the graph.

	As each vertex in $W$ has degree at least $k$
	and the vertices in $V \setminus W$ have degree 0,
	the output is a valid solution.
	Observe that all vertices in $V \setminus W$ must have degree 0 for all possible solutions.
	Thus, all solutions must be a subset of the edges of the induced graph $G[W]$.
	Hence, the returned solution is indeed maximum.
\end{proof}
For the case when $0\not\in \Ex$ together with $\Ex\neq [1,k]$, for all $k\ge 1$,
we show similar lower bounds as for the decision version.
These lower bounds are again based on half-induced matchings.
\begin{theorem}[Lower Bound for Maximization Version I]
	\label{thm:opt:lbAntiFactor}
	Fix a finite set $\Ex \subseteq \SetN$ such that
	\begin{itemize}
		\item
		$\maxgap(\Exbar)>1$
		\item
		and $\Ex$ contains a half-induced matching of size $h\ge 2$.
	\end{itemize}
	For every constant $\epsilon>0$, there is no algorithm that can solve \MaxAntiFactor{\Ex} in time $\Ostar{(h-\epsilon)^\pw}$
	even if we are given a path decomposition of width $\pw$,
	unless SETH fails.
\end{theorem}
The underlying reduction is again split into two parts.
The first one is a lower bound for \MaxAntiFactorR{\Ex}
and follows the same procedure as for the decision version.
\begin{lemma}[Lower Bound for \MaxAntiFactorR{\Ex}]
	\label{lem:opt:lbAntiFactorR}
	Let $\Ex \subseteq \SetN$ be a fixed set
	such that it contains a half-induced matching of size $h\ge 2$.

	Let $f_\Ex\from \SetN\to\SetR^+$ be an arbitrary function
	that may depend on $\Ex$.

	For every constant $\epsilon>0$,
	there is no algorithm that can solve \MaxAntiFactorR{\Ex}
	in time $\Ostar{(h-\epsilon)^{\pw+f_{\Ex}(\totalDeg)}}$,
	where $\totalDeg = \max_{\text{bag } \bag{}} \sum_{v \in \bag{}\cap V_C} \deg(v)$,
	even if we are given a path decomposition of width $\pw$,
	unless SETH fails.
\end{lemma}
As for the decision version,
the second step follows the ideas from \cite{MarxSS21}
but we have to modify the constructions to take care of the cofinite set
of allowed degrees.
We define \MaxAntiFactor{(\Ex,\Y)} in the natural way,
as we did for the decision version.
\begin{lemma}
	\label{lem:opt:AntiFactorRToAntiFactor}
	Fix a finite set $\Ex \subseteq \SetN$ such that
	$0 \notin \Ex$
	and $\maxgap(\Exbar) > 1$
	but $\Ex \neq [1,k]$, for all $k\ge 1$.
	Let $\Y \subseteq\SetN$ be arbitrary.

	There is a many-one reduction from \MaxAntiFactorR{\Y}
	to \MaxAntiFactor{(\Ex,\Y)}
	such that pathwidth increases by at most $f(\totalDeg)$,
	the size by a factor of $f(\totalDeg)$,
	and the degree to at least $\max \Ex+2$,
	where $\totalDeg = \max_{\text{bag } \bag{}} \sum_{v \in \bag{}\cap V_C} \deg(v)$.
\end{lemma}
Using \cref{lem:opt:lbAntiFactorR,lem:opt:AntiFactorRToAntiFactor},
we prove the lower bound when parameterizing by the set.
\begin{proof}[Proof of \cref{thm:opt:lbAntiFactor}]
	The proof combines \cref{lem:opt:lbAntiFactorR,lem:opt:AntiFactorRToAntiFactor}.
	The analysis of the size and the running time is the same
	as for the decision version in \cref{thm:dec:lbAntiFactor}.
\end{proof}
Additionally, we show a lower bound when parameterizing by the number of excluded degrees.
Again this proof follows the one from the decision version.
\begin{theorem}[Lower Bound for Maximization Version II]
	\label{thm:opt:lbAntiFactorBySize}
	For all $\ex\ge 3$ and $\epsilon>0$,
	\MaxAntiFactorSize{\ex} cannot be solved in time $\Ostar{(\ex+1-\epsilon)^{\pw}}$
	on graphs given with a path decomposition of width $\pw$,
	unless SETH fails.
\end{theorem}

\input{cofinite/opt-rel}

%% file: cofinite/opt-rel.tex
\subsection{Replacing the Relations}
We can again use some of the machinery
from the lower bound for \MaxBFactor in \cite{MarxSS21}.
As for the decision version, we need to carefully check the construction
as the set of allowed degrees is now cofinite.

We start by defining the realization of a relation for the maximization version
which varies slightly from the definition for the decision version,
as we cannot rule out the existence of solutions in all cases.
\begin{definition}[{Realization (cf.\ Definition~6.1 in \cite{MarxSS21})}]
	Let $R\subseteq\{0,1\}^k$ be a relation.
	Let $G$ be a graph with dangling edges $D=\{d_1,\ldots,d_k\}$.
	We say that \emph{$G$ realizes $R$ with penalty $\beta$}
	if we can
	efficiently construct/find an $\alpha>0$ such that for every $D'\subseteq D$:
	\begin{itemize}
		\item
		If $D'\in R$, then there is a solution
		$S\subseteq E(G)$ such that $S\cap D=D'$ and $\abs{S}=\alpha$.
		\item
		If $D'\not\in R$, then, for every solution
		$S\subseteq E(G)$ such that $S\cap D=D'$, we have $\abs{S}\leq \alpha-\beta$.
	\end{itemize}
\end{definition}
The construction of the most basic building block,
i.e., forcing edges to be in the solution, still applies
since there is a gap of size at least 2 in $\Exbar$.
\begin{lemma}[Lemma~6.6 in \cite{MarxSS21}]\label{lem:relation:hw2-optimization}
	There is a function $f\from\SetN\to\SetN$ such that the following holds.
	We can
	realize $\HWeq[2]{2}$
	(with distinct portal vertices)
	with arbitrary penalty $\beta$
	by simple graphs
	using $f(\beta)$ vertices with set $\Ex$ and
	degree at most $\max \Ex+1$.
\end{lemma}
We use this as a building block
to realize $\HWeq[k]{1}$ and $\EQ{k}$ relations,
following which we argue that all relations can be realized.
\begin{lemma}\label{lem:relation:hw1-eq-optimization}
	There is a function $f\from\SetN\times\SetN\to\SetN$ such that the following holds.
	For any $k\geq 1$, we can
	realize
	$\HWeq[k]{1}$
	and $\EQ{k}$
	with arbitrary penalty $\beta$
	by simple graphs
	using $f(\beta,k)$ vertices of
	degree at most $\max \Ex+3$.
\end{lemma}
\begin{proof}
	We know $\maxgap(\Exbar)>1$.
	Let $a\ge 0$ be such that $[a,a+d+1] \cap \Ex = [a+1,a+d]$.
	We realize the relations in several steps
	since some of our
	constructions for $\HWeq[k]{1}$ depend on $\EQ{k'}$ gadgets and vice versa.

	\begin{itemize}
		\item
		\HWin[1]{\{0,1\}}:%
		\footnote{
		Strictly speaking this gadget does not realize \HWin[1]{\{0,1\}}
		as the solution size varies
		depending on whether the dangling edge is selected or not.
		We write \EQ{1} if we use a gadget that is indeed realizing this relation.
		The same holds for such gadgets of higher degree.
		}
		Force at least $\max\Ex+1$ edges to a new vertex $v$
		by $\ceil{ (\max\Ex+1) / 2 }$ copies of a $\HWeq[2]{2}$ node.
		Add one dangling edge to this vertex $v$
		and replace all nodes by their realizations.
		Note that the relation \HWin[1]{\{0,1\}} is always satisfied and hence,
		we never need the penalty of the \HWeq[2]{2} nodes,
		i.e., setting it to $1$ is already sufficient.

		\item $\HWeq[1]{1}$:
		We connect a \HWeq[2]{2} gadget to a \HWin[1]{\{0,1\}} gadget from above.
		Replace both relations by their realizations
		with penalty $\beta+1$.
		It is clear that unless the dangling edge is selected, we have a penalty  of at least $\beta$
		because the relation \HWeq[2]{2} is not satisfied.

		\item
		$\HWin[d+1]{\{0,d+1\}}$:
		Start with a new vertex and force $a$ edges to it
		by making it adjacent to $a$ copies of a \HWeq[1]{1} node.
		Further, add $d+1$ dangling edges to the new vertex
		and replace all nodes by their realization with a penalty of $\beta+2$.

		\item $\HWeq[2]{1}$:
		Create two complex vertices $v_1$ and $v_2$ with relation $\HWin[d+1]{\{0,d+1\}}$.
		For all $i\in[d-1]$ (we know $d-1 \ge 1$),
		create a new vertex $u_i$,
		connect it to $v_1$ and $v_2$, and
		attach $a+d$ vertices with relation $\HWeq[1]{1}$ to $u_i$.
		Unless the $\HWeq[1]{1}$ nodes incur a penalty, each $u_i$
		needs at least one of their edges
		to $v_1$ or $v_2$ to be
		selected because $a+d\in \Ex$.

		Now we have two cases.
		If $a>0$, then connect $v_1$ and $v_2$
		to a new vertex $u_0$ connected to $a-1$ vertices with relation $\HWeq[1]{1}$.
		Vertex $u_0$ needs at most one of its edges
		to the $\HWin[d+1]{\{0,d+1\}}$ node to be
		selected because $a\not\in \Ex, a+1\in\Ex$.

		If $a=0$, then we know that we have another gap between
		$a'$ and $a'+d'+1$
		for some $a'>0$ and $d' \ge 1$.
		Connect $v_1$ and $v_2$ to
		a new vertex $u_0$ which is connected to $a'-1\geq 0$ vertices with relation $\HWeq[1]{1}$.
		Similar arguments apply as before.

		Finally, we replace all relations by their realizations with a
		penalty of $\beta+2$.

		\item $\HWeq[3]{1}$: When $a>0$, we can realize this in the same way as
		$\HWeq[2]{1}$ by using the
		construction with three $\HWin[d+1]{\{0,d+1\}}$ vertices $v_i$ instead of two.

		When $a=0$, we first do the following: Connect a vertex to $d-1$ nodes with
		relation $\HWeq[1]{1}$ and $3$ nodes with relation $\HWeq[2]{1}$.
		Add one dangling edge to each of the $\HWeq[2]{1}$ gadgets.
		When all nodes are replaced with realizations of penalty $\beta+3$,
		this realizes a $\HWeq[3]{1}$ or a $\HWin[3]{\{0,1\}}$ depending on whether
		$d+2\in \Ex$. In the former case, we are done.

		In the latter case, we reuse the construction for $\HWeq[2]{1}$ once more.
		Similarly as for the case $a > 0$,
		we use three vertices with relation $\HWin[d+1]{\{0,d+1\}}$ instead of two.
		Then, replace the vertex $u_0$ and its attached $\HWeq[1]{1}$ nodes
		by the realization of $\HWin[1]{\{0,1\}}$ which we created above.
		As before, we replace all vertices by their realization
		with penalty $\beta+3$.

		\item
		\HWeq[k]{1}:
		By the previous items, the cases $k=1,2,3$ are already handled.
		For larger $k$ we use the same inductive construction
		as the one from Lemma~5.5 in \cite{MarxSS21}.
		We connect a \HWeq[k]{1} gadget and a \HWeq[3]{1} gadget
		to a common \HWeq[2]{1} gadget.
		The degree of this gadget is $k+1$.
		If either the \HWeq[3]{1} or the \HWeq[k]{1} gadget
		are incident to exactly one edge,
		then the other gadget is also incident to one edge
		because of the shared \HWeq[2]{1} gadget.
		Recursively replace all gadget by their realization
		with penalty $\beta+1$
		as each additionally selected edge forces one more relation to be invalid.

		\item $\EQ{k}$:
		Create vertices $v_1,\dots,v_k$ and $u_1,\dots,u_k$,
		all with relation \HWin[d+1]{\{0,d+1\}}.
		We connect $v_i$ to $v_{i+1}$
		and $u_i$ to $u_{i+1}$, for all $i \in [k-1]$.
		Moreover we connect $v_i$ and $u_i$ to $(d+1)-3 \ge 0$
		shared \HWeq[2]{1} nodes by on edge each.
		The vertices $v_1$ and $u_1$
		are connected by one more \HWeq[2]{1} node.
		Similarly for $v_k$ and $u_k$.
		Then, we make one \HWin[1]{\{0,1\}} node adjacent to each $u_i$
		and add one dangling edge to each $v_i$.
		Replacing all relations by their realization
		with penalty $\beta+2k(d+1)$ completes the construction.

		If zero or $k$ dangling edges are selected,
		then we can select all edges incident to the $u_i$s or $v_i$s, respectively.
		Since $v_1$ and $u_1$ are always connected by a \HWeq[2]{1} node,
		\footnote{
		If $d=2$, the nodes $v_i,u_i$ are not be connected for all $i\in [2,k-1]$.
		}
		we cannot create a larger solution without violating any relation.
		As we can extend this selection to the realizations,
		the size of the solution is in both cases the same.

		Now, assume that $\ell \in [k-1]$ dangling edges are selected.
		Then, for any solution, we know that there must be some $v_i$ and $v_{i+1}$
		such that at least one of these two relations is not satisfied.
		Hence, we loose a factor of $\beta+2k(d+1)$.
		Further observe that the graph contains $2k(d+1)$ edges.
		Hence, we loose a factor of $\beta$ compared to the optimal solution.
		For this to work we crucially need that the nodes with relation
		\HWin[d+1]{\{0,d+1\}} and \HWin[1]{\{0,1\}} are no realizations
		so that the dangling edges are not included in the size of the solutions.
		\qedhere
	\end{itemize}
\end{proof}
As for the decision version,
it suffices to realize \EQ{k} and \HWeq[k]{1}.
Then, we can use the same construction as in Theorem~6.2 in \cite{MarxSS21}
to realize arbitrary relations.
As we can realize \EQ{k} and \HWeq[k]{1}, for all $k$ and not just specific $k$,
we do not need any constraints for the relation.
\begin{lemma}
	There is a function $f\from\SetN\times\SetN\to\SetN$ such that the following holds.
	Let $R\subseteq\{0,1\}^e$ be a relation.
	Then, we can realize $R$
	with arbitrary penalty $\beta$
	by simple graphs
	using $f(e,\beta)$ vertices of
	degree at most $\max \Ex+3$.
\end{lemma}
With this lemma we can prove the second step of our reduction,
namely \cref{lem:opt:AntiFactorRToAntiFactor}.
\begin{proof}[Proof of \cref{lem:opt:AntiFactorRToAntiFactor}]
	We follow the outline of the proof for the decision version
	from \cref{lem:dec:AntiFactorRToAntiFactor}.
	For the realizations of the relations
	we set the penalty to be twice the degree of the relation.
	Thus, the size depends only on the degree of the relation
	(and the set $\Ex$).

	Consider the realization of a complex vertex of degree $\delta$.
	Assume we try to create a larger solution
	by additionally selecting up to $\delta$ incident edges.
	Thus, if this selection does not satisfy the relation
	(otherwise the total size does not change by definition)
	we lose a factor of $2\delta$ by the choice of the penalty.
	Hence, this selection does not increase the size of the solution
	and, therefore, all relations must be satisfied.
	See Lemma~6.3 in \cite{MarxSS21} for a detailed proof.
\end{proof}

%% file: cofinite/count.tex
\section{Lower Bounds for the Counting Version} \label{sec:lower:count}

In this section we prove the two lower bounds for the counting version.
While the lower bound for the decision and maximization version of  \AntiFactor{\Ex} rely on half-induced matching,
we avoid this dependence for \CountAntiFactor{\Ex}
by using interpolation techniques.
This allows us to show a tight lower bound
compared to the running time of the algorithm from \cref{thm:algo:paraByList}.
For the case when $\Ex = \{0\}$, that is \countECover,
we show a completely independent but also tight lower bound
	in \cref{sec:lower:edge-cover}.

We also parameterize by the size $\ex$ of the set of forbidden degrees.
We design a new construction to prove the \sharpW{1}-hardness of
\CountAntiFactorSize{\ex}, even if $\ex = 1$, when parameterizing by treewidth.
Hence, \CountAntiFactorSize{\ex} is most likely not fixed-parameter tractable.

Both bounds use the same two-step approach as
for the decision and optimization version;
we first show the hardness of an intermediate problem
which uses arbitrary relations
and then remove these relations
by a chain of reductions to obtain the actual lower bounds.

\subparagraph*{Parameterizing by the Maximum of the Set.}
We first show a lower bound for the intermediate problem \CountAntiFactorR{\Ex},
which is the counting version of \AntiFactorR{\Ex}.
Recall that we define
$\totalDeg = \max_{\text{bag } \bag{}} \sum_{v \in \bag{}\cap V_C} \deg(v)$
for a graph which contains relations.

\begin{lemma}[Lower Bound for \CountAntiFactorR{\Ex}]
	\label{lem:count:lbAntiFactorR}
	Let $\Ex \subseteq \SetN$ be a fixed, non-empty and finite set.
	%

	For every constant $\epsilon>0$, there is no algorithm that can solve \CountAntiFactorR{\Ex}
	in time $\Ostar{(\max \Ex+2-\epsilon)^{\pw}}$
	even if we are given a path decomposition of width $\pw$
	and $\totalDeg \in \O(\max\Ex)$,
	unless \#SETH fails.
\end{lemma}
We make use of \cref{lem:count:AntiFactorRToAntiFactor} to remove the relations.
We extend the definition of \AntiFactor{(\Ex,\Y)}
in the natural way
to the counting version \CountAntiFactor{(\Ex,\Y)}.
\begin{lemma}
	\label{lem:count:AntiFactorRToAntiFactor}
  Let $\Ex \subseteq \SetN$ be a finite set
  such that $\Ex \not\subseteq \{0\}$.
	Let $\Y \subseteq \SetN$ be arbitrary
	(possibly be given as input).

	There is a Turing reduction from \CountAntiFactorR{\Y}
	to \CountAntiFactor{(\Ex,\Y)}
	increasing the size from $n$ to $n \cdot f(\max\Ex)$,
	decreasing $\totalDeg$ to zero,
	and increasing $\pw$ to $\pw+\totalDeg \cdot f(\max \Ex)$.
\end{lemma}
Combining \cref{lem:count:lbAntiFactorR,lem:count:AntiFactorRToAntiFactor},
we can prove the first lower bound for the counting version.
\begin{theorem}[Lower Bound for Counting Version I]
  \label{thm:count:lbAntiFactor}
  Let $\Ex \subseteq \SetN$ be a finite and fixed set
  such that $\Ex \not\subseteq \{0\}$.
	For every constant $\epsilon>0$, there is no algorithm that can solve \CountAntiFactor{\Ex}
	in time $\Ostar{(\max \Ex+2-\epsilon)^\pw}$
	even if we are given a path decomposition of width $\pw$,
	unless \#SETH fails.
\end{theorem}
\begin{proof}
  For a given \CountAntiFactorR{\Ex} instance $H$
	with $\totalDeg_H \in \max\Ex$
  we apply \cref{lem:count:AntiFactorRToAntiFactor} with $\Ex=\Y$
  to obtain a \CountAntiFactor{\Ex} instance $G$.

  We know $n_G \le n_H \cdot f(\max \Ex)$,
	and $\pw_G \le \pw_H + \totalDeg_H \cdot f(\max \Ex)$.
  Assume the claimed algorithm exists, for some $\epsilon >0$,
	and run it on this instance:
  \begin{align*}
    (\max\Ex+2-\epsilon)^{\pw_G} \cdot n_G^{\O(1)}
    &\le (\max\Ex+2-\epsilon)^{\pw_H + \totalDeg_H \cdot f(\max \Ex)} \cdot (n_H \cdot f(\max \Ex))^{\O(1)} \\
    &\le (\max\Ex+2-\epsilon)^{\pw_H + \totalDeg_H \cdot f(\max \Ex) + f'(\max \Ex)} \cdot n_H^{\O(1)}
  \end{align*}
	By our assumption, we have $\totalDeg_H \in \O(\max\Ex)$.
	Since $\Ex$ is fixed,
	we can treat $\max\Ex$ as a constant.
	Hence, the term
	$\totalDeg_H \cdot f(\max \Ex) + f'(\max \Ex)$ in the exponent
	contributes only a constant.
	Thus, the running time directly contradicts \#SETH
	by \cref{lem:count:lbAntiFactorR}.
\end{proof}

\subparagraph*{Parameterizing by the Size of the Set.}
If we do not fix the set $\Ex$ but only the \emph{size} of the set,
the decision and optimization version of \AntiFactorSize{\ex}
are still FPT parameterized by treewidth.
For \CountAntiFactorSize{\ex} the following result
conditionally rules out such algorithms.
\begin{lemma}
	\label{lem:count:lbCount1AntiFactorR}
  There exists a constant $c$
  such that there
  is no $\O(n^{p-c})$ algorithm for \CountAntiFactorSizeR{1} on
  $n$-vertex graphs
	with $\totalDeg\in\O(1)$,
	even if only one set is used
	and
	we are given a path decomposition of width $p$,
	unless \#SETH is false.
\end{lemma}
Combined with \cref{lem:count:AntiFactorRToAntiFactor} to remove the relations,
we get the hardness result in \cref{thm:count:lbAntiFactorBySize}.
\begin{theorem}[Lower Bound for Counting Version II]
	\label{thm:count:lbAntiFactorBySize}
  There exists a constant $c$
  such that there is no $\O(n^{p-c})$ algorithm
	for \CountAntiFactorSize{1}
	on $n$-vertex graphs,
  even if we are given a path decomposition of width $p$,
  unless \#SETH is false.
\end{theorem}
\begin{proof}
	Let $G$ be a given \CountAntiFactorSizeR{1} instance
	which uses just one set $\Y \subseteq \SetN$ with $\abs{\Y}=1$.
	Use \cref{lem:count:AntiFactorRToAntiFactor} with $\Ex=\{2\}$
	to transform $G$
	into a \CountAntiFactorSize{1} instance $H$.

	Now assume the claim is false,
	i.e., for all $c>0$, there is a $\O(n_H^{\pw_H-c})$ algorithm.
	Run this algorithm on the new instance $H$:
	\begin{align*}
		\O(n_H^{\pw_H-c})
		\le \O(n_G^{\pw_G+\O(1)-c})
		\le \O(n_G^{\pw-(c-\O(1))})
		\le \O(n_G^{\pw-c'}).
	\end{align*}
	Hence, for all $c'$, there is an $\O(n_G^{\pw_G-c'})$ algorithm
	for \CountAntiFactorSizeR{1}.
	This contradicts \#SETH by \cref{lem:count:lbCount1AntiFactorR}.
\end{proof}
We prove \cref{lem:count:lbCount1AntiFactorR} by a
reduction from the \sharpW{1}-hard problem
\textsc{Counting Colorful Hitting $k$-Sets}.
Hence,
the following result holds by applying \cref{lem:count:AntiFactorRToAntiFactor} as before.
\begin{theorem}
	\label{thm:count:w1}
  \CountAntiFactorSize{1} is \sharpW{1}-hard.
\end{theorem}

\input{cofinite/count-lb}
\input{cofinite/count-w1}
\input{cofinite/count-rel}

%% file: cofinite/count-lb.tex
\subsection{High-level Construction for SETH Lower Bound}
\label{sec:lower:count:seth}

\begin{figure}
	\centering
	\input{img/count/cofinite-make_cofinite}
	\caption{
  The reduction in the proof of \cref{lem:count:lbAntiFactorR},
  i.e., the lower bound of \CountAntiFactorR{\Ex},
	follows from a chain of reduction
  where, for some $\ex \ge 0$, we have $B = \{\ex+1\}$
  and $\Ex \subseteq \SetN$ is a finite and non-empty set with $\max\Ex = \ex$.
	The combined reduction is shown by a dotted line.
  \\  
  For each problem on the left side a lower bound is formally proved,
  while for the problems on the right only reductions are stated.
	}
	\label{fig:count:lower:chain}
\end{figure}
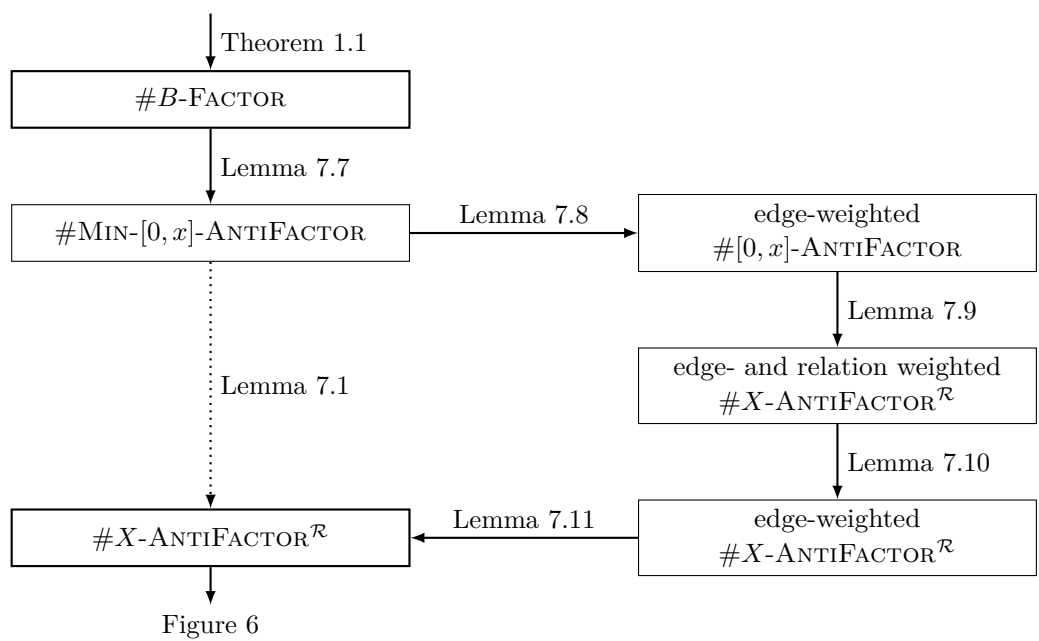

We show the hardness of \CountAntiFactorR{\Ex},
i.e., \cref{lem:count:lbAntiFactorR},
by a reduction from \CountBFactor.
Recall, for the decision and optimization version
we added some gadgets between the (simple) vertices and their neighbors.
For the counting version we use a different approach
which only attaches gadgets to the (simple) vertices
without altering the remaining graph structure.
As this modification is not directly possible,
the reduction splits into multiple steps.
See \cref{fig:count:lower:chain} for an illustration of these steps.
\begin{itemize}
  \item
  The first step makes the transition from a finite set of allowed degrees
  (i.e., the \CountBFactor problem)
  to a simple cofinite set of allowed degrees
  (i.e., \CountAntiFactor{[0,\ex]} for some $\ex\ge0$).
  Formally, we show in \cref{lem:count:lbSimpleMinAFR} a lower bound for
  \CountMinAntiFactor{\Ex} with $\Ex=[0,\ex]$ and $\ex\ge0$.

  The \CountMinAntiFactor{\Ex} problem is the variant of \CountAntiFactor{\Ex}
  where we only count solutions of \emph{minimum} size.

  \item
  In the second step we remove the requirement of counting only minimum solutions.
  In \cref{lem:count:minSimpleAFRtoSimpleAFR}
  we use interpolation techniques to reduce from \CountMinAntiFactor{\Ex}
  to the \emph{edge-weighted}
  \CountAntiFactor{\Ex} problem where $\Ex=[0,\ex]$ with $\ex\ge0$.

  For the edge-weighted version of the problem
  the edges are assigned weights by which they contribute to the solution
  if they are selected.
  The weight of a solution is the product
  of the weights of the selected edges.
  The output of the problem is the sum of all weighted solutions.
  The unweighted problem can be seen as assigning weight 1 to all edges.

  \item
  As a third step we transition from simple cofinite sets of allowed degrees
  to arbitrary cofinite sets of allowed degrees.
  This step introduces relations,
  a technique we have already seen in earlier problems.
  Formally, \cref{lem:count:simpleAFRtoRelWeightedAFR}
  shows a reduction from the problem of the previous step
  to the \emph{edge- and relation-weighted} \CountAntiFactorR{\Ex} problem. 

  In this problem each accepted input of a relation is assigned a (rational) weight.
  Then, an accepted input contributes by this weight to the solution.
  The weight of the solution is the product of the weights of the relations
  multiplied by the weights of the selected edges.
  The problem without relation-weights can be obtained
  by assigning weight 1 to all accepted inputs
  and weight 0 to all rejected inputs.

  \item
  The last two steps are dedicated to removing the weights
  from the relations and from the edges.
  We remove the relation-weights
  in \cref{lem:count:VertexWeightedToEdgeWeightedAntiFactorR}
  by using additional edge-weights.

  \item
  The last step in \cref{lem:count:WeightedAntiFactorRToAntiFactorR}
  removes these edge-weights by an appropriate Turing-reduction
  combined with known techniques for the interpolation
  of multivariate polynomials.
\end{itemize}

\subparagraph*{The Holant Framework.}
We use the Holant framework (cf.\ \cite{CaiHL10,CaiLX11,GuoL17,HuangL16,KowalczykC16,Lu11})
to contextualize several versions of \CountAntiFactorSize{}.
The framework was already extensively used in proving the \#SETH lower bounds
for \textsc{Counting Perfect Matchings} \cite{CurticapeanM16}
and \textsc{Counting General Factors} \cite{MarxSS21}.

In the Holant framework, we are given a signature graph
$\Omega=(V,E)$, where every edge $e\in E$ has a weight $w_e$.
Every vertex $v\in V$ is labeled with a
signature $f_v\from\{0,1\}^{I(v)}\to \mathbb{Q}$, where $I(v)$ is the incidence
vector of edges incident to $v$.
A \emph{solution} is a subset of edges such that the
signature for each vertex is non-zero.
The weight of a solution is the product of the weights of all edges
in the solution multiplied by the product of the signatures of all the vertices.
Then, the Holant of $\Omega$ is defined as the sum
of the weights of all solutions.
\[
\Hol{\Omega}=\sum_{x\in \{0,1\}^{E(\Omega)}} \prod_{e\in x} w_e\prod_{v\in V(\Omega)}f_v(x|_{I(v)}).
\]
The Holant problem is easily seen to be a weighted generalization of
the counting versions of \GenFac and \AntiFactorSize{}.
For example,
the problem \CountAntiFactor{\Ex} is a Holant problem
on unweighted graphs where every vertex has the following symmetric relation.
\[
f(z) = \begin{cases}
	1 & \text{ if } \hw(z)\not\in \Ex\\
	0 & \text{ if } \hw(z)\in \Ex
\end{cases}
\]
where $\hw(\cdot)$ is the Hamming weight operator. We call
vertices with such functions to be $\HWin{\Exbar}$ nodes.

For relations $R_1, \ldots, R_k$,
we define $\Hol{R_1,\ldots,R_k}$ to be the set of Holant problems
where every edge is unweighted
and every vertex has signature $R_j$, for some $j\in [k]$.
By an abuse of notation also let $R_j$ be a \emph{family} of
relations. For example, we may use $\Hol{\HWeq{1}}$ when
every vertex has relation \HWeq[k]{1}, for some $k$.

The \emph{edge- and relation-weighted} version of \CountAntiFactorR{\Ex}
corresponds to a variant of the Holant problem
where we have the signature $\HWin{\Exbar}$ for all simple vertices
not having a specified signature.

The \emph{edge-weighted} version of \CountAntiFactorR{\Ex}
corresponds to a variant of the Holant problem
where we additionally require that the value of the signatures is either 0 or 1,
i.e., they only accept or reject.

\subparagraph*{The Chain of Reductions.}
We prove the lower bound in \cref{lem:count:lbSimpleMinAFR}
by a reduction from \CountBFR where $B=\{\max\Ex+1\}$.
Observe that a direct reduction to \CountAntiFactor{\Ex} is non-trivial
although $B \subseteq \Exbar$:
Let $H$ be a given \CountBFR instance.
When treating $H$ as a \CountAntiFactorR{\Ex} instance $G$,
each solution of $H$
is also a solution for $G$
because of our choice of $B$.
However, the converse is not true.
Indeed, as $\Ex$ is finite (and thus the set of allowed degree is cofinite),
the degree of the solutions for $G$
can be larger than $\max\Ex+1$.
We avoid this issue by only counting solutions which have minimum size.
Then, no vertex can have degree larger than $\max\Ex+1$ in the solution.
\begin{lemma}
	\label{lem:count:lbSimpleMinAFR}
	Let $\Ex = [0,\ex-1]$ be a fixed set for some $\ex \ge 1$.
	For every constant $\epsilon>0$,
	there is no algorithm that can,
	even if we are given a path decomposition of width $\pw$,
	solve \CountMinAntiFactor{\Ex}
	in time $\Ostar{(\ex+1-\epsilon)^{\pw}}$,
	unless \#SETH fails.
\end{lemma}
\begin{proof}
  Assume the claimed algorithm for \CountMinAntiFactor{\Ex} exists.
  We use this algorithm to give a faster algorithm for \CountBFactor
  where $B \deff \{\max\Ex+1\} = \{\ex\}$
  to contradict \#SETH by \cref{thm:factor-tw-results}.

  Let $H$ be an instance of \CountBFactor.
  Use \cref{thm:factor-tw-results} to check if $H$ has a solution,
  i.e., an $\ex$-factor.
  If $H$ has no such solution, then return $0$ as solution.
  Otherwise,
  we replace all vertices in $H$ with set $B$ (of allowed degrees)
  by vertices with the set $\Ex$ (of forbidden degrees)
  to obtain an instance $G$ of \CountMinAntiFactor{\Ex}.
  Then, we apply the algorithm for \CountMinAntiFactor{\Ex}
  and return the number of solutions.

  We first analyze the runtime of the algorithm for \CountBFactor.
  By \cref{thm:factor-tw-results},
  it can be checked in time 
  $\Ostar{(\max B +1)^{\pw}}=\Ostar{(\ex+1)^\pw}$
  if $H$ has a solution.
  Given $H$, the construction of instance $G$
  takes time polynomially in the size of $H$.
  Hence, the claimed algorithm for \CountMinAntiFactor{\Ex}
  directly contradicts \cref{thm:factor-tw-results}
  assuming the above algorithm is correct.

  It remains to argue that the above procedure correctly solves \CountBFactor.
  By definition, we have $B \subseteq \SetN\setminus\Ex$
  and thus, every solution for $H$ is also a solution for $G$.
  Hence, we get
  \[
    \#\operatorname{Sol}(H) \le \#\operatorname{Sol}(G)
  \]
  To show the converse of this inequality,
  we can assume that $H$ has a solution $S_0$
  (the algorithm checked this initially).
  Consider a minimum solution $S$ of $G$
  that is not an $\ex$-factor (the other case is trivial).
  Since $\Ex=[0,\ex-1]$,
  there must be some vertex $v\in V(G)$
  such that $\deg_S(v) > \ex$.
  We get
  \[
    2 \abs{S} = \sum_{v\in V(G)} \deg_S(v) > \abs{V(G)} \cdot \ex
    .
  \]
  Recall, we assumed that $H$ has a solution $S_0$
  and moreover, $S_0$ is an $\ex$-factor.
  Using that all vertices are incident to exactly $\ex$ selected edges in $S_0$,
  we have
  \[
    2\abs{S_0}=\sum_{v\in V(H)} \deg_{S_0}(v)=\abs{V(H)} \cdot \ex
    .
  \]
  Hence, $S$ is of larger size than $S_0$
  which contradicts the assumption that $S$ is a minimum solution for $G$.

  Therefore, there is no vertex $v$ in $S$ with $\deg_S(v) > \ex$
  which proves that $S$ is a solution for $H$
  and we get
  \[
    \Slns(G)
    \le
    \Slns(H)
  \]
  which concludes the proof.
\end{proof}
As a next step we remove the restriction of counting only minimum solutions.
For this we introduce edge-weights
and compute the (weighted) number of solutions.
Then, we recover the number of solutions for the original instance
by polynomial interpolation techniques.

\begin{lemma}
  \label{lem:count:minSimpleAFRtoSimpleAFR}
  Let $\Ex =[0,\ex-1] \subseteq \SetN$ be a finite and non-empty set
  for some $\ex \ge 1$.
  There is a Turing-reduction from \CountMinAntiFactor{\Ex}
  to edge-weighted \CountAntiFactor{\Ex}
  with a single edge-weight
  running in time $n^{\O(1)}$
  without changing the pathwidth of the graph.
\end{lemma}
\begin{proof}
  Let $G$ be an instance of \CountMinAntiFactor{\Ex}.
  For a positive integer $w$,
  we define the instance $G_w$ of edge-weighted \CountAntiFactor{\Ex}
  as the modification of $G$, where we put weight $w$
  on every edge of $G$.

  For all $i\in [0,m]$, where $m$ denotes the number of edges of $G$,
  we denote by $a_i$ the number of solutions of $G$
  with respect to \CountAntiFactor{\Ex}
	that have size exactly $i$.
  Note, $a_i$ includes solutions that do not have minimum size.

  With this notation, we can rewrite the number of solutions of $G_w$ as
  \[
    \Slns(G_w) = \sum_{i=0}^m a_i w^i
    .
  \]
  If we treat $w$ as a variable,
  $\Slns(G_w)$ can be seen as a polynomial in $w$
  of degree at most $m$.

  We use the oracle-access to \CountAntiFactor{\Ex}
  to compute $\Slns(G_w)$ for $m+1$ different values of $w$.
  Then we can easily, i.e., in time polynomially in $m$,
  recover the coefficients of $\Slns(G_w)$.
  With $i^* = \min \{ i \in [0,m] \mid a_i\neq 0\}$
  we get that $G$ has $a_{i^*}$ solutions
  with respect to \CountMinAntiFactor{\Ex}.

  Since the construction of $G_w$ can be done in polynomial time
  and recovering the $a_i$s also takes polynomial time,
  this finishes the proof.
\end{proof}
The next step now drops the requirement on $\Ex$ to be a contiguous interval.
Instead arbitrary finite sets are allowed,
which results in arbitrary cofinite sets of allowed degrees.
However, to remove this constraint we have to make use of (weighted) relations.

\begin{lemma}
  \label{lem:count:simpleAFRtoRelWeightedAFR}
  Let $\Ex \subseteq \SetN$ be a fixed, non-empty and finite set.
  There is a many-one reduction from
  edge-weighted \CountAntiFactor{[0,\max\Ex]} with $r$ weights
  to edge- and relation-weighted \CountAntiFactorR{\Ex}
  with at most $r+\max\Ex+2$ weights
  such that
  the size increases by a constant factor,
  $\totalDeg$ of the resulting instances is $\max\Ex+1$
  and pathwidth increases only by 1.
\end{lemma}
\begin{figure}
	\begin{subfigure}[t]{0.25\textwidth}
		\centering
		\includegraphics[page=3]{img/count/cofinite-main-gadget-construction}
		\caption{The simple vertex before the modification.}
		\label{fig:count:BFactorRToWeightedAntiFactorR:original}
	\end{subfigure}\hfill
	\begin{subfigure}[t]{0.3\textwidth}
		\centering
		\includegraphics[page=1]{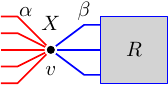}
		\caption{
    The gadget resulting from \cref{lem:count:simpleAFRtoRelWeightedAFR}.
    }
		\label{fig:count:BFactorRToWeightedAntiFactorR}
	\end{subfigure}\hfill
	\begin{subfigure}[t]{0.4\textwidth}
		\centering
		\includegraphics[page=2]{img/count/cofinite-main-gadget-construction}
		\caption{The modifications of the complex vertices from \cref{lem:count:VertexWeightedToEdgeWeightedAntiFactorR}.
		The vertices $u_i$ are assigned relation $\HWin[1]{\{0,1\}}$.
		}
		\label{fig:count:VertexWeightedToEdgeWeightedAntiFactorR}
	\end{subfigure}
	\caption{
  The modifications of the vertices
  in the different steps of the reductions
  (cf.\ \cref{lem:count:simpleAFRtoRelWeightedAFR,%
  lem:count:VertexWeightedToEdgeWeightedAntiFactorR}).
  The outer edges are shown in red,
  the inner edges are shown in blue,
  and the weighted edges in black.
  }
	\label{fig:count:gadget}
\end{figure}
For the reduction we construct a gadget
and attach to each (simple) vertex a copy of this gadget.
The gadget is such that each vertex
is incident to at least $\max\Ex+1$ selected edges from the original graph.
We define the relation-weights such that all other solutions cancel out
and hence, all solutions are also valid for \CountAntiFactor{[0,\max\Ex]}.
We exploit the fact that the set $\Ex$
does not allow certain combinations of selected incident edges
but once a vertex is incident to at least $\max\Ex+1$ selected edges,
all states are essentially equivalent (with respect to $\Ex$).

\begin{proof}[Proof of \cref{lem:count:simpleAFRtoRelWeightedAFR}]
	We start with an edge-weighted \CountAntiFactor{[0,\max\Ex]} instance $H$
	and apply the following transformation, illustrated in
	\cref{fig:count:BFactorRToWeightedAntiFactorR:original,%
	fig:count:BFactorRToWeightedAntiFactorR},
	for each simple vertex $\widehat v$
  to obtain an instance $G$ of \CountAntiFactorR{\Ex}.

	\proofsubparagraph*{Transformation.}
	We remove $\widehat v$ (with set $B$)
	and create a simple vertex $v$ (with set $\Ex$) and a complex vertex $u_v$.
	We connect $v$ to all neighbors of $\widehat v$ in $H$
	and call these edges \emph{outer} edges.
	We connect $v$ and $u_v$ by $\max\Ex+1$ parallel edges
	which we call \emph{inner} edges.%
	\footnote{Though these parallel edges disappear later,
	one could place \EQ{2} nodes on them to obtain a simple graph.}
	We assign the set $\Ex$ to $v$
	and the relation $R$ to $u_v$.
	(The relation $R$ is actually the same for all simple vertices.)
	For all $\beta \in [0,\max\Ex+1]$,
  relation $R$ accepts exactly one (arbitrary) set
  of exactly $\beta$ selected inner edges with weight $w_\beta$.
  All other sets of selected edges are accepted with weight 0,
  i.e., rejected.
  Before we define the weights, we provide some intuition.

	The weights for relation $R$ are chosen such that
  $v$ is always incident to $k$ selected outer edges,
	for some $k \ge \max\Ex+1$.
	For this we exploit that the weights $w_\beta$ might also be negative.
	Then, the ``bad'' solutions cancel out
	and only the ``good'' solutions contribute to the total number of solutions.

	If we denote by $\alpha$ the number of selected outer edges
  and by $\beta$ the number of selected inner edges,
	then the weights must satisfy
	\begin{equation}
		\sum_{\substack{\beta=0\\\alpha+\beta \notin \Ex}}^{\max\Ex+1} w_\beta
		= \begin{cases}
				1 & \alpha\ge \max\Ex+1 \\
				0 & \alpha\le \max\Ex
			\end{cases}
    .
    \label{eqn:count:weightRemoval:constraints}
  \end{equation}
  If we consider each $w_\beta$ as a variable,
  then the constraints from \cref{eqn:count:weightRemoval:constraints}
	form a system of linear equations with $\max\Ex+2$ variables.
  Observe further that all constraints with $\alpha\ge \max\Ex+1$ are identical.
  Hence, there are $\max\Ex+2$ constraints.
	Now consider the matrix $M$ formed by the coefficients
	of \cref{eqn:count:weightRemoval:constraints}.
	Observe that the anti-diagonal of $M$ and all entries below are ones
  as $\max\Ex+1, \max\Ex+2,\dots \notin \Ex$.
  Moreover, every entry one row above the anti-diagonal in $M$ is zero
	by the fact that $\max\Ex \in \Ex$.
	This observation directly implies that the sums for $\alpha+1$ and $\alpha$
	differ by (at least) one summand,
	i.e., by $w_{\max\Ex-\alpha}$.
	Hence, the matrix $M$ is invertible and there is a solution
	to \cref{eqn:count:weightRemoval:constraints}.
	Moreover,
	when starting with the sums for $\max\Ex+1$ and $\max\Ex$,
	we can find the value for $w_0$
	and eliminate it from the system of linear equations.
	Then, we repeat this process
	to iteratively eliminate $w_{1}$ up to $w_{\max\Ex+1}$.

	We conclude that there are values for each $w_\beta$
	such that \cref{eqn:count:weightRemoval:constraints} hold.
	Moreover, the above procedure already finds these values
	in time polynomially in the number of variables, which is $\max\Ex+2$.

  \proofsubparagraph*{Finalizing the Reduction.}
  Since the weights for the gadget are chosen such that
  the remaining graph does not see a difference
  between the vertex $\widehat v$ and this new gadget,
  we can replace all simple vertices by this procedure
  to get an edge- and relation-weighted \CountAntiFactorR{\Ex} instance
  which we denote by $G$.

	Recall that $H$ is the initial edge-weighted \CountAntiFactor{\Ex} instance.
	We have $n_G \le 2n_H$,
	$\totalDeg_G = \max\Ex+1$,
	and $\pw_G \le \pw_H + 1$.
	This finishes the proof.
\end{proof}
The next step of our reduction removes the weights from the relations
by using additional weighted edges.
\begin{lemma}
	\label{lem:count:VertexWeightedToEdgeWeightedAntiFactorR}
	Let $\Ex \subseteq \SetN$ be an arbitrary set
	(possibly given as input).

	We can many-one reduce edge- and relation-weighted \CountAntiFactorR{\Ex}
	with $r$ different weights
	to edge-weighted \CountAntiFactorR{\Ex}
  with at most $r$ weights
	such that
	\begin{itemize}
		\item
		the size increases by a multiplicative factor of $\O(r)$,
		\item
		$\totalDeg$ increases to $\totalDeg + r+1$,
		\item
		and $\pw$ increases to $\pw+1$.
	\end{itemize}
\end{lemma}
\begin{proof}
	We apply the following procedure to each complex vertex $u$
  (that uses weights).
	See \cref{fig:count:VertexWeightedToEdgeWeightedAntiFactorR}
  for an illustration of the modification.
	Let $R$ be the relation of $u$.
	Let $w_1,\dots,w_{r'}$ be the $r'<r$ different weights used by $R$.
	Assume without loss of generality that $r'=r$.

	For all $i\in[r]$,
	we add a vertex $u_i$ with relation $\HWin[1]{\{0,1\}}$
	and make it adjacent to $u$ by an edge of weight $w_i$.

	Based on $R$, we design a new relation $R'$ as follows:
	Whenever $R$ accepts the input $x$ with weight $w_i$, for some $i\in [r]$,
	then $R'$ accepts $x$
	but additionally requires that the edge to $u_i$ is selected
	while the edges to the other $u_{i'}$ remain unselected.

	One can easily check that this modification does not change the solution.
	Moreover, the pathwidth increases by at most 1
	and the degree of the complex vertices by at most $r$.
\end{proof}
In the next step of our reduction, we remove the edge weights from the graph.
First observe that we do not have to change edges of weight 1.
Furthermore, we can simply remove all edges with weight 0.

\begin{lemma}
	\label{lem:count:WeightedAntiFactorRToAntiFactorR}
	Let $\Ex \subseteq \SetN$ be a finite and non-empty set
	(possibly given as input).
	There is a Turing-reduction from edge-weighted \CountAntiFactorR{\Ex}
	with $r$ different weights
	to unweighted \CountAntiFactorR{\Ex}
	running in time $n^{\O(r)}$.
	The reduction
	is
	such that
	\begin{itemize}
		\item
		the size increases by a multiplicative factor of $\O(\log^2(n))$,
		\item
		the degree of the simple vertices stays the same,
		\item
		$\totalDeg$ increases to $\totalDeg + \O(1)$,
		\item
		and $\pw$ increases to $\pw+\O(1)$.
	\end{itemize}
\end{lemma}
\begin{proof}
	We use the same interpolation technique
	already used in \cite{CurticapeanM16}
	and later in \cite{MarxSS21} to remove the edge weights.

	We replace the edge weights $w_1,\dots,w_r$
	by variables $x_1,\dots,x_r$
	and treat the value of the solution as an $r$-variate polynomial $P$.

	Observe that, for each $i\in [r]$,
	$P$ has degree at most $m < n^2$ in $x_i$.
	Assume we can choose $n^2$ values for each $x_i$ independently.
	If we can evaluate $P$ on all $(n^2)^r$ combinations of values,
	then,
	by Lemma~1 in \cite{Curticapean15Icalp},
	we can recover the coefficients of $P$ in $n^{\O(r)}$ time.
	Finally, we output $P(w_1, \dots, w_r)$.

	It remains to show that we can achieve $n^2$ different values for each $x_i$.
	We first show how to realize edge weights of the form $2^i$.
	We replace the edge of weight $2^i$ by a chain of $2i+3$ edges
	where, for all $j\in[i]$, the $2j+1$th edge is a parallel edge.%
	\footnote{
	To avoid parallel edges,
	one can place one \EQ{2} node on each parallel edge.
	}
	We assign the relation \HWeq{1} to all vertices
	that are introduced by this gadget.
	When the original edge was selected,
	we have, for each $j \in [i]$, two choices
	for selecting one of the $2i+1$th parallel edges.
	Moreover, these choices are independent from each other.
	This contributes a weight of $2^i$ to the solution.
	If the original edge was not selected,
	then we only have a unique solution;
	we select the edges that are not parallel edges,
	i.e., the 2nd, 4th, \dots, edge.

	We use this as building block to realize arbitrary positive weights $W$.
	From the binary representation of $W$ we get a set $S=\{s_1,\dots,s_\ell\}\subseteq \SetN$
	such that $W = \sum_{i=1}^\ell2^{s_i}$.
	Let the edge of weight $W$ be between the vertices $v$ and $u$.
	We create new vertices $v_1,v_1',\dots,v_\ell,v_\ell'$
	and $u_1,u_1',\dots,u_\ell,u_\ell'$
	all with relation \HWeq{1}.
	Set $v_{\ell+1}=v$ and $u_{\ell+1}=u$.
	For all $i\in[\ell]$,
	we connect $v_i$ to $v_i'$ and connect $v_i'$ to $v_{i+1}$
	and similarly for $u_i, u_i', u_{i+1}$.
	For all $i\in[\ell]$,
	we add an edge of weight $2^{s_i}$ between $v_i$ and $u_i$.
	We use the previous method to realize the weight of these edge.
	The total weight sums up to $W$
	if exactly one weighted edge can be selected at a time.
	Now assume the edge from $v_i$ to $u_i$ is selected.
	Then, the edges from $u_i$ to $u_i'$ is not selected.
	By the relation of $u_i'$,
	the edge from $u_i'$ to $u_{i+1}$ must be selected.
	Therefore, the weighted edge from $u_{i+1}$ to $v_{i+1}$ cannot be selected.
	As the edge $u_{i+1}$ to $u_{i+1}'$ also cannot be selected
	and the argument applies recursively to the vertices $u_{i'}$ and $u_{i'}'$ with $i'<i$,
	no other weighted edge can be selected.

	Thus, it suffices to consider the weights from $1$ to $n^2$.
	Adding the parallel edges and the paths
	increases pathwidth only by a constant.
	The same holds for $\totalDeg$
	as the degree of each vertex is at most three.
	The size increases by a factor of $\O(\log^2(n))$
	as we introduce $\O(\log(n))$ many new vertices
	and the gadget for weight $2^i$ is of size $\O(i)$.
\end{proof}
\begin{note*}
	The previous construction uses weights between $1$ and $n^2$.
	By this choice the size increases by a logarithmic factor (in $n$) only.
	Now assume that the construction would use weights from $2^1$ to $2^{n^2}$.
	This choice could lead to a polynomial blow-up of the size of the graph.
	While this is not an issue for the lower bound for \CountAntiFactor{\Ex},
	this lemma could not be applied for the proof of the \sharpW{1}-hardness
	which we give later.
\end{note*}
Now, we can combine the previous steps
and prove the lower bound for \CountAntiFactorR{\Ex}.
\begin{proof}[Proof of \cref{lem:count:lbAntiFactorR}]
	Let $H$ be an edge-weighted \CountMinAntiFactor{\Ex} instance 
	with $\O(1)$ different edge-weights.
	We apply \cref{%
  lem:count:minSimpleAFRtoSimpleAFR,%
  lem:count:simpleAFRtoRelWeightedAFR,%
	lem:count:VertexWeightedToEdgeWeightedAntiFactorR,%
	lem:count:WeightedAntiFactorRToAntiFactorR}
  to obtain $n_H^{\O(\max\Ex)}$ instances $G_i$ of \CountAntiFactorR{\Ex}.
	As $\Ex$ is fixed, there are $n_H^{\O(1)}$ such instances.

	For each such instance $G_i$,
  we get $n_{G_i} \le \O(n_H \log^2(n_H))$,
	$\totalDeg_{G_i} \in \O(\max\Ex)$,
	and $\pw_{G_i} \le \pw_H + \O(1)$.
  Assume a fast algorithm for \CountAntiFactorR{\Ex} exists
	for some $\epsilon>0$.
	Then, since $\totalDeg_{G_i}\in\O(1)$,
	we can run this algorithm on all these instances $G_i$
	and recover the solution for $H$.
	The total running time of this process is
  \begin{align*}
		&n_H^{\O(1)} \cdot \sum_i
    (\max\Ex+2-\epsilon)^{\pw_{G_i}} \cdot n_{G_i}^{\O(1)}\\
		&\le \sum_i (\max\Ex+2-\epsilon)^{\pw_H + \O(1)}
		 	\cdot (n_H \log^2(n_H))^{\O(1)} \\
		&\le (\max\Ex+2-\epsilon)^{\pw_H} \cdot n_H^{\O(1)}
		.
  \end{align*}
	Thus, the running time directly contradicts \#SETH
	by \cref{lem:count:lbSimpleMinAFR}.
\end{proof}

%% file: img/count/cofinite-make_cofinite.tex
\centering
\begin{tikzpicture}[
problem/.style={rectangle, draw=black, text=black, minimum height=0.75cm,
minimum width=4.5cm
},
break/.style = {
  align=center,
  text width= 5cm,
},
red/.style
  = {draw=black, -latex,thick},
lem/.style={text=black,midway}
]
  \node[problem,break,thick] (BFactor) at (0, 0)
    {\CountBFactor};

  \node[problem, break] (simpleMinAntiFactor) [below =of BFactor]
    {\CountMinAntiFactor{[0,\ex]}};

  \node[problem, break] (simpleAntiFactor) [right=3cm of simpleMinAntiFactor]
        {edge-weighted \CountAntiFactor{[0,\ex]}};

  \node[problem, break] (relWeightedAntiFactor) [below=of simpleAntiFactor]
        {edge- and relation weighted \CountAntiFactorR{\Ex}};

  \node[problem, break] (edgeWeightedAntiFactor) [below=of relWeightedAntiFactor]
        {edge-weighted\\\CountAntiFactorR{\Ex}};

  \node[problem,break,thick] (AntiFactorR) at (BFactor |- edgeWeightedAntiFactor)
        {\CountAntiFactorR{\Ex}};

  \draw[red] ($(BFactor.north) + (0,0.75cm)$) -- (BFactor)
    node[lem, right]
    {\cref{thm:factor-tw-results}};

  \draw[red] (BFactor) -- (simpleMinAntiFactor)
    node[lem, right]
    {\cref{lem:count:lbSimpleMinAFR}};
  
  \draw[red] (simpleMinAntiFactor) -- (simpleAntiFactor)
    node[lem, above]
    {\cref{lem:count:minSimpleAFRtoSimpleAFR}};
  
  \draw[red] (simpleAntiFactor) -- (relWeightedAntiFactor)
    node[lem, right]
    {\cref{lem:count:simpleAFRtoRelWeightedAFR}};
  
  \draw[red] (relWeightedAntiFactor) -- (edgeWeightedAntiFactor)
    node[lem, right]
    {\cref{lem:count:VertexWeightedToEdgeWeightedAntiFactorR}};
  
  \draw[red] (edgeWeightedAntiFactor) -- (AntiFactorR)
    node[lem, above]
    {\cref{lem:count:WeightedAntiFactorRToAntiFactorR}};
  
  \draw[red,dotted] (simpleMinAntiFactor) -- (AntiFactorR)
    node[lem, right]
    {\cref{lem:count:lbAntiFactorR}};

  \draw[red] (AntiFactorR) -- ($(AntiFactorR.south) - (0,0.5cm)$)
    node[lem, below, pos=1]
    {\cref{fig:count:realization:chain}};

\end{tikzpicture}

%% file: cofinite/count-w1.tex
\subsection{High-level construction for \#W[1]-Hardness}
In this section, we show the proof of \cref{lem:count:lbCount1AntiFactorR},
i.e., the lower bound for \CountAntiFactorSizeR{1}.
For this we use a similar high-level idea as for the previous subsection.
However, we do not start from \#SAT
but from \textsc{Counting Colorful Hitting $k$-Sets} (\#\ColHS).
\begin{definition}[\#\ColHS]
	Given $k$ disjoint sets $U_1,U_2,\dots,U_k$ of $n$ elements each
	and sets $A_1,A_2,\dots,A_m$ each a subset of $\bigcup_{i\in [k]} U_i$,
	find the size of the set
	\[
		S=\{(s_1,s_2,\dots,s_k) \in U_1\times U_2\times \dots \times U_k
		\mid \forall j\in[m] \; \exists i\in[k]: s_i\in A_j\}
	\]
\end{definition}
We use a known lower bound for \#\ColHS as an initial step of our proof.
\begin{lemma}[Lemma~6.1 from \cite{CurticapeanM16}] \label{lem:count:colHS-lb}
	For each fixed $k\geq 2$, the following holds.
	If there exists an $\epsilon>0$ such that, for each $d\in \SetN$,
	there is an $\O(n^{k-\epsilon})$ time algorithm for
	\#\ColHS on instances with universe size $n$ and $\O(k^d\cdot \log^d n)$ sets,
	then \#SETH fails.
	Here, the constant factor in the running
	time may depend on $d$.
\end{lemma}
The main result of this section is \cref{lem:count:colHStoCountAntiFactorR}.
\begin{lemma}
	\label{lem:count:colHStoCountAntiFactorR}
	There is a polynomial time many-one reduction from \#\ColHS
	with a universe size of $n$ elements each and $m$ sets
	to relation-weighted \CountAntiFactorSizeR{1}
	with $3$ different weights
	such that the size is bounded by $\O(k \cdot n \cdot m)$,
	the total degree of the complex nodes is $\O(1)$,
	and the pathwidth is at most $k+\O(1)$.
\end{lemma}
We combine \cref{lem:count:colHStoCountAntiFactorR}
with the intermediate steps from \cref{sec:lower:count:seth}
to prove \cref{lem:count:lbCount1AntiFactorR}.
This also proves \cref{thm:count:w1} since
all our reductions are parameter preserving and
\#\ColHS is \sharpW{1}-hard \cite{CurticapeanM16}.

\begin{proof}[Proof of \cref{lem:count:lbCount1AntiFactorR}]
	For a given \#\ColHS instance $H$
	we apply \cref{lem:count:colHStoCountAntiFactorR}
	to get a relation-weighted \CountAntiFactorSizeR{1} instance,
	then use \cref{lem:count:VertexWeightedToEdgeWeightedAntiFactorR}
	to transform this into an edge-weighted \CountAntiFactorSizeR{1} instance,
	and finally apply \cref{lem:count:WeightedAntiFactorRToAntiFactorR}
	to transform this into polynomially many unweighted \CountAntiFactorSizeR{1} instance $G_i$.

	For each such instance $G_i$, we get
	$n_{G_i} \le \O(n \cdot m \cdot k \cdot \log^2(n \cdot m \cdot k))$
	and $\pw_G = k + c_1$
	for some constant $c_1>0.$

	Note that \cref{lem:count:WeightedAntiFactorRToAntiFactorR}
	adds an overhead of
	$\O((n \cdot m \cdot k)^{c_2})$, for some $c_2>0$,
	to the running time.
	Assume that, for all $c>0$,
	there is an $\O(N^{\pw-c})$ time algorithm for \CountAntiFactorSizeR{1}
	on instances of size $N$ and pathwidth $\pw$.
	We apply this algorithm to all instances
	and get the following time bound:
	\begin{align*}
		\O\left(
			(n \cdot m \cdot k)^{c_2} + \sum_i n_{G_i}^{\pw_{G_i}-c}
			\right)
		& \le \O\left(
			(n \cdot m \cdot k)^{c_2} \cdot {\left(n \cdot m \cdot k \cdot \log^2(n \cdot m \cdot k)\right)}^{\! k + c_1-c}
			\right)
	\intertext{With \(c_3 = c_1+c_2 \), this can be bound by}
		&\le\O\left(
			\left(
				n \cdot m \cdot k \cdot \log^2(n \cdot m \cdot k)
			\right)^{k + c_3-c}
			\right)\\
		&\le\O\left(
			(n \cdot m \cdot k)^{k+c_3-c} \cdot \log^{2(k+c_3-c)}(n \cdot m \cdot k)
			\right) \\
		\intertext{Since \(\log^a N \le N^{b}\), for large enough \(N\),
		we get:}
		&\le\O\left(
			(n \cdot m \cdot k)^{k + c_3-c} \cdot (n\cdot m\cdot k)^{c_4}
			\right) \\
		&\le\O\left(
			n^{k+c_5-c} \cdot m^{k+c_5-c} \cdot k^{k+c_5-c}
			\right) \\
		&\le\O\left(
			n^{k+c_5-c} (k^d \log^d n)^{k+c_5-c} \cdot k^{k+c_5-c}
			\right) \\
		\intertext{We again use the above bound for the \(\log^d(n)\)
		and the fact that \(k\) is fixed:}
			&\leq \O(n^{k+c_5-c} n)
			 = \O(n^{k+c_5+1-c}).
	\end{align*}
	By choosing $c>c_5+1$,
	we have an algorithm for \#\ColHS that runs in time
	$\O(n^{k-\epsilon})$, for some constant $\epsilon>0$.
	This contradicts \cref{lem:count:colHS-lb}.
\end{proof}
In the remaining part of the section we prove \cref{lem:count:colHStoCountAntiFactorR}.

\subparagraph*{Idea of the Construction.}
We reduce \#\ColHS to relation-weighted \CountAntiFactorSizeR{1}.
Let $k$ be the number of sets $U_i$,
$n$ the size of the sets $U_i$,
and $m$ the number of set $A_j$ we want to hit.
We follow a similar high-level construction as for previous SETH based lower bounds
and derive the construction from a $k \times m$ grid.
The $i$th row corresponds to set $U_i$
and the $j$th column to the set $A_j$.
See \cref{fig:count:W1high-level} for an illustration.

\begin{figure}
	\centering
	\includegraphics[page=1]{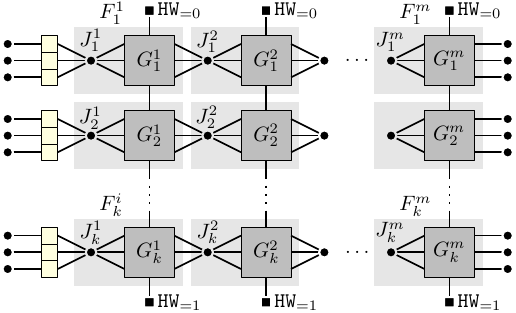}
	\caption{
	High-level construction of the \CountAntiFactorSizeR{1} instance.
	The boxes containing $G_i^j$ and $J_i^j$ describe the gadgets $F_i^j$.
	The details of the gadgets $G_i^j$ are shown in \cref{fig:count:W1mainGadget}.
	}
	\label{fig:count:W1high-level}
\end{figure}

At the crossing point of the $i$th row and the $j$th column,
the gadget $F_i^j$ checks
whether the element selected from $U_i$ hits the set $A_j$.
For the encoding of the elements,
we connect $F_i^j$ and $F_i^{j+1}$ by $n-1$ edges
such that a selection of $\ell$ edges
encodes the selection of the $\ell+1$th element from $U_i$.
Moreover, $F_i^j$ is connected to $F_{i+1}^j$ by a single edge.
This edge indicates whether $A_j$ was hit by an element from $U_{i'}$ for some $i' \le i$.

We first formally describe the gadgets $F_i^j$.
Then, we construct the graph and these gadgets
such that the pathwidth of the construction is not too large,
i.e., bounded by $k + \O(1)$.

Fix a gadget $F_i^j$ in the following.
Let $\IN_1,\dots,\IN_{n-1}$ be the edges to $F_i^{j-1}$,
let $\OUT_1,\dots,\OUT_{n-1}$ be the edges to $F_i^{j+1}$,
let $\IN_v$ be the (vertical) edge to $F_{i-1}^j$,
and let $\OUT_v$ be the (vertical) edge to $F_{i+1}^j$.
We define $F_i^j$ such that it satisfies the following constraints.
\begin{enumerate}
	\item
	The selection of the outgoing edges must be monotonous.

	Formally, there must be an $\ell \in [n]$
	such that the edges $\OUT_1,\dots,\OUT_{n-\ell}$ are not selected
	but the edges $\OUT_{n-\ell+1},\dots,\OUT_{n-1}$ are selected.

	\item
	$F_i^j$ ensures that the information about the selected element is transferred consistently.

	That is,
	for all $\ell\in[n]$,
 	edge $\IN_\ell$ is selected if and only if edge $\OUT_\ell$ is selected.

	\item
	$F_i^j$ checks whether the element from $U_i$ hits the set $A_j$.
	This information is encoded by the edge between $F_i^j$ and $F_{i+1}^j$.

	If $\IN_v$ is selected,
	then $\OUT_v$ is selected.
	Otherwise,
	$\OUT_v$ is selected if and only if
	there is an $\ell\in[n]$ such that
	exactly $\ell-1$ edges from $\IN_1,\dots,\IN_{n-1}$ are selected
	and the $\ell$th element in $U_i$ is contained in $A_j$.
\end{enumerate}

\subparagraph*{Construction of the Instance.}
We set $\Ex=\{n-2\}$
which is possible as only the \emph{size} of $\Ex$ is fixed
and not $\Ex$ itself.
The graph is formally defined as follows:
\begin{itemize}
	\item
	For all $i\in[k]$ and all $j\in[m]$,
	there is a vertex $J_i^j$.
	\item
	For all $i\in[k]$ and all $j\in[m]$,
	there is a weighted gadget $G_i^j$ such that
	$G_i^j$ is connected to $J_i^j$ and $J_i^{j+1}$ by $n-1$ parallel edges each
	and to $G_{i\pm1}^j$ by a single edge each.
	The weights are defined later.

	\item
	For all $i\in[k]$ and all $j\in[m]$,
	the vertex $J_i^j$ together with the gadget $G_i^j$
	constitute the gadget $F_i^j$ from above.

	\item
	For all $i\in[n]$,
	there is a monotonicity gadget
	whose behavior is identical to the one of the $F_i^j$s for edges $\OUT_1,\dots,\OUT_{n-1}$
	(i.e., the first property from above)
	and attach it to $J_i^1$ by $n-1$ parallel edges.
	These gadgets are defined later as the fourth stage of $G_i^j$.
	\item
	For all $i\in[n]$,
	attach $n-1$ fresh vertices to $G_i^m$.
	\item
	For all $j\in[m]$,
	there is a vertex with relation \HWeq[1]{0}
	that is adjacent to $G_1^j$.
	\item
	For all $j\in[m]$,
	there is a vertex with relation \HWeq[1]{1}
	that is adjacent to $G_k^j$.
\end{itemize}
\subparagraph*{Proof of Correctness.}
Assume that we choose the weights for $G_i^j$
such that the combined behavior of $J_i^j$ and $G_i^j$
corresponds to that of $F_i^j$ as we described it above.
Then, it is clear that there is a one-to-one correspondence
between the solutions for the \#\ColHS instance
and the solutions for the \CountAntiFactorR{\Ex} instance.
For this we crucially depend on the nodes with relations \HWeq[1]{0} and \HWeq[1]{1}.
The \HWeq[1]{0} relations ensure
that every set $A_j$ is initially not hit.
Since the gadgets $G_i^j$ propagate the selection of the top edge to the bottom edge
whenever the set is not hit by the selected element,
there must be some $G_{i'}^j$ such that the element from $U_{i'}$ hits $A_j$.
Otherwise the relation \HWeq[1]{1} cannot be satisfied, but this is forbidden.

We first show how to define the weights for the $G_i^j$
such that we can treat $G_i^j$ and $J_i^j$ together as the $F_i^j$ gadget.
Observe that the degree of $G_i^j$ is large
and more specifically depends on the input.
Hence, we cannot treat $G_i^j$ as a complex node.
Thus, in a second step we construct
the $G_i^j$ from smaller building blocks
such that pathwidth and degree are constants.

\subparagraph*{Defining the Weights.}
We focus on a fixed $F_i^j$ and thus $G_i^j$ in the following.
By the monotonicity assumption,
we know that we can describe the signature of $F_i^j$ by a function $\widehat f(\alpha, \IN_v, \gamma, \OUT_v)$,
where $\alpha$ denotes the \emph{number} of selected edges from $\{\IN_\ell\}_\ell$
and $\gamma$ denotes the number of selected edges from $\{\OUT_\ell\}_\ell$.
It actually suffices to consider the function $f(\alpha,\gamma)$,
with the same meaning for $\alpha$ and $\gamma$ as for $\widehat f$
where we have
\[
	\widehat f (\alpha, \IN_v, \gamma, \OUT_v) =
	\begin{cases}
		0 & \IN_v = 1 \land \OUT_v = 0 \\
		f(\alpha,\gamma) & \IN_v = \OUT_v \\
		f(\alpha,\gamma) & \IN_v = 0 \land \OUT_v = 1 \land \text{ the \(\gamma+1\)th element of \(U_i\) hits } A_j
	\end{cases}
\]
By the definition of $F_i^j$, we need
$f(\alpha,\alpha)=1$, for all $\alpha$,
and $f(\alpha,\gamma)=0$, whenever $\gamma \neq \alpha$.

Now, consider the signature of $G_i^j$
and let $\beta$ denote the number of selected edges between $J_i^j$ and $G_i^j$.
We can sue the same arguments to denote this signature by $g(\beta,\gamma)$
if we additionally impose monotonicity constraints on the ingoing and outgoing edges.
From the definition of $\Ex$ we get:
\[
	f(\alpha, \gamma) = \sum_{\substack{\beta=0 \\ \alpha+\beta \neq n-2}}^{n-1} g(\beta,\gamma)
\]
	To satisfy the constraints for $f$,
	one can easily check that it suffices to set
	\begin{itemize}
		\item
		for all $\gamma \neq n-1$:
		$g(n-2-\gamma, \gamma)\deff -1$,

		\item
		for all $\delta \neq n-1$:
		$g(n-1,\delta) = g(\delta,n-1) \deff 1$,

		\item
		$g(n-1,n-1) \deff -(n-2)$,

		\item
		and $g(\beta,\gamma)=0$ otherwise.
	\end{itemize}
As seen above,
we lift the function $g$ to all inputs of $G_i^j$
by imposing that the edges going to $J_i^j$ are selected monotonically.
But in contrast to the previous requirement for the edges $\IN_\ell$ of $F_i^j$,
now, for some $\ell\in[n]$,
the top/first $\ell$ edges $\OUT_1,\dots,\OUT_\ell$ are selected
and the bottom $n-1-\ell$ edges $\OUT_{\ell+1},\dots,\OUT_n$ are not selected.

\subparagraph*{Construction of the Gadgets.}
It remains to define the gadgets $G_i^j$.
We decompose the gadget into four stages
as shown in \cref{fig:count:W1mainGadget}.
We informally describe the purpose of the stages
and then define the relations formally.
For this we refer to the edges as left, top, bottom, and right edges
and abbreviate them by $\ell, t, b, r$.
If there are more than one edge on one side,
we enumerate them from left to right or from top to bottom.
\begin{enumerate}
	\item
	Check whether the ingoing edges are selected monotonically
	starting from the top.
	While doing this,
	select the $\ell$th edge going to the next stage
	if and only if the first $\ell-1$ edges are selected.
	The vertical edges encode whether all ingoing edges of the blocks above have been selected.

	We formally define
	$M = \{\ell tb,tr,\emptyset \}$.

	\item
	This stage assigns the different weights to the solutions
	depending on the number of selected ingoing and outgoing edges.
	By the first and last stage (monotonicity gadgets),
	this information is encoded by the position of the selected edge
	coming from the first stage
	and the position of the selected edge going to the third stage.
	The first and second vertical edge encodes whether $n-1$ ingoing edges
	and $n-1$ outgoing edges are selected,
	respectively.

	We define
	$R_0 = \{
		\ell r,
		\ell t_2b_2, t_1b_1r,
		\emptyset, t_1b_1, t_2b_2, t_1t_2b_1b_2
	\}$.
	The input $\ell r$ of $R_0$ is assigned weight $-1$
	and all other accepted inputs are assigned weight $1$.
	$R_W$ completely agrees with $R_0$
	except that the input $t_1t_2b_1b_2$ is assigned weight $-(n-2)$.

	\item
	If the $\ell$th edge from the bottom is selected,
	then this corresponds to element $u_\ell$ being selected from $U_i$.
	We propagate this information to the next stage.
	If $u_\ell$ is contained in $A_j$,
	then we assign relation $C^+$ to the $\ell$th block from the bottom
	and otherwise the relation $C^0$.
	The vertical edges are selected
	if the set has been hit by any element
	(possibly also by some element from some $U_{i'}$ with $i' < i$).

	We define
	$C^0=\{\emptyset, \ell r, tb, \ell tbr\}$
	and $C^+ = C^0 \cup \{\ell br \}$.

	\item
	This stage is a rotated version of the first stage.
	That is,
	if the $\ell$th edge from the bottom is selected,
	then we select $\ell-1$ outgoing edges.
	The vertical edges now correspond to the fact
	whether all outgoing edges from below have been selected.

	We define $M' = \{tbr, \ell b, \emptyset \}$.
\end{enumerate}
\begin{figure}
	\centering
	\includegraphics[page=2]{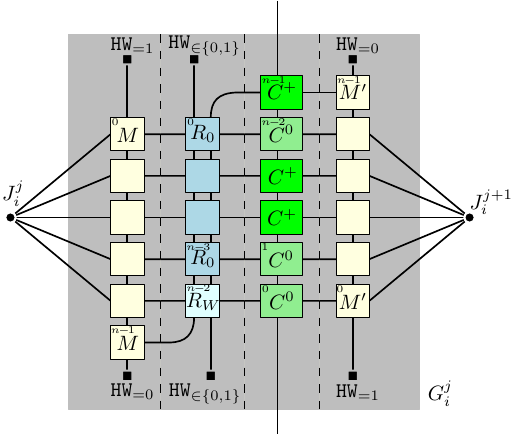}
	\caption{
	The gadget $G_i^j$ checks whether the element from $U_i$ hits the set $A_j$
	(here $u_3,u_4,u_6$ hit $A_j$ and the other elements not).
	It additionally ensures together with $J_i^j$
	that the information about this selected element
	is consistent in the whole graph.
	The different stages are separated by dashed lines.
	}
	\label{fig:count:W1mainGadget}
\end{figure}

\begin{claim}
	The pathwidth of the final graph is $k+\O(1)$.
\end{claim}
\begin{claimproof}
	We use a mixed search strategy to clean the graph from left to right.
	For this we place searchers on the vertices $J_1^j$, for all $j$,
	and clean the monotonicity gadgets going from top to bottom.

	The remaining graph is cleaned in $m$ rounds.
	For the $j$th round we clean the gadgets $G_i^j$ from top to bottom.
	Before cleaning each $G_i^j$
	we place one searcher on $J_i^{j+1}$
	and afterwards we remove the searcher from $J_i^{j}$.
	Then, we continue with the next row until all $G_i^j$ are cleaned
	and we proceed to the $j+1$th round.
\end{claimproof}
For the correctness of the final construction it remains to show
that the gadget $G_i^j$ together with vertex $J_i^j$
satisfy all properties of $F_i^j$.
For this we observe
that the fourth stage of $G_i^j$ ensures property (1) of $F_i^j$.
The third stage ensures the third property
and the
first and second stage together with $J_i^j$ ensure property (2)
by the choice of the weights.

%% file: cofinite/count-rel.tex
\subsection{Replacing the Relations}
\label{sec:lower:count:relations}
In this section, we show
the reduction from \CountAntiFactorR{\Y} to \CountAntiFactor{(\Ex,\Y)},
i.e., we prove \cref{lem:count:AntiFactorRToAntiFactor},
by a chain of reductions
(cf.\ \cref{fig:count:realization:chain}).
We make use of the Holant framework,
which was also used in \cite{MarxSS21},
to formally state the results.
The first step uses Lemmas~7.5 and 7.6 from \cite{MarxSS21}.
Observe for this that the lemmas work for \CountBFactor even when
$B$ is cofinite,
that is \CountAntiFactor{\overline B}
because the simple vertices of the instance are not changed in any way.

\begin{figure}
	\centering
	\input{img/count/cofinite-chain-of-reduction}
	\caption{
	Steps in the chain of reductions
	from \CountAntiFactorR{\Y}
	to \CountAntiFactor{(\Ex,\Y)}, i.e., \Hol{\HWin{\overline \Y},\HWin{\Exbar}}.
	Dotted lines indicate results
	obtained by combined reductions.
	}
	\label{fig:count:realization:chain}
\end{figure}
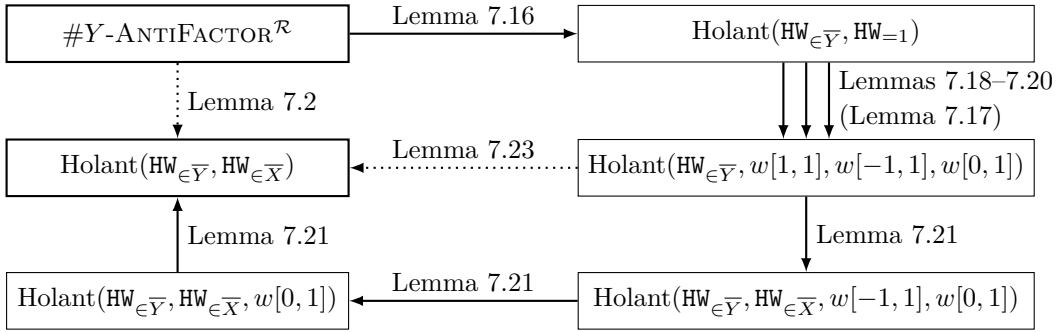
\begin{lemma}[Lemma~7.5 and~7.6 in \cite{MarxSS21}]
	\label{lem:count:relationsToForcingEdges}
	There is a polynomial-time Turing reduction from \CountAntiFactorR{\Ex}
	to $\Hol{\HWin{\Exbar}, \HWeq{1}}$
	such that the maximum degree increases to at least 6 and the
	pathwidth increases by at most a constant depending only on $\totalDeg$,
	where $\totalDeg$ denotes the maximum total
	degree of the complex nodes in any bag of the path decomposition.
\end{lemma}
\begin{note*}
	Lemma~7.5 in \cite{MarxSS21} requires that the relation is even,
	i.e., the Hamming weight of every accepted input is even.
	We can easily make every relation even by adding an additional input
	that is selected whenever the parity of the original input is odd.
	This additional input is then connected to a \EQ{1} node,
	which can easily be realized
	by forcing $\max\Ex+1$ edges to a fresh vertex using \HWeq[1]{1} nodes.
\end{note*}
Before proceeding with the next steps,
we define, for all $x,y\in \SetZ$,
$\wtnode{x,y}$ as a new type of node
which has one dangling edge $e$ and the following signature:
\[
f(e) = \begin{cases}
	x & \text{ if $e$ is not selected}
	\\
	y & \text{ if $e$ is selected}
\end{cases}
.
\]
Observe that \HWeq[1]{1} is precisely \wtnode{0,1}
and \HWin[1]{\{0,1\}} corresponds to \wtnode{1,1}.
In the following constructions we additionally use a \wtnode{-1,1} node.
We use the \wtnode{x,y} notation in the following wherever possible.

\begin{lemma}
	\label{lem:count:forcingEdges}
	Let $\Ex \subseteq \SetN$ be a finite set
	such that $\Ex \not\subseteq \{0\}$.
	Let $R_1,\dots,R_d$ be $d$ arbitrary relations for some $d \ge 0$.
	There is a polynomial-time Turing reduction from
	\[
	\Hol{R_1,\dots,R_d,\HWin{\Exbar}, \HWeq{1}}
	\text{ to } \Hol{R_1,\dots,R_d,\HWin{\Exbar}, \wtnode{1,1}, \wtnode{-1,1}, \wtnode{0,1}}
	\]
	such that
	$\totalDeg$ increases to $\totalDeg \cdot f(\max \Ex)$
	and $\pw$ increases to $\pw+\totalDeg \cdot f(\max \Ex)$.
\end{lemma}
For the proof of \cref{lem:count:forcingEdges}
we distinguish the following cases depending on the structure of $\Ex$:
\begin{enumerate}
	\item
	\label{case:forcingEdges:gapTwo}
	$\maxgap(\Exbar)>1$:
	\begin{enumerate}
		\item
		\label{case:forcingEdges:gapTwo:notOneToK}
		$\Ex \neq [k]$ for all $k\geq 2$
		\item
		\label{case:forcingEdges:gapTwo:oneToK}
		$\Ex = [k]$ for some $k\geq 2$
	\end{enumerate}

	\item
	\label{case:forcingEdges:gapOne}
	$\maxgap(\Exbar)=1$:
	In this case, $\Ex$ contains no two consecutive numbers
	but contains at least one positive integer.

	\item
	\label{case:forcingEdges:gapZero}
	\label{case:forcingEdges:gapZero:zeroToK}
	$\maxgap(\Exbar)=0$:
	In this case, there is a $k \ge 1$ such that $\Ex=[0,k]$.
\end{enumerate}
Note that the uncovered case $\Ex=\{0\}$ corresponds to the edge-cover problem.
In this case, our techniques to realize relations fail.
However,
we show the tight lower bound by a separate reduction in \cref{sec:lower:edge-cover}.

We group Cases~\ref{case:forcingEdges:gapTwo:oneToK}
and \ref{case:forcingEdges:gapZero:zeroToK} together
such that we are left with three separate proofs.
Then, \cref{lem:count:forcingEdges} follows from
\cref{lem:count:forcingEdges:case1,lem:count:forcingEdges:case2,%
lem:count:forcingEdges:case3}.

\begin{lemma}[Case~\ref{case:forcingEdges:gapTwo:notOneToK} of \cref{lem:count:forcingEdges}]
	\label{lem:count:forcingEdges:case1}
	Let $\Ex\subseteq \SetN$ be a finite set such that
	$\maxgap(\Exbar)>1$ and $\Ex\neq [k]$ for some $k\geq 2$.
	Let $R_1,\dots,R_d$ be $d$ arbitrary relations for some $d \ge 0$.
	There is a polynomial-time many-one reduction from
	\[
	\Hol{R_1,\dots,R_d,\HWin{\Exbar}, \HWeq{1}}
	\text{ to } \Hol{R_1,\dots,R_d,\HWin{\Exbar}, \wtnode{1,1}, \wtnode{0,1}}
	\]
	such that
	$\totalDeg$ increases to $\totalDeg \cdot f(\max \Ex)$
	and $\pw$ increases to $\pw+\totalDeg \cdot f(\max \Ex)$.
\end{lemma}
\begin{proof}
	Since $\maxgap(\Exbar)=d>1$,
	there is some $a\ge 0$ such that
	$[a,a+d+1]\cap \Ex=[a+1,a+d]$.
	We show how to get $\HWeq[2]{1}$ and \HWeq[3]{1} nodes.
	This is enough to get $\HWeq[\ell]{1}$ nodes for arbitrary $\ell$
	by using the construction from \cref{lem:relation:hw1-decision}
	(cf.~Lemma~5.5 from \cite{MarxSS21} for more details).

	We follow the process from Lemma~5.8 from \cite{MarxSS21}
	to first get a $\EQ{3}$ node:
	We first construct a \EQ{d+1} node
	by forcing $a$ edges to a new vertex using \wtnode{0,1} nodes.
	Adding $d-2$ copies of a \wtnode{1,1} node to this \EQ{d+1} node
	gives us a \EQ{3} node.

	Next we construct a \HWeq[2]{1} node.
	Since $\maxgap(\Exbar)>1$ and $\Ex \neq [k]$,
	there are $a'>0$ and $d'>0$ such that
	$[a',a'+d'+1]\cap \Ex=[a'+1,a'+d']$.
	Take two $\EQ{3}$ nodes and set one edge of each of them as a dangling edge.
	Using the remaining two edges, connect each of them to two vertices $u$ and $v$.
	Attach $a'-1$ copies of a \wtnode{0,1} node to $u$
	and $\max \Ex$ copies of a \wtnode{0,1} node to $v$.
	Observe that, by virtue of $u$,
	at most one $\EQ{3}$ node can have its edges selected
	and by virtue of $v$, at least one $\EQ{3}$ node must have its edges selected.
	Thus, exactly one dangling edge must be selected in any solution.

	To get a \HWeq[3]{1} node we distinguish two cases.
	If $d'\ge 2$ or $a \ge 1$ (and $d\ge 2$),
	we can use the same approach as for the \HWeq[2]{1} node
	but using three \EQ{3} nodes instead.
	For the case $a=0$ and $d'=1$, this construction is not sufficient.
	Then, we force $d-1$ edges to a new vertex and add three \HWeq[2]{1} nodes via one edge to it.
	The remaining edges of the \HWeq[2]{1} nodes are the dangling edges.
	If $d+2 \in \Ex$, then we get a \HWeq[3]{1} gadget and we are done.
	Otherwise we have a \HWin[3]{\{0,1\}} gadget.
	In this case we can use the construction for the \HWeq[2]{1} gadget
	with three \EQ{3} nodes and replace the vertex $u$ by this \HWin[3]{\{0,1\}} gadget.
\end{proof}

\begin{lemma}[Case~\ref{case:forcingEdges:gapOne} of \cref{lem:count:forcingEdges}]
	\label{lem:count:forcingEdges:case2}
	Let $\Ex\subseteq \SetN$ be a finite set such that
	$\maxgap(\Exbar)=1$.
	Let $R_1,\dots,R_d$ be $d$ arbitrary relations for some $d \ge 0$.
	There is a polynomial-time many-one reduction from
	\[
	\Hol{R_1,\dots,R_d,\HWin{\Exbar}, \HWeq{1}}
	\text{ to } \Hol{R_1,\dots,R_d,\HWin{\Exbar}, \wtnode{1,1}, \wtnode{-1,1}, \wtnode{0,1}}
	\]
	such that
	$\totalDeg$ increases to $\totalDeg \cdot f(\max \Ex)$
	and $\pw$ increases to $\pw+\totalDeg \cdot f(\max \Ex)$.
\end{lemma}
\begin{proof}
	Let $k=\max \Ex$.
	Since $\maxgap(\Exbar)=1$
	and $k\in\Ex$, we have $k-1,k+1\not\in\Ex$.
	Create a vertex $v$ and attach $k-1$ copies of a $\wtnode{0,1}$ node,
	one $\wtnode{1,1}$ node (by a red edge) and
	a $\wtnode{-1,1}$ node (by a blue edge) to it.
	Finally, we add $\ell$ dangling edges to $v$.
	See \cref{fig:count:forcingGadget:case2} for an illustration.
	We claim that $v$ acts as a $\HWeq[\ell]{1}$ node.

	\begin{figure}
		\centering
		\includegraphics{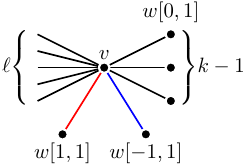}
		\caption{
		Construction of the gadget in \cref{lem:count:forcingEdges:case2}
		and \cref{lem:count:forcingEdges:case3}.
		}
		\label{fig:count:forcingGadget:case2}
		\label{fig:count:forcingGadget:case3}
	\end{figure}

	Because of the $k-1$ different $\wtnode{0,1}$ nodes,
	$v$ cannot be incident to exactly one more edge.
	\begin{itemize}
		\item Suppose no dangling edges are chosen.
		Then, we can extend the solution by selecting neither of the blue and red edges
		or selecting both of them.
		The number of solutions is $1+(-1)=0$ and therefore,
		this is not a possibility.
		\footnote{
		Technically this is a solution but contributes $0$ to the Holant and thus can be ignored (or treated as invalid).
		}
		\item
		Suppose more than one dangling edge is chosen.
		Then, we can pick any combination
		of the red and blue edges. The number of solutions is $1-1+1-1=0$.
		As before, this is not a
		possibility.
		\item Suppose exactly one dangling edge is chosen.
		Then, we must pick at least one of the
		red or blue edges.
		The number of solutions is $1-1+1=1$.
	\end{itemize}
	Thus, the only solution is to pick exactly one dangling edge.
\end{proof}

\begin{lemma}[Cases~\ref{case:forcingEdges:gapTwo:oneToK} and~\ref{case:forcingEdges:gapZero:zeroToK} of \cref{lem:count:forcingEdges}]
	\label{lem:count:forcingEdges:case3}
	For some integer $k$,
	let $\Ex=[0,k]$ with $k\geq 1$ or $\Ex=[k]$ with $k\geq 2$.
	Let $R_1,\dots,R_d$ be $d$ arbitrary relations for some $d \ge 0$.
	There is a polynomial-time Turing reduction from
	\[
	\Hol{R_1,\dots,R_d,\HWin{\Exbar}, \HWeq{1}}
	\text{ to } \Hol{R_1,\dots,R_d,\HWin{\Exbar}, \wtnode{1,1}, \wtnode{-1,1}, \wtnode{0,1}}
	\]
	such that
	$\totalDeg$ increases to $\totalDeg \cdot f(\max \Ex)$
	and $\pw$ increases to $\pw+\totalDeg \cdot f(\max \Ex)$.
\end{lemma}
\begin{proof}
	We first reduce from $\Hol{\HWin{\Exbar}, \HWeq{1}}$
	to $\Hol{\HWin{\Exbar}, \HWin{\{0,1\}}}$ using Lemma~7.8 from \cite{MarxSS21}.
	Then, the construction is the same as for the reduction in \cref{lem:count:forcingEdges:case2}.
	See \cref{fig:count:forcingGadget:case3} for the construction.

	We claim that $v$ acts as a $\HWin[\ell]{\{0,1\}}$ node.
	Note that after attaching the $k-1$ different $\wtnode{0,1}$ nodes,
	$v$ must have at least two more selected incident edges
	(zero and one are not allowed).

	\begin{itemize}
		\item
		Suppose more than one dangling edge is chosen.
		Then, we can pick any combination
		of the red and blue edges. The number of solutions is $1-1+1-1=0$. Therefore, this is not a
		possibility.
		\item Suppose no dangling edges are chosen.
		Then, we can extend the solution only by selecting the blue and the red edge simultaneously.
		The number of solutions is $1$.
		\item Suppose exactly one dangling edge is chosen.
		Then, we must pick at least one of the
		red or blue edges. The number of solutions is $1-1+1=1$.
	\end{itemize}
	Thus, the only solution is to pick exactly one or zero dangling edges.
\end{proof}
Next, we show that we can realize $\wtnode{x,y}$ nodes.
In particular,
we can get the $\wtnode{1,1}$, $\wtnode{-1,1}$, and $\wtnode{0,1}$
nodes introduced by \cref{lem:count:forcingEdges}.

\begin{lemma}\label{lem:count:wtnode-removal}
	Let $\Ex\subseteq \SetN$ be a fixed, finite set
	with $\Ex\not\subseteq \{0\}$.
	Let $R_1,\dots,R_d$ be $d$ arbitrary relations for some $d \ge 0$.
	The following holds for arbitrary values $x$ and $y$.
	There is a polynomial-time Turing reduction from
	$\Hol{R_1,\dots,R_d, \HWin{\Exbar}, \wtnode{x,y}}$
	to $\Hol{R_1,\dots,R_d, \HWin{\Exbar}}$
	such that
	$\totalDeg$ decreases
	and $\pw$ increases to $\pw+\totalDeg \cdot f(\max \Ex)$.
\end{lemma}

\begin{proof}
	We use Lemma~7.11 from \cite{MarxSS21} as our prototype. However, some arguments from their
	proof do not follow in our case.

	Let $G$ be an instance of $\Hol{\HWin{\Exbar}, \wtnode{x,y}}$
	and let $U$ be the set of $\wtnode{x,y}$ nodes in $G$.
	Let $A_i$ denotes the number of possible solutions in $G$
	where for exactly $i$ of the $\wtnode{x,y}$ nodes
	the dangling edge is not selected
	and for the other $\abs{U}-i$ nodes the dangling edge is selected.
	Then, we have
	\begin{align}
		\label{eq:count:holantOriginalGraph}
		\Hol{G} = \sum_{i=0}^{\abs{U}} A_i x^i y^{\abs{U}-i}
		.
	\end{align}
	For a new parameter $d$,
	we construct graphs $G_d$ from $G$
	where we replace each $\wtnode{x,y}$ node by a gadget $H_d$
	which has exactly one dangling edge.
	The construction of $H_d$ is given later
	as it depends on $\Ex$.
	Let $h_0(d)$ denote the number of solutions for $H_d$
	when the dangling edge is not selected
	and $h_1(d)$ when the dangling edge is selected.
	We get
	\[
		\Hol{G_d}
		= \sum_{i=0}^{\abs{U}} A_i h_0(d)^i h_1(d)^{\abs{U}-i}
		= h_1(d)^{\abs{U}} \sum_{i=0}^{\abs{U}} A_i {\left(\frac{h_0(d)}{h_1(d)}\right)}^{\!\! i}
		.
	\]
	Assume we can find at least $\abs{U}+1$ values for $d$ such that
	for all values the ratios ${h_0(d)} / {h_1(d)}$ are pairwise different.
	After computing $\Hol{G_d}$ for these values of $d$
	we can recover the value of each $A_i$.
	By \cref{eq:count:holantOriginalGraph}, we can finally output the value of $\Hol{G}$.

	It remains to construct the gadgets $H_d$ and to find the values for $d$.
	We construct the gadgets in a way such that
	there are constants $F_1$, $F_2$, and $F_3$ only depending on $\Ex$ with
	\begin{align*}
		h_0(d) &\deff F_0 \cdot h_0(d-1) + F_1 \cdot h_1(d-1) &  h_0(1) &\deff F_0\\
		h_1(d) &\deff F_1 \cdot h_0(d-1) + F_2 \cdot h_1(d-1) &  h_1(1) &\deff F_1
		.
	\end{align*}
	Given these properties of $H_d$,
	we use \cref{prop:count:sequenceWithoutRepetitions}
	to find sufficiently many values for $d$.
	\begin{proposition}
		[Special Case of Proposition~7.7 in \cite{MarxSS21}]
		\label{prop:count:sequenceWithoutRepetitions}
		Given three constants $F_0$, $F_1$, and $F_2$ with
		$F_0 F_2 \neq (F_1)^2$ and $F_0,F_1 \neq 0$.
		Let $\{A_n\}_{n\in \SetN}, \{B_n\}_{n\in \SetN}$
		be two sequences with
		\begin{align*}
			\begin{bmatrix}
				A_n \\
				B_n
			\end{bmatrix}
			= M \cdot
				\begin{bmatrix}
					A_{n-1} \\
					B_{n-1}
				\end{bmatrix}
			= M^n \cdot U
			\quad
			\text{ where }
			\quad
			M
			= \begin{bmatrix}
					F_0 & F_1 \\
					F_1 & F_2
				\end{bmatrix}
			\quad
			\text{ and }
			\quad
			U
			= \begin{bmatrix}
					A_0 \\
					B_0
				\end{bmatrix}
			= \begin{bmatrix}
					F_0 \\
					F_1
				\end{bmatrix}
			.
		\end{align*}
		Then $\{{A_n} / {B_n}\}_{n\in \SetN}$ is a sequence which does not contain any repetitions.
	\end{proposition}
	\begin{proof}
		By Proposition~7.7 in \cite{MarxSS21}
		we need to show that $M$ is invertible and $U$ is not an eigenvector of $M$.
		By assumption, $M$ is obviously invertible.
		Now, assume that $U$ is an eigenvector of $M$.
		Then, there is some $\lambda\neq 0$ such that
		\(
			(F_0)^2 + (F_1)^2 = \lambda F_0
		\)
		and
		\(
			F_1(F_0+F_2) = \lambda F_1
		\).
		This implies
		\(
			F_1(F_0)^2 + (F_1)^3 = F_1(F_0)^2+ F_0F_1F_2
		\)
		which contradicts the assumptions.
	\end{proof}
	As a last step we construct the $H_d$ gadgets.

	\proofsubparagraph*{\boldmath Case 1: $0\in \Ex$ or $1\not\in \Ex$.}
	We first show how to get a
	$\HWin[1]{\{0,1\}}$ node.
	\begin{itemize}
		\item If $0,1\not\in \Ex$, then any vertex with a dangling edge
		acts as a $\HWin[1]{\{0,1\}}$ node.
		\item If $0\in \Ex,1\not\in \Ex$, then attach $\max \Ex+1$ pendant
		vertices to any vertex $v$. Then, $v$ acts as a $\HWin[1]{\{0,1\}}$ node.
		\item If $0\in \Ex,1\in \Ex$, then take a clique of size $\min(\Exbar)+1$
		and split the edge between two vertices into two dangling edges.
		This now acts as a $\HWeq[2]{2}$ node.
		Attaching $\ceil{ ({\max \Ex+1}) / {2}}$ copies of such a $\HWeq[2]{2}$ node
		to a new vertex with one dangling edge gives us a $\HWin[1]{\{0,1\}}$ node.
	\end{itemize}

	\begin{figure}[t]
		\centering
		\includegraphics{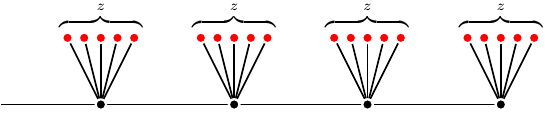}
		\caption{Gadget to realize $\wtnode{x,y}$ nodes in Case 1. Red nodes are $\HWin[1]{\{0,1\}}$ nodes.}
		\label{fig:count:wtnode-removal:case1}
	\end{figure}

	The gadget $H_d$ consists of a path of $d$ vertices
	with a dangling edge on the first vertex.
	For an integer $z\geq \max \Ex+1$ that we choose later,
	attach $z$ pendant $\HWin[1]{\{0,1\}}$ nodes to each
	vertex in the path.
	See \cref{fig:count:wtnode-removal:case1} for an illustration.
	By this definition, we get
	\begin{align*}
		F_0 = \sum_{i\geq 0: i  \in \Exbar} \binom{z}{i},\qquad
		F_1 = \sum_{i\geq 0: i+1\in \Exbar} \binom{z}{i},\qquad
		\text{and}\qquad
		F_2 = \sum_{i\geq 0: i+2\in \Exbar} \binom{z}{i}.
	\end{align*}

	We claim that there is a $z$
	such that the assumptions from \cref{prop:count:sequenceWithoutRepetitions} hold.
	If we can choose $z$ larger than $\max\Ex+1$,
	then $F_0$, $F_1$, and $F_2$ are never equal to 0.
	For contradictions sake, suppose that $F_0F_2 = (F_1)^2$.
	We first expand the equations above. Then for every $z$,
	\begin{align*}
		\left( \sum_{i\in \Exbar} \binom{z}{i}\right)\left( \sum_{i+2\in \Exbar} \binom{z}{i}\right)
		&= \Bigg( \sum_{i+1\in \Exbar} \binom{z}{i}\Bigg)^2\\
		\left( 2^z-\sum_{i\in \Ex} \binom{z}{i}\right) \left( 2^z- \sum_{i+2\in \Ex} \binom{z}{i}
		\right) &= \Bigg( 2^z-\sum_{i+1\in \Ex} \binom{z}{i}\Bigg)^2
	\end{align*}
	which implies $2^z Q_1(z) = Q_2(z)$,
	where
	\begin{align*}
		Q_1(z)&=\left (2\sum_{i+1\in \Ex} \binom{z}{i} -
		\sum_{i\in \Ex} \binom{z}{i} -\sum_{i+2\in \Ex} \binom{z}{i} \right)\\
		Q_2(z)&=\Bigg(\sum_{i+1\in \Ex} \binom{z}{i} \Bigg)^{2}
		-\left (\sum_{i\in \Ex} \binom{z}{i}\right )
		\left( \sum_{i+2\in \Ex} \binom{z}{i}  \right).
	\end{align*}
	For large enough $z$, we argue that $Q_1(z)$ is not identically zero.
	Observe that the second term in $Q_1(z)$ gives a non-zero $z^{\max \Ex}$
	monomial whereas the other two terms cannot give a monomial of this degree.
	Now, since $\Ex$ is a fixed, finite set, $Q_1$ and $Q_2$ are polynomials with constant degree.
	Thus, $Q_1(z)$ is zero only for finitely many $z$. Hence,
	there are infinitely many (positive) $z$ such that $Q_1(z)$ is non-zero.
	For each such $z$ we have
	\begin{align*}
		\abs{2^z Q_1(z)} &= \abs{Q_2(z)}.
	\end{align*}
	This is immediately a contradiction since $2^z \in \omega(z^c)$, for any constant $c$,
	if $z$ is large enough.
	Thus, there is some positive integral value of $z$ such that $F_0F_2\neq (F_1)^2$. We use this value of $z$
	in the construction of the gadget.
	Note that $z$ only depends on $\Ex$ and can thus be precomputed.

	\proofsubparagraph*{\boldmath Case 2: $0\not\in \Ex, 1\in \Ex$.}

	In this case, we first get an $\EQ{2}$ node in the following way.
	We create a
	clique of size $\min(\Exbar\setminus\{0\})+1$ and split one edge into two dangling edges.

	\begin{figure}[t]
		\centering
		\includegraphics{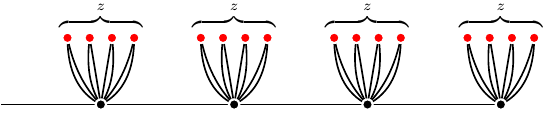}
		\caption{Gadget to realize $\wtnode{x,y}$ nodes in Case 2. Red nodes are $\EQ{2}$ nodes.}
		\label{fig:count:wtnode-removal:case2}
	\end{figure}

	For the gadget $H_d$, we create a path of length $d$ and add a
	dangling edge to the first vertex. To each vertex in the path,
	attach $z$ many $\EQ{2}$ nodes via both of its dangling edges,
	where $z$ is chosen later such that $2z\geq \max \Ex+1$.
	See \cref{fig:count:wtnode-removal:case2} for an illustration.
	We get
	\begin{align*}
		F_0 = \sum_{i\geq 0: 2i  \in \Exbar} \binom{z}{i}, \qquad
		F_1 = \sum_{i\geq 0: 2i+1\in \Exbar} \binom{z}{i}, \qquad
		\text{and}\qquad
		F_2 = \sum_{i\geq 0: 2i+2\in \Exbar} \binom{z}{i}.
	\end{align*}
	Suppose there is no $z$ such that $F_0F_2 \neq (F_1)^2$.
	Then, for every $z$,
	\begin{align*}
		\left( \sum_{2i\in \Exbar} \binom{z}{i}\right)\left( \sum_{2i+2\in \Exbar} \binom{z}{i}\right)
		&= \Bigg( \sum_{2i+1\in \Exbar} \binom{z}{i}\Bigg)^2\\
		\left( 2^z-\sum_{2i\in \Ex} \binom{z}{i}\right) \left( 2^z- \sum_{2i+2\in \Ex} \binom{z}{i}
		\right) &= \Bigg( 2^z-\sum_{2i+1\in \Ex} \binom{z}{i}\Bigg)^2
	\end{align*}
	which implies $2^z Q_1(z) = Q_2(z)$,
	where
	\begin{align*}
		Q_1(z)&=\left (2\sum_{2i+1\in \Ex} \binom{z}{i} -
		\sum_{2i\in \Ex} \binom{z}{i} -\sum_{2i+2\in \Ex} \binom{z}{i} \right)\\
		Q_2(z)&=\Bigg(\sum_{2i+1\in \Ex} \binom{z}{i} \Bigg)^{2}
		-\left (\sum_{2i\in \Ex} \binom{z}{i}\right )
		\left( \sum_{2i+2\in \Ex} \binom{z}{i}  \right).
	\end{align*}
	As before, $Q_1(z)$ is not identically zero for large enough $z$.
	The first term of $Q_1(z)$ gives a constant of $2$,
	the second term cannot give a constant since $0\not\in \Ex$,
	and the third term gives a constant of $-1$ if $2\in\Ex$.
	Thus, $Q_1(z)$
	has a non-zero constant term and hence, is zero only for finitely many $z$.
	Therefore, for
	infinitely many positive values of $z$, we have
	\begin{align*}
		\abs{2^zQ_1(z)}=\abs{Q_2(z)}
		.
	\end{align*}
	This is a contradiction
	since $2^z \in \omega(z^c)$, for any constant $c$,
	and $Q_1$ and $Q_2$ are polynomials
	of constant degree.
	Thus, we can choose a $z$ such that $F_0F_2\neq (F_1)^2$ to do
	the interpolation.
\end{proof}

\begin{lemma}\label{lem:count:wtnode-interpolation}
	Let $\Ex\subseteq \SetN$ be a fixed, finite set
	with $\Ex \not\subseteq \{0\}$.
	Let $R_1,\dots,R_d$ be $d$ arbitrary relations for some $d \ge 0$.
	There is a polynomial-time Turing reduction from
	\[
	\Hol{R_1,\dots,R_d,\HWin{\Exbar}, \wtnode{1,1}, \wtnode{-1,1}, \wtnode{0,1}}
	\text{ to } \Hol{R_1,\dots,R_d,\HWin{\Exbar}}
	\]
	such that
	$\totalDeg$ decreases
	and $\pw$ increases to $\pw+\totalDeg \cdot f(\max \Ex)$.
\end{lemma}
\begin{proof}
	We first use \cref{lem:count:wtnode-removal} to remove the $\wtnode{1,1}$ nodes.
	Observe that this can alternatively be done by a simple construction
	using a fresh vertex with $\max\Ex+1$ forced edges.
	Then, we apply \cref{lem:count:wtnode-removal} two more times
	to remove the $\wtnode{-1,1}$ nodes
	and finally the $\wtnode{0,1}$ nodes.
\end{proof}
Now, we can prove the reduction from \CountAntiFactorR{\Y} to \CountAntiFactor{(\Ex,\Y)}.
\begin{proof}[Proof of \cref{lem:count:AntiFactorRToAntiFactor}]
	By the reduction of \cref{lem:count:relationsToForcingEdges},
	we can reduce \CountAntiFactorR{\Y}
	to \Hol{\HWin{\overline \Y},\HWeq{1}}.
This can trivially be reduced to
\Hol{\HWin{\overline \Y},\HWin{\Exbar},\HWeq{1}}
as we do not have any vertices with relation \HWin{\Exbar}.
Then, we invoke \cref{%
lem:count:forcingEdges,%
lem:count:wtnode-interpolation}
such that the vertices with relation \HWin{\overline \Y} are
not changed (or used for any construction).
By this we end the reduction with \Hol{\HWin{\overline \Y}, \HWin{\Exbar}}
which precisely corresponds to \CountAntiFactor{(\Ex,\Y)}.
\end{proof}

%% file: img/count/cofinite-chain-of-reduction.tex
\centering
\begin{tikzpicture}[
problem/.style={rectangle, draw=black, text=black, minimum height=0.75cm,
minimum width=4.5cm
},
red/.style
  = {draw=black, -latex,thick},
lem/.style={text=black,midway}
]
  \node[problem,thick] (antFacR) at (0, 0)
    {\CountAntiFactorR{\Y}};

  \node[problem,thick] (antFac) [below=of antFacR]
    {\Hol{\HWin{\overline\Y}, \HWin{\Exbar}}};

  \node[problem, minimum width=6cm] (HWmixed1) [right=3cm of antFac]
    {\Hol{\HWin{\overline\Y},\wtnode{1,1},\wtnode{-1,1},\wtnode{0,1}}};

  \node[problem, minimum width=6cm] (HWone) at (HWmixed1 |- antFacR)
        {\Hol{\HWin{\overline\Y},\HWeq{1}}};

  \node[problem, minimum width=6cm] (HWmixed2) [below=of HWmixed1]
        {\Hol{\HWin{\overline\Y}, \HWin{\Exbar},\wtnode{-1,1},\wtnode{0,1}}};

  \node[problem] (HWmixed3) at (antFac |- HWmixed2)
        {\Hol{\HWin{\overline\Y}, \HWin{\Exbar},\wtnode{0,1}}};

  \draw[red] (antFacR) -- (HWone)
    node[lem, above]
    {\cref{lem:count:relationsToForcingEdges}};

  \draw[red] ($(HWone.south) - (3mm,0)$)
    -- ($(HWmixed1.north) - (3mm,0)$);
  \draw[red] (HWone) -- (HWmixed1);
  \draw[red] ($(HWone.south) + (3mm,0)$)
    -- ($(HWmixed1.north) + (3mm,0)$)
    node[lem, right, align=left]
    {\cref{lem:count:forcingEdges:case1,lem:count:forcingEdges:case2,lem:count:forcingEdges:case3}\\
    (\cref{lem:count:forcingEdges})
    };

  \draw[red] (HWmixed1) -- (HWmixed2)
    node[lem, right]
    {\cref{lem:count:wtnode-removal}};

  \draw[red] (HWmixed2) -- (HWmixed3)
    node[lem, above]
    {\cref{lem:count:wtnode-removal}};

  \draw[red] (HWmixed3) -- (antFac)
    node[lem, right]
    {\cref{lem:count:wtnode-removal}};

  \draw[red,dotted] (HWmixed1) -- (antFac)
    node[lem, above]
    {\cref{lem:count:wtnode-interpolation}};

  \draw[red,dotted] (antFacR) -- (antFac)
    node[lem, right]
    {\cref{lem:count:AntiFactorRToAntiFactor}};

\end{tikzpicture}

%% file: cofinite/ecover.tex
\section{Lower Bound for Counting Edge Covers}\label{sec:lower:edge-cover}
Observe that the construction in the previous section
does not apply for the case $\Ex=\{0\}$,
which precisely corresponds to \countECover.
Due to its special structure, we show a different reduction
in this section
which still leads to the expected hardness result.
\begin{theorem}
	\label{thm:lower:count:hardnessEdgeCover}
	For every constant $\epsilon>0$,
	no algorithm can solve
	\countECover in time $\Ostar{(2-\epsilon)^\pw}$
	even if a path decomposition of width $\pw$ is given,
	unless \#SETH fails.
\end{theorem}
The main step towards proving this lower bound
is to show a lower bound for \countMIS.
\begin{definition}[Counting Maximum Independent Sets (\countMIS)]
	Given a graph $G=(V,E)$.
	Let $\mathcal I=\{I\subseteq V \mid\forall\{u,v\}\in E: u\in I \lor v\in I\}$
	and $M = \max_{I\in\mathcal I} \abs{I}$.

	Compute $\abs{\{ I \in\mathcal I \mid M = \abs{I}\}}$.
\end{definition}
We split the lower bound for \countECover into the following steps.

\begin{enumerate}
	\item
	\cref{lem:hardnessSatCountingSat} reduces the counting version of SAT
	to a variant of SAT where we want to count the satisfying assignments
	given the promise that the formula is satisfiable.

	\item
	In \cref{lem:lower:count:hardnessCountMIS} we show a lower bound
	for \countMIS on low degree graphs under \#SETH.

	\item
	\cref{lem:lower:count:hardnessCountMISregular}
	extends the previous lower bound to regular graphs.

	\item
	To prove the lower bound in \cref{thm:lower:count:hardnessEdgeCover},
	we reduce \countMIS on regular graphs to \countECover
	using ideas from \cite{BordewichDK08,BubleyD97}.
\end{enumerate}
We start with the reduction from \#SAT to the variant of \#SAT.
\begin{lemma}\label{lem:hardnessSatCountingSat}
	Given a SAT formula $\phi$ with $n$ variables and $m$ clauses,
	we can construct in time linear in the output size a SAT formula
	$\psi$ on $n+1$ variables and $m+n$ clauses
	such that $\#SAT(\phi)+1=\#SAT(\psi)$.
\end{lemma}
\begin{proof}
	Let $x_1,\dots,x_n$ be the variables of $\phi$
	and $x_0$ be a free variable.
	We define $\psi$ such that $\psi \equiv (x_0 \lor \phi) \land (\neg x_0 \lor(x_1 \land \dots \land x_n))$.
	One can easily see that the assignment $x_0=x_1=\dots=x_n=\true$ is satisfying for $\psi$.
	Further, for $x_0=\true$ this is the only satisfying assignment.
	For $x_0=\false$
  the satisfying assignments of $\phi$ directly transfer to $\psi$.

	One can easily transform the right-hand side into a CNF
	by adding $x_0$ to each clause of $\phi$
	and additionally adding, for all $i\in[n]$, the clauses $(\neg x_0 \lor x_i)$.
\end{proof}
Lokshtanov, Marx and Saurabh have shown a
lower bound of $\Ostar{(2-\epsilon)^\pw}$
for \emph{finding some} maximum independent set on graphs with
pathwidth $\pw$ under SETH \cite{LokshtanovMS18}.
However, their reduction is not parsimonious
and, therefore, does not hold under \#SETH.
We strengthen this result by showing a lower bound for \countMIS
based on the weaker assumption of \#SETH.

\begin{lemma}\label{lem:lower:count:hardnessCountMIS}
	For every constant $\epsilon>0$,
	there is an $r>0$ such that
	\countMIS on graphs with maximum degree $r$
	cannot be solved in time $\Ostar{(2-\epsilon)^\pw}$,
	even if we are given a path decomposition of width $\pw$,
	unless \#SETH fails.
\end{lemma}
\begin{proof}
	Let $\phi$ be a SAT instance with $n$ variables and $m$ clauses.
	By \cref{lem:hardnessSatCountingSat}, we can assume
	that $\phi$ is satisfiable.
	We construct a graph $G_\phi$ which mainly consists of two types of gadgets.
	The first type encodes the assignment and
	ensures that this information is transferred consistently
	while the second one checks if a clause is satisfied
	by the encoded assignment.
	The ideas for both types of gadgets is that
	the size of the maximum independent set decreases significantly
	if the properties of these gadgets are not fulfilled.
	For this we bound the size of the maximum independent set
	for each type of gadget separately.
	Finally,
	we argue that the number of maximum independent sets of $G_\phi$
	corresponds to the number of satisfying assignments of the formula $\phi$.

	\proofsubparagraph*{Encoding Assignments.}
	We group $g$ variables together, where $g$ is chosen later.
	For each group, we encode the partial assignments by subsets of vertices such that
	each partial assignment corresponds to some $S \subseteq [d]$ of size $d/2$,
	where $d$ is also chosen later as an even integer.
	The high-level construction has a grid-like structure:
	There are $t \deff \ceil{n/g}$ rows,
	one for each group of variables,
	and $m$ columns, one for each clause.
	See \cref{fig:ecover:mis:high-level} for an illustration of the construction.

	We repeat the following construction
	for every $i\in[t]$, denoting the group,
	and every $j\in[m]$, denoting the clause.
	We create a gadget $D_i^j$
	with $d$ new input vertices	$v_{i,1}^j,\dots,v_{i,d}^j$
	which ensures the consistency
	of the encoded assignments.

	We encode the partial assignments to the variables of the $i$th group
	by subsets of these vertices (or their indices from $[d]$).
	To each assignment we assign a unique subset of $[d]$ of size exactly $d/2$.
	This is possible as we set $d$
	such that $\binom{d}{d/2} \ge 2^g$.%
	\footnote{
	Note that for different $i$, we could use in principle different mappings.
	}
	We assign the all zeros assignment
	to the remaining subsets of $[d]$ of size exactly $d/2$.
	Then, every subset of $[d]$ of size $d/2$
	or rather the corresponding selection of the input vertices
	encodes a partial assignment for the $i$th group.

	To ensure the assignment is consistent for all $j$,
	we add $d$ cliques to $D_i^j$, each with $\binom{d}{d/2}$ vertices.
	Each vertex of such a clique is identified with a unique subset of size $d/2$ from $[d]$.
	We denote the vertex for the set $S$ in the $\ell$th clique by $w_{i,S,\ell}^j$.
	For all $S\subseteq[d]$ with $\abs{S}=d/2$
	and all $\ell \in[d]$,
	we connect $w_{i,S,\ell}^j$ to the vertices
	$v_{i,k}^j$ and $v_{i,k}^{j+1}$, for all $k\in [d]\setminus S$.
	\footnote{
	For $j=m$, the vertices are connected only to the first set of vertices,
	as the vertices $v_{i,k}^{m+1}$ do not exist.
	}
	That is,
	if the vertices $\{v_{i,k}^j \mid k \in S\}$
	are selected,
	where $S\subseteq [d]$ and $\abs{S}=d/2$,
	then we can also select the vertices $\{w_{i,S,\ell}^j\}_\ell$
	by the definition of the graph.

	\begin{figure}
		\centering
		\includegraphics[width=\textwidth]{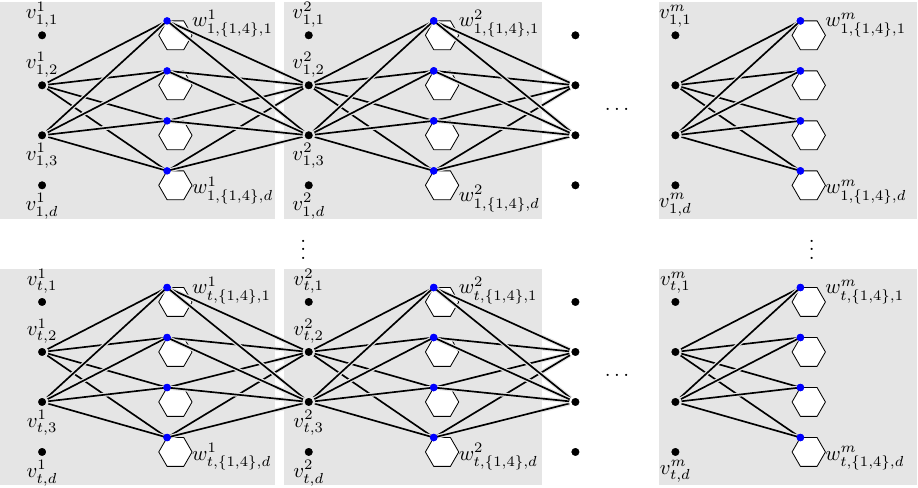}
		\caption{
		High-Level construction of the graph for $d=4$.
		The hexagons represent $6$-cliques
		where the blue vertex corresponds to the vertex of set $S=\{1,4\}$.
		}
		\label{fig:ecover:mis:high-level}
	\end{figure}

	\begin{claim}
		\label{clm:ecover:mis:misInSingleConsistencyGadget}
		The maximum independent sets of $D_i^j$ have size exactly $d/2+d$.
	\end{claim}
	\begin{claimproof}
		For arbitrary $S \subseteq [d]$ with $\abs{S}=d/2$
		(depending on the assignment),
		we select the vertices $\{v_{i,k}^j \mid k \in S\}$
	and the vertices $w_{i,S,\ell}^j$, for all $\ell\in[d]$.

	Now, assume that there is a independent set
	where more than $d/2+d$ vertices are selected.
	Trivially each of the $d$ cliques can contribute at most one vertex to the total number.
	Hence, there must be more than $d/2$ input vertices selected.
	Then, all sets of $d/2$ input vertices intersect with the set of selected input vertices.
	Thus, all vertices of the cliques are adjacent to at least one selected input vertex.
	As a consequence this is not an independent set
	and the bound is tight.
	\end{claimproof}

	As a next step, we consider, for a fixed $i$, all $D_i^j$ together.
	\begin{claim}
		\label{clm:ecover:mis:misInMultipleConsistencyGadgets}
		Let $i\in[t]$ be fixed.
		The size of the maximum independent set in the graph induced by $D_i^1,\dots,D_i^m$
		is exactly $(d/2+d)m$
		if and only if
		for each $D_i^j$ the same subset of vertices
		(with respect to the indices) is selected.
		Otherwise the size of the maximum independent set is smaller.
	\end{claim}
	\begin{claimproof}
		The ``if'' follows directly by the definition of the graph
		and by bounding the size for each $C_i^j$ individually
		using \cref{clm:ecover:mis:misInSingleConsistencyGadget}.

		For the ``only if'' direction
		we are given an maximum independent set of size $(d/2+d)m$.
		Assume there is an $j\in[m]$
		and some $S,S' \subseteq [d]$ with $\abs{S}=\abs{S'}=d/2$ and $S\neq S'$
		such that
		the vertices $\{v_{i,k}^j \mid k\in S\}$
		and the vertices $\{v_{i,k}^{j+1} \mid k\in S'\}$ are selected.
		Thus, there is some $k \in S' \setminus S$
		such that $w_{i,S,\ell}^j$ is connected to $v_{i,k}^{j+1}$.
		As we are given an independent set,
		either none of the vertices $w_{i,S,\ell}^j$ is selected
		or $v_{i,k}^{j+1}$ is not selected.
		Therefore, the size can be at most $(d/2+d)m-1$
		by using the upper bound from \cref{clm:ecover:mis:misInSingleConsistencyGadget}
		for the remaining gadgets.
	\end{claimproof}
	By \cref{clm:ecover:mis:misInMultipleConsistencyGadgets}, for a fixed $i$,
	the assignment is actually the same for each $D_i^j$
	when considering a maximum independent set.
	We say that the partial assignment is consistently encoded
	if the maximum independent set has the size
	from \cref{clm:ecover:mis:misInMultipleConsistencyGadgets}.

	\begin{figure}
		\centering
		\includegraphics[page=1,width=\textwidth]{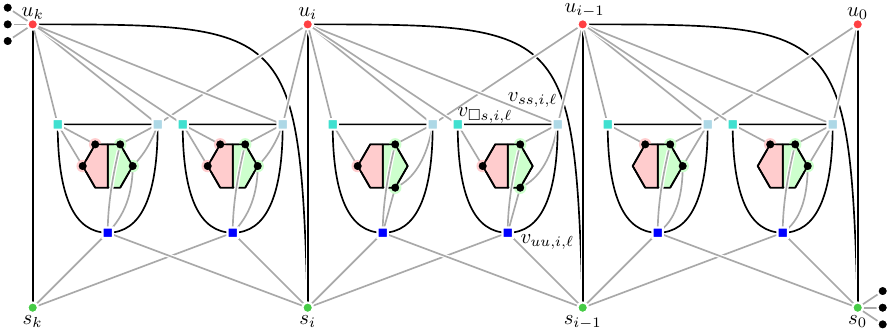}
		\caption{The clause gadget from the lower bound for \countMIS
		with $d=2,k=3$.
		The split hexagons correspond to the $d$ cliques of each group
		where the top part contains the $w_{i,S,\ell}$ with not satisfying $S$
		and the bottom part the ones with satisfying $S$.
		Black and gray edges connect vertices of the same and different type,
		respectively.
		}
		\label{fig:lower:mis:clause}
	\end{figure}

	\proofsubparagraph*{Encoding the Clauses.}
	The clause gadget checks whether a specific clause
	is satisfied by the encoded partial assignments.
	See \cref{fig:lower:mis:clause} for an illustration of the construction.
	For the construction we fix some clause $C_j$ of $\phi$.
	We omit $j$ as superscript from all vertices in the following
	to simplify notation.
	Without loss of generality we assume that $C_j$
	contains only variables from the first $k$ groups.

	The gadget consists of three types of vertices:
	\begin{enumerate}
	\item
	We create vertices $u_0,\dots,u_k$ and $s_0,\dots,s_k$
	and connect each pair $u_i,s_i$ by an edge.
	Moreover, for all $i\in[k]$, we connect $s_{i-1}$ to $u_i$.
	The idea is that the vertices $s_i$ represent the ``satisfied'' state
	while the vertices $u_i$ represent the ``unsatisfied'' state.

	\item
	We repeat the following procedure for all $i\in[k]$.
	We introduce vertices $v_{ss,i,\ell}$,
	$v_{uu,i,\ell}$,
	and $v_{\square s,i,\ell}$
	which are pairwise connected for each $\ell \in [d]$,
	i.e., for each $\ell$ the vertices form a triangle.
	We make $v_{ss,i,\ell}$ adjacent to $u_{i-1}$ and $u_i$.
	Likewise, we connect $v_{uu,i,\ell}$ to $s_{i-1}$ and $s_i$.
	Additionally, we connect $v_{\square s,i,\ell}$ to $u_i$.

	For all $\ell\in[d]$,
	we connect the vertices $v_{uu,i,\ell}$ and $v_{ss,i,\ell}$
	to $w_{i,S,\ell}$
	for each $S \subseteq [d]$ with $\abs{S}=d/2$
	that represent a \emph{satisfying} partial assignment.
	Conversely,
	for all $\ell\in [d]$,
	we connect the vertex $v_{\square s, i,\ell}$ to $w_{i,S,\ell}$
	for each $S$ representing an \emph{unsatisfying} partial assignment.

	The idea is as follows.
	The vertices $v_{uu,i,\ell}$ are selected
	if the clause is not yet satisfied
	and the clause is also not satisfied by the assignment for the $i$th group.
	The vertices $v_{ss,i,\ell}$ are selected
	if the clause is already satisfied
	but not satisfied by the assignment for the $i$th group.
	The vertices $v_{\square s,i,\ell}$ are selected
	if the clause is satisfied by the assignment for the $i$th group
	independent from whether the clause is already satisfied.

	\item
	We make $s_0$ and $u_k$ adjacent to $d$ new vertices, each.
	This ensures that in a maximum independent set the vertices $s_0$ and $u_k$ are not selected.
	Instead,
	the vertices $u_0$ and $s_k$ should be selected.
	Which corresponds to the case 
	that the clause is initially not satisfied
	but eventually is satisfied.
	\end{enumerate}
	For the proof of correctness, we only count the vertices that are newly introduced.
	That is, we do not count the vertices $w_{i,S,\ell}$.

	\begin{claim}
		If the assignment is consistently encoded and satisfies the clause,
		then
		there is a unique maximal independent set with
		$(k+1) + (kd) + (2d)$ vertices.
	\end{claim}
	\begin{claimproof}
		We select $k+1$ vertices of the first type,
		$kd$ vertices of the second type,
		and $2d$ vertices of the third type.
		Bounding the size of the maximum independent set for each type of vertices separately,
		shows that no independent set can be of larger size.

		By these observations,
		we must select all vertices of the third type.
		As we must select, for all $i\in[0,k]$, either $s_i$ or $u_i$,
		the vertex $u_0$ must be selected.
		We select the other vertices in rounds.

		If $u_{i-1}$ is selected,
		we cannot select $v_{ss,i,\ell}$ for the independent set.
		By the above observation
		and since we want to construct an independent set,
		we must select either all $v_{uu,i,\ell}$ vertices, for $\ell\in[d]$,
		or all $v_{\square s,i,\ell}$ vertices, for $\ell\in[d]$.
		Observe that we can only select the vertices $v_{\square s,i,\ell}$
		if the assignment for group $i$ is satisfying
		because these vertices are only connected to the vertices
		that represent an \emph{unsatisfying} assignment.
		Conversely,
		the vertices $v_{uu,i,\ell}$ can only be selected if the assignment is unsatisfying,
		as they are connected to vertices representing a \emph{satisfying} assignment.

		If $s_{i-1}$ is selected,
		we cannot select $v_{uu,i,\ell}$ for the independent set
		because of the adjacency of these two vertices.
		We know that $s_{i}$ must also be selected, as $u_i$ cannot be selected.
		Therefore,
		either all $v_{ss,i,\ell}$ vertices, for $\ell\in[d]$,
		or all $v_{\square s,i,\ell}$ vertices, for $\ell\in[d]$, are selected.
		Depending on the assignment which we are given,
		we can apply the same argument as before.

		By this procedure, we eventually arrive at a point,
		where either $u_k$ or $s_k$ must be selected.
		Since the clause is satisfied by the encoded assignment,
		there is some $c\in[k]$ such that $u_{c-1}$ and $s_c$ are selected.
		Thus, $s_k$ is selected
		and we can also select the remaining vertices of the third type,
		namely the ones adjacent to $u_k$.
	\end{claimproof}
	\begin{claim}
		Assume the assignment is consistently encoded.
		If the maximum independent set has size $(k+1) + (kd) + (2d)$,
		then this assignment satisfies the clause.
	\end{claim}
	\begin{claimproof}
		From the definition of the graph, we get
		that $k+1$ vertices of the first type,
		$kd$ vertices of the second type,
		and $2d$ vertices of the third type must be selected.

		Therefore, $u_0$ and $s_k$ must be selected.
		Since, for all $i\in[k]$, either $u_i$ or $s_i$ must be selected
		and
		the fact that the vertices $s_{i-1}$ and $u_{i}$ are adjacent,
		there must be some $c\in[k]$
		such that $u_0,\dots,u_{c-1}$ and $s_c,\dots,s_k$ are selected.
		This especially implies that none of the vertices
		$v_{uu,c,\ell}$ and $v_{ss,c,\ell}$, for any $\ell\in[d]$, can be selected
		because they are connected to $s_c$ and $u_{c-1}$, respectively.
		Therefore, the vertices $v_{\square s,c,\ell}$ must be selected.
		This is only possible if,
		for some $S$ corresponding to a satisfying assignment,
		the vertices $w_{c,S,\ell}$ are selected.
		Thus,
		the assignment satisfies the clause.
	\end{claimproof}

	\proofsubparagraph*{Analysis of the Construction.}
	For the final part of the lower bound,
	it remains to analyze the pathwidth of the graph $G_\phi$ and to bound the largest degree.
	Both properties depend on the parameters $d$ and $g$ we still need to set.

	\begin{claim}
		The pathwidth of the graph $G_\phi$ is bounded by
	\(
	\ceil{{n}/{g}} \cdot d + \binom{d}{d/2} + d + \O(1).
	\)
	\end{claim}
	\begin{claimproof}
	We use a mixed search approach.
	For this reason, we iterate over all clauses and
	clean the clause gadget and the corresponding columns
	starting from the first row (which corresponds to the first group of variables).
	We place $d$ searchers on the input vertices
	of the next gadget
	and then clean the cliques and clause gadget sequentially
	by cleaning repetition by repetition.
	As the cliques are not connected to each other,
	this requires only $\binom{d}{d/2}+\O(1)$ searchers.
	Moreover, the cliques are only connected to vertices that already have a searcher on them,
	that is, the vertices $v_{i,k}^j$,
	$v_{ss,i,\ell}$, $v_{uu,i,\ell}$, and $v_{\square s,i,l}$,
	for a fixed $\ell$, as we place searcher simultaneously on them.
	Thus, we can clean the clause gadgets with $\O(1)$ searchers.
	\end{claimproof}
	The degree of the vertices is bounded by $\binom{d}{d/2}+2d+\O(1)$.
	We choose the value for the parameter $r$ from the lemma statement
	to be precisely this value.
	As we see in the remaining proof,
	the value of $d$ depends on $\epsilon$ only.

	\proofsubparagraph*{Lower Bound.}
	Assume we are given an algorithm for \countMIS
	with running time $\Ostar[N]{(2-\eps)^\pw}$ for a graph with $N$ vertices
	and maximum degree $r$
	for some $\eps>0$.
	We use this algorithm to count the number of satisfying assignments
	of a SAT formula $\phi$ in time $\Ostar{(2-\delta)^n}$
	where $\delta>0$.

	We choose an $\alpha>1$ such that $(2-\eps)^{\alpha} \le (2-\delta)$
	for some suitable $\epsilon>\delta>0$.
	For ease of notation, we define
	$\eps' = \log(2-\eps)$ and $\delta'=\log(2-\delta)$.
	Hence, we have that $\eps' \cdot \alpha \le \delta'<1$.
	Now, we choose some $0<\beta<1$ such that $\alpha \beta >1$.
	To sue the construction from above,
	we set $d$ to be an even integer large enough such that
		$g \deff \floor{\beta d} > 0$,
		$d \le \alpha \floor{\beta d}$,
		and
		$d \le 2^{(1-\beta)d +1}$.

	We first show that we can encode all partial assignments
	by our choice of parameters.
	For this, it suffices to bound $\binom{2b}{b}$ from below by $4^b/b$
	which follows directly by induction starting from $b=4$.
	Then, we get
	\[
		\binom{d}{d/2} \ge \frac{2^{d+1}}{d} \ge 2^{\beta d} \ge 2^g.
	\]
	Now, for the formula $\phi$, we construct the graph $G_\phi$ as given above
	and count the number of maximum independent sets.
	Since we know that $\phi$ is satisfiable,
	the number of maximum independent sets corresponds to the number of satisfying assignments.

	The running time of the whole procedure can be bound by
	\begin{align*}
		\Ostar[(n+m)]{(2-\eps)^\pw}
		&\le
		\Ostar[(n+m)]{(2-\eps)^{\ceil{\frac{n}{g}} \cdot d + \binom{d}{d/2} + d + \O(1)}} \\
		&\le
		\Ostar[(n+m)]{(2-\eps)^{\frac{n}{g} \cdot d + \binom{d}{d/2} + 2d +  \O(1)}}.
		\intertext{
		Observe that \(d\) does not depend on the input formula
		but just on \(\eps\).
		Thus, the term \((2-\eps)^{\binom{d}{d/2} + 2d + \O(1)}\)
		contributes only a constant factor to the overall running time
		and can be hidden by the \((n+m)^{\O(1)}\) term.
		}
		&\le
		\Ostar[(n+m)]{(2-\eps)^{\frac{n}{g} \cdot d}} \\
		&\le
		\Ostar[(n+m)]{2^{\eps'\frac{n}{\floor{\beta d}} \cdot d}} \\
		\intertext{
		By our choice of \(d\), \(\alpha\), and \(\beta\) we get
		}
		&\le
		\Ostar[(n+m)]{2^{\eps' \alpha n}}
		\le
		\Ostar[(n+m)]{2^{\delta' n}}
		\le
		\Ostar[(n+m)]{(2-\delta)^{n}}
		.
	\end{align*}
	This running time now immediately contradicts \#SETH and the claim follows.
\end{proof}
The next step extends the previous result to regular graphs.
\begin{lemma}\label{lem:lower:count:hardnessCountMISregular}
	For every constant $\epsilon>0$,
	there is an $r>0$ such that
	\countMIS on
	$r$-regular graphs given with a path decomposition of width $\pw$
	cannot be solved in time $\Ostar{(2-\epsilon)^{\pw}}$,
	unless \#SETH fails.
\end{lemma}
\begin{proof}
	Let $\epsilon>0$ be an arbitrary constant.
	By \cref{lem:lower:count:hardnessCountMIS},
	there is some $r'>0$ (just depending on $\epsilon$)
	such that \countMIS on graphs of maximum degree $r'$
	cannot be solved in time $\Ostar{(2-\epsilon)^\pw}$,
	unless \#SETH fails.
	Let $G$ be such a \countMIS instance.
	We modify $G$ such that we get a regular graph $H$.

	Set $r$ to be the odd number of $r'$ and $r'+1$.
	We first construct a gadget $J$ containing a distinguished vertex $v_J$
	as portal
	that is ``forced'' to \emph{not} be in any solution of
	the \countMIS instance.
	For this gadget, we first
	take a clique with $r+1$ vertices and remove an edge between two arbitrary vertices.
	Then, add an edge from both of these vertices to a new vertex $v_J$.
	We repeat this process $(r-1)/2-1$ more times resulting in $(r-1)/2$ modified cliques.
	We denote the resulting gadget by $J$
	and look at $v_J$ as the portal vertex of $J$.
	Observe that $v_J$ has degree $r-1$ whereas
	all other vertices in the gadget have degree $r$.
	Observe that $J$ has a unique maximum independent set
	of size $r-1$ which does not contain $v_J$,
	that is, the independent set where, for each of the $(r-1)/2$ modified cliques,
	the two non-adjacent vertices are selected.

	Now, for each vertex $u$ in $G$,
	we introduce $r-\deg(u)$ copies of the gadget $J$
	and make all copies of $v_J$ adjacent to $u$.
	Let $H$ denote the resulting $r$-regular graph.
	Observe that $H$
	retains all the maximum independent sets from $G$.

	The size of the graph $H$ can be bounded by $\O(nr^3)$
	and the pathwidth increases by an additive term of $r+2$
	compared to the pathwidth of $G$.
	Recall that $r$ depends on $r'$ which only depends on the fixed $\epsilon$.
	When running the claimed algorithm on the modified instance $H$,
	the change of the size and pathwidth
	contributes only a constant to the running time
	which asymptotically does not change the running time.
	By \cref{lem:lower:count:hardnessCountMIS}, this contradicts \#SETH.
\end{proof}
Now we have everything ready for the last step of the reduction
where we reduce from \countMIS on regular graphs to \countECover.

\begin{proof}[Proof of \cref{thm:lower:count:hardnessEdgeCover}]
	For a given $\epsilon>0$,
	we know from \cref{lem:lower:count:hardnessCountMISregular},
	that there is an $r>0$ such that
	\countMIS on $r$-regular graphs
	cannot be solved in time $\Ostar{(2-\epsilon)^\pw}$,
	unless \#SETH fails.
	Let $H=(V,E)$ be such an $r$-regular \countMIS instance.
	We follow the ideas from the \sharpP-completeness proof of \countECover in \cite{BordewichDK08}
	to obtain a Turing reduction from \countMIS to \countECover.

	Let $I_j(H)$ be the number of independent
	sets of size exactly $j$ in $H$.
	We subdivide all edges of $H$
	by placing a new vertex on each edge.
	Let $G$ be the resulting graph
	and let $U$ be the set of new vertices.
	Let $N_i(G)$ be the number of
	subsets $E' \subseteq E(G)$ such that $i$ vertices of $V$ are not
	covered and all vertices of $U$ are covered. We follow a similar argument
	as in \cite{BordewichDK08, BubleyD97} to count the number of subgraphs that
	contribute to $N_i(G)$.
	
	\begin{claim}
		\label{clm:hardnessEC:firstEquality}
		With the notation as before, it holds that
		\[
			3^{rn/2-jr} I_j(H) = \sum_{i=j}^{n} \binom{i}{j} N_i(G).
		\]
	\end{claim}
	\begin{claimproof}
		We prove the claim by analyzing how independent sets for $H$
		can be transformed into edge covers of $G$.
		The first formulation covers with the left-hand side
		and the second corresponds to the right-hand side.

	For some independent set $S \subseteq V(H)$ of size $j$ in $H$,
	consider some vertex $v\in S$.
	Let
	$(v,v_1),\ldots,(v,v_{r})$ be the incident edges of $u$.
	Suppose these edges are subdivided to get the edges
	$(v,u_1),(u_1,v_1),\dots,(v,u_r),(u_r,v_r)$
	for some $u_1,\ldots,u_r\in U$.
	We select none of $(v,u_1),\ldots,(v,u_r)$ and select all of $(u_1,v_1),\ldots,(u_r,v_r)$
	in the solution.
	We repeat the previous selection procedure for every $v\in S$.
	Now, observe that there are $rn/2-jr$ edges in $H$
	that are not incident to a vertex from $S$.
	For each of these edges,
	there is a unique vertex $u\in U$ subdividing this edge in $G$.
	We select at least one of the two subdividing edges incident to $u$;
	there are three possible ways to do this.
	The number of possible ways for this
	forms the left-hand side of the equation.

	For the right-hand side, observe that the above
	process creates every subset of edges in $G$
	where $S$ is not covered and all of $U$ is covered.
	If a selection $E' \subseteq E(G)$ leaves exactly $i$ vertices, say $S' \subseteq V$, uncovered in $G$,
	then the above process generates $E'$ only if $S\subseteq S'$.
	This means that there are exactly $\binom{i}{j}$ sets $S$ of size $j$
	for which the process creates $E'$.
	\end{claimproof}

	By \cref{clm:hardnessEC:firstEquality},
	it suffices to recover the values of $N_i(G)$
	to recover the values of $I_j(H)$.

	Let $N_{ij}(G)$ be the number of edge subsets of $G$
	that leave exactly $i$ vertices in $V$ and exactly $j$ vertices in $U$
	uncovered.
	We attach paths of length $m$ (i.e., with $m$ edges) to all
	vertices $v\in V$ and attach paths of length $k$ to all vertices $u\in U$. Let
	the resulting graph be $G_{m,k}$.
	For a path of length $\ell$ (i.e., with $\ell$ edges),
	let $M_\ell$ be the number of edge covers of such a path.
	Observe that $M_\ell = M_{\ell-1} + M_{\ell-2}$ with $M_1=M_2=1$
	which is precisely the definition of the Fibonacci numbers.

	Now, consider an edge subset $E'$ of $G$
	that leaves exactly $i$ vertices in $V$ and exactly $j$ vertices in $U$
	uncovered.
	(Note that $E'$ contributes towards $N_{ij}(G)$.)
	We analyze how $E'$ can be extended to an edge cover of $G_{m,k}$.
	For a vertex $v \in V \cup U$ that is already covered by some edge
	we can choose
	an edge cover of the path attached to $v$ that either covers $v$ or not.
	For a vertex $v \in V \cup U$ that is not covered by any edge we must choose an edge
	cover of the path attached to $v$ that also covers $v$.
	Denote by $\cEC(G_{m,k})$ the number of edge covers of $G_{m,k}$.
	Then, by the above observations, we get
	\begin{align*}
		\cEC(G_{m,k})
			&=\sum_{i=0}^{n} \sum_{j=0}^{rn/2}
								M_m^i(M_m+M_{m-1})^{n-i}M_k^j(M_k+M_{k-1})^{rn/2-j} N_{ij}(G)\\
			&=M_m^nM_k^{rn/2}\sum_{i=0}^{n} \Bigg( 1+\frac{M_{m-1}}{M_m}\Bigg)^{n-i}
				\sum_{j=0}^{n} \Bigg( 1+\frac{M_{k-1}}{M_k} \Bigg)^{rn/2-j} N_{ij}(G).
	\end{align*}
	We interpret $\cEC(G_{m,k})$ as a polynomial in the two variables
	$1+{M_{m-1}} / {M_m}$
	and $1+{M_{k-1}} / {M_k}$.
	We use the algorithm for \countECover
	to get the value of $\cEC(G_{m,k})$.
	Since the $M_\ell$s correspond to Fibonacci numbers,
	we know that ${M_\ell} / {M_{\ell-1}}$ takes infinitely many distinct values
	(though the sequence is converging).
	Hence, we can use interpolation to recover $N_{ij}(G)$
	for all values of $i$ and $j$.
	Since $N_i(G)=N_{i0}(G)$,
	we can recover the number of maximum independent sets of $H$,
	that is $I_{j^*}$ where $j^* = \max\{j' \mid I_{j'} \neq 0\}$.

	The only remaining task is to argue that we do not affect the pathwidth much.
	Observe that we only subdivided edges
	and attached graphs of constant pathwidth to vertices.
	Both of these steps do not alter the pathwidth by more than a constant. Thus, given a
	$\Ostar{(2-\epsilon)^\pw}$ algorithm for \countECover, we
	get a $\Ostar{(2-\delta)^\pw}$ algorithm for \countMIS
	on $r$-regular graphs.
	By \cref{lem:lower:count:hardnessCountMISregular},
	this contradicts \#SETH.
\end{proof}